\documentclass[11pt]{article}
\usepackage[left=1in,top=1in,right=1in,bottom=1in,head=0in,nofoot]{geometry}
\usepackage{amsmath,amssymb,amsthm}
\usepackage{enumerate}
\usepackage{graphicx}
\graphicspath{{./figs/}}
\usepackage{xcolor}
\usepackage{subcaption}
\usepackage{caption}
\usepackage{authblk}
\usepackage{float}
\usepackage[title]{appendix}
\usepackage{times}
\usepackage{bm}
\usepackage{natbib}

\usepackage[plain,noend]{algorithm2e}

\renewcommand{\emph}[1]{#1}
\newcommand{\com}[2]{\left(\begin{array}{c}#1 \\ #2\end{array}\right)}
\newcommand{\lb}{\left(}
\newcommand{\rb}{\right)}
\newcommand{\eps}{\epsilon}

\newcommand{\td}{\tilde}
\newcommand{\E}{\mathbb{E}}
\newcommand{\R}{\mathbb{R}}

\renewcommand{\L}{\mathcal{L}}
\renewcommand{\P}{\mathbb{P}}

\newcommand{\one}{\textbf{1}}

\newcommand{\med}{\mbox{median}}

\newcommand{\iid}{i.i.d.}
\newcommand{\CPT}{CPT}
\newcommand{\GA}{genetic algorithm}
\renewcommand{\SS}{stochastic search}
\newcommand{\ANOVA}{ANOVA}
\newcommand{\LAD}{LAD}

\DeclareMathOperator*{\diag}{\ensuremath{diag}}

\newcommand{\Ju}{Jureckova}

\newtheorem{theorem}{Theorem}

\newtheorem{proposition}{Proposition}

\newtheorem{remark}{Remark}

\begin{document}

\title{An Assumption-Free Exact Test For Fixed-Design Linear Models With Exchangeable Errors}

\author[1]{Lihua Lei}
\affil[1]{Departments of Statistics, Stanford University \thanks{lihualei@stanford.edu}}

\author[2]{Peter J. Bickel}
\affil[2]{Departments of Statistics, University of California, Berkeley \thanks{bickel@stat.berkeley.edu}}

\maketitle

\begin{abstract}
We propose the Cyclic Permutation Test (\CPT) to test general linear hypotheses for linear models. This test is non-randomized and valid in finite samples with exact Type I error $\alpha$ for an arbitrary fixed design matrix and arbitrary exchangeable errors, whenever $1 / \alpha$ is an integer and $n / p \ge 1 / \alpha - 1$. The test involves applying the marginal rank test to $1 / \alpha$ linear statistics of the outcome vector, where the coefficient vectors are determined by solving a linear system such that the joint distribution of the linear statistics is invariant with respect to a non-standard cyclic permutation group under the null hypothesis.The power can be further enhanced by solving a secondary non-linear travelling salesman problem, for
which the genetic algorithm can find a reasonably good solution. Extensive simulation studies show that the CPT has comparable power to existing tests. When testing for a single contrast of coefficients, an exact confidence interval can be obtained by inverting the test. Furthermore, we provide a selective yet extensive literature review of the century-long efforts on this problem, highlighting the novelty of our
test.
\end{abstract}

~

\section{Introduction}

In this article, we consider the following fixed-design linear model 
\begin{equation*}
  y_{i} = \beta_{0} + \sum_{j=1}^{p}x_{ij}\beta_{j} + \eps_{i}, \qquad i=1,\ldots, n,
\end{equation*}
where the $\eps_{i}\, (i = 1, \ldots, n)$ are stochastic errors and the $x_{ij}\, (i = 1, \ldots, n, j = 1, \ldots, p)$ are treated as fixed quantities. Throughout we will use the following compact notation 
\begin{equation}\label{eq:lm_matrix}
y = \beta_{0}\one + X\beta + \eps,
\end{equation}
where $y = (y_{1}, \ldots, y_{n})^{T}$ denotes the response vector, $X = (x_{ij})\in \R^{n\times p}$ denotes the design matrix, $\eps = (\eps_{1}, \ldots, \eps_{n})^{T}$ denotes the error terms and $\one\in\R^{n}$ denotes the vector with all entries equal to one. 
Our focus is on testing a general linear hypothesis:
\begin{equation}
  \label{eq:H0}
H_{0}: R^{T}\beta = 0, \quad \mbox{where }R\in\R^{p\times r}  \mbox{ is a fixed matrix with rank }r. 
\end{equation}
Testing linear hypotheses in linear models is ubiquitous and fundamental in numerous areas. One important example is to test whether a particular coefficient is zero, i.e. $H_{0}: \beta_{1} = 0$, a special case where $R = (1, 0, \ldots, 0)^{T}\in \R^{p\times 1}$. Another important example is to test the global null, i.e. $H_{0}: \beta = 0$, equivalent to the linear hypothesis with $R = I_{p\times p}$. We refer to Chapter 7 of \cite{lehmann2006testing} for an extensive discussion of other examples. By inverting a test with valid Type I error control, we can obtain a confidence interval/region for $R^{T}\beta$. This is of particular interest when $r = 1$, which corresponds to a single linear contrast of the regression coefficient.

Testing linear hypotheses in linear models is one of the most fundamental and long-lasting problems in statistics, as well as a convenient powerful prototype to motivate methods for more complicated statistical problems. In the past century, several types of methods have been proposed: normal theory-based tests \citep{fisher1922goodness, fisher1924036}, permutation tests \citep{pitman1937significance, pitman1938significance}, rank-based tests \citep{friedman1937use}, tests based on regression R-estimates \citep{hajek1962asymptotically}, M-estimates \citep{huber1973robust} and L-estimates \citep{bickel1973some}, resampling-based tests \citep{freedman1981bootstrapping}, median-based tests \citep{theil1950rank, brown1951median}, symmetry-based tests \citep{hartigan1970exact} and non-standard tests \citep{meinshausen2015group}. Here we list only the earliest reference we could track down for each category to highlight the chronology of the methodological development; an extensive literature review is provided in Appendix \ref{sec:epitome}.

For a given confidence level $1 - \alpha$, a test is exact if the Type I error is below or equal to $\alpha$, in finite samples without any asymptotics. Exact tests are intellectually and practically appealing because they provide strong error control without the requirement of a large sample or artificial asymptotic regimes. However, perhaps surprisingly, there is no test that is exact under reasonably general assumptions to the best of our knowledge. A brief summary of the conditions under which the existing tests are exact is as follows:
  \begin{itemize}
  \item Regression t- and F-tests are exact with normal errors;
  \item Permutation tests are exact for the global null or certain null hypotheses for certain analysis of variance (\ANOVA) problems \citep[e.g.][]{brown1982distribution};
  \item Rank-based tests are exact for \ANOVA ~problems;
  \item Tests based on regression R-, M- or L-estimates can be made exact for the global null;
  \item \cite{hartigan1970exact}'s test is exact for certain forms of balanced \ANOVA ~problems with symmetric errors and $r = 1$;
  \item \cite{meinshausen2015group}'s test is exact for rotationally invariant errors with known noise level, and if the $\eps_{i}$s are independent and identically distributed (\iid), rotation invariance implies the normality of the $\eps_{i}$ \citep{maxwell1860v};
  \item Other tests are exact either for the global null or under restrictive assumptions or require excessive computation.
  \end{itemize}

In this article we develop an exact test, which we refer to as the \emph{Cyclic Permutation Test} (\CPT), that is valid in finite samples, and can accommodate an arbitrary fixed design matrix and arbitrary error distributions, provided that the error terms are exchangeable. Exchangeability is weaker than the frequently made assumption of \iid ~random variables. Further, the test is non-randomized if $1 / \alpha$ is an integer and $n / (p - r) > 1 / \alpha - 1$. The former condition is true for all common choices of $\alpha$, e.g. $0.1, 0.05, 0.01, 0.005$. The latter requirement is also reasonable in various applications. For instance, when $\alpha = 0.05$, the condition reads $n / (p - r) > 19$, which is true if $n / p > 19$ or $p - r$ is small; both are typical in social science applications. Admittely, it may be stringent in areas like genetics where $p$ is often larger than $n$. However, valid inference, or even identification, in those problems would require extra assumptions on the sparsity of $\beta$, geometry of $X$, and distribution of $\eps$, which are not  in accordance with the goal of this paper to develop assumption-free tests. We demonstrate the power of the \CPT ~through extensive simulation studies and show it is comparable to the existing ones. Although exchangeability may not be valid in certain applications, the \CPT ~is the first procedure that is provably exact with reasonable power under such weak assumptions. We want to emphasize that the goal of this paper is not to propose a procedure that is superior to existing tests, but rather to expand the toolbox of exact inference and, hopefully, motivate the development of novel methods for other problems.


\section{Cyclic Permutation Test}\label{sec:CPT}

\subsection{Main idea}

Throughout the article we denote the set $\{1, \ldots, n\}$ by $[n]$. First we show that it is sufficient to consider the sub-hypothesis:
\begin{equation}
  \label{eq:partial_null}
  H_{0}: \beta_{1} = \ldots = \beta_{r} = 0.
\end{equation}
In fact
, let $U_{R}\in \R^{p\times r}$ be an orthonormal basis of the column span of $R$ and $V_{R}\in \R^{p\times (p - r)}$ be an orthonormal basis of the orthogonal complement. Then $\beta = U_{R}U_{R}^{T}\beta + V_{R}V_{R}^{T}\beta$. Let $\td{X} = (XU_{R} \,\vdots \, XV_{R})$, where $\vdots$ marks the partition of columns, and $\td{\beta}^{T} = (\beta^{T}U_{R}, \beta^{T}V_{R})$. Then the linear model \eqref{eq:lm_matrix} can be re-formulated as 
\begin{equation*}
  y = \beta_{0}\one + XU_{R}(U_{R}^{T}\beta) + XV_{R}(V_{R}^{T}\beta) + \eps = \beta_{0}\one + \sum_{j=1}^{r}\td{X}_{j}\td{\beta}_{j} + \sum_{j=r+1}^{p}\td{X}_{j}\td{\beta}_{j} + \eps.
\end{equation*}
On the other hand, since $R$ has full column rank, the null hypothesis \eqref{eq:H0} is equivalent to $H_{0}: \td{\beta}_{1} = \ldots = \td{\beta}_{r} = 0$, which is typically referred to as a sub-hypothesis \citep[e.g.][]{adichie1978rank}. For this reason, we will focus on \eqref{eq:partial_null} without loss of generality throughout the rest of the paper. 

Our idea is to construct a pool of linear statistics $S = (S_{0}, S_{1}, \ldots, S_{m})$ such that $S$ is distributionally invariant under the \emph{left shifting operator} $\pi_{L}$ under the null, in the sense that
\begin{equation}
  \label{eq:CPG}
  S \stackrel{d}{=} \pi_{L}(S) \stackrel{d}{=} \pi_{L}^{2}(S) \stackrel{d}{=} \cdots \stackrel{d}{=} \pi_{L}^{m}(S),
\end{equation}
\begin{equation}
  \label{eq:leftshifting}
 \mbox{where} \quad   \pi_{L}^{k}(S) = (S_{k}, S_{k+1}, \ldots, S_{m}, S_{0}, S_{1}, \ldots, S_{k-1}), \quad k = 1, 2, \ldots, m.
\end{equation}
Let $\mathrm{Id}$ denote the identity mapping. Then $\mathcal{G} = \{\mathrm{Id}, \pi_{L}, \ldots, \pi_{L}^{m}\}$ forms a group, which we refer to as the \emph{cyclic permutation group}. We say a pool of statistics $S$ is invariant under the cyclic permutation group if $S$ satisfies \eqref{eq:CPG}. The following proposition describes the main property of statistics that are invariant under the cyclic permutation group.
  \begin{proposition}\label{prop:MRT}
    Assume that $S = (S_{0}, S_{1}, \ldots, S_{m})$ is invariant under the cyclic permutation group. Let $R_{0}$ be the rank of $S_{0}$ in descending order, i.e. $R_{0} = |\{j\ge 0: S_{j} \ge S_{0}\}|$. Then 
    \begin{equation}\label{eq:validity_pvalue}
      \L(R_{0})\succeq \mathrm{Unif}([m+1]) \Longrightarrow \text{ if }p \triangleq \frac{R_{0}}{m + 1}, \,\, \L(p)\succeq \mathrm{Unif}([0, 1])
    \end{equation}
where $\L$ denotes law, $\succeq$ denotes stochastic dominance, and $\mathrm{Unif}([0, 1])$ denotes the uniform distribution on $[0, 1]$. Furthermore, $R_{0}\sim \mathrm{Unif}([m + 1])$ if $S$ has no ties with probability $1$.
\end{proposition}

Based on the p-value defined in \eqref{eq:validity_pvalue}, we can derive a test that rejects the null hypothesis if $p\le \alpha$. We refer to this simple test as a \emph{marginal rank test}. The following proposition shows that the marginal rank test is valid in finite samples and can be exact under mild conditions.
\begin{proposition}\label{prop:exact_MRT}
  Suppose $S = (S_{0}, S_{1}, \ldots, S_{m})$ is invariant under the cyclic permutation group under $H_{0}$ and let the p-value be defined as in \eqref{eq:validity_pvalue}. Then $\P_{H_{0}}(p\le \alpha)\le \alpha$. If $1 / \alpha$ is an integer, $m + 1$ is divisible by $1 / \alpha$, and $S$ has no ties almost surely, then $\P_{H_{0}}(p\le \alpha)= \alpha$.
\end{proposition}

In practice, the reciprocals of commonly-used confidence levels (e.g. $0.1, 0.05, 0.01, 0.005$) are integers. In these cases it is sufficient to set $m = 1 / \alpha - 1$ to obtain an exact test. 

The rank used in the marginal rank test only gives one-sided information and may not be suitable for two-sided tests. More concretely, $S_{0}$ may be significantly different from $S_{1}, \ldots, S_{m}$ under the alternative but the sign of the difference may depend on the true parameters. An intuitive remedy is to apply the marginal rank test on the following modified statistics
\begin{equation}\label{eq:tdSj}
  \td{S}_{j} = |S_{j} - \med\{(S_{j})_{j=0}^{m}\}|.
\end{equation}
Whenever $S_{0}$ is significantly different from $S_{1}, \ldots, S_{m}$, $\td{S}_{0}$ is significantly larger than $\td{S}_{1}, \ldots, \td{S}_{m}$. The following proposition guarantees the validity of the transformation \eqref{eq:tdSj}. In particular, the transformation in \eqref{eq:tdSj} satisfies the condition.
\begin{proposition}\label{prop:transform_S}
  If $S = (S_{0}, S_{1}, \ldots, S_{m})$ is invariant under the cyclic permutation group,
\[\td{S} = \{g(S_{0}; S), g(S_{1}; S), \ldots, g(S_{m}; S)\}\] is invariant under the cyclic permutation group for every $g$ such that 
\[g(z; w) = g(z; \pi_{L}w).\]
\end{proposition}

In this article, we consider linear statistics 
\[S_{j} = y^{T}\eta_{j}, \quad j = 0, 1, \ldots, m,\]
and apply the marginal rank test on $\td{S}_{0}, \ldots, \td{S}_{m}$ defined in \eqref{eq:tdSj}. Partition $X$ into $(X_{[r]}\,\, X_{[-r]})$ and $\beta$ into $(\beta_{[r]}, \beta_{[-r]})$. The linear model \eqref{eq:lm_matrix} implies that
\begin{equation}\label{eq:etajy}
y^{T}\eta_{j} = (\one^{T}\eta_{j})\beta_{0} + (X_{[r]}^{T}\eta_{j})^{T}\beta_{[r]} + (X_{[-r]}^{T}\eta_{j})^{T}\beta_{[-r]} + \eps^{T}\eta_{j}.
\end{equation}
In the next three subsections we will show how to construct the $\eta_{j}$s to guarantee the Type I error control and to enhance power. Surprisingly, the only distributional assumption on $\eps$ is exchangeability:
\begin{enumerate}[{A}1]
\item The error vector $\eps$ has exchangeable components, i.e. for any permutation $\pi$ on $[n]$,
\[(\eps_{1}, \ldots, \eps_{n})\stackrel{d}{=}(\eps_{\pi(1)}, \ldots, \eps_{\pi(n)}).\]
\end{enumerate}

\subsection{Construction for Type I Error Control}

Under $H_{0}$, \eqref{eq:etajy} can be simplified as
\begin{equation}\label{eq:etajy_H0}
y^{T}\eta_{j} = \underbrace{(\one^{T}\eta_{j})\beta_{0} + (X_{[-r]}^{T}\eta_{j})^{T}\beta_{[-r]}}_{\text{deterministic part}} + \underbrace{\eps^{T}\eta_{j}}_{\text{stochastic part}}.
\end{equation}
To ensure the distributional invariance of $\{y^{T}\eta_{0}, \ldots, y^{T}\eta_{m}\}$ under the cyclic permutation group, it is sufficient to construct $\eta_{j}$s such that the deterministic parts are identical for all $j$ and the stochastic parts are invariant under the cyclic permutation group. To match the deterministic parts, we can simply set $X_{[-r]}^{T}\eta_{j}$ to be equal, as stated in the following condition.
\begin{enumerate}[{C}1]
\item There exists $\gamma_{[-r]}\in \R^{p-r}$ such that
\[X_{[-r]}^{T}\eta_{j} = \gamma_{[-r]},\,\,\,  (j = 0, 1, \ldots, m) .\]
\end{enumerate}
To ensure the invariance of the stochastic parts, intuitively the $\eta_{j}$s should be left shifted transforms of each other. To be concrete, consider the case where $n = 6$ and $m = 2$. Then given any $\eta^{*} = (\eta^{*}_{1}, \eta^{*}_{2}, \eta^{*}_{3}, \eta^{*}_{4}, \eta^{*}_{5}, \eta^{*}_{6})^{T}$, the following construction would imply the invariance of $\{\eps^{T}\eta_{j}: j = 0, 1, \ldots, m\}$ under the cyclic permutation group:
\[\eta_{0} = (\eta^{*}_{1}, \eta^{*}_{2}, \eta^{*}_{3}, \eta^{*}_{4}, \eta^{*}_{5}, \eta^{*}_{6})^{T}, \quad \eta_{1} = (\eta^{*}_{3}, \eta^{*}_{4}, \eta^{*}_{5}, \eta^{*}_{6}, \eta^{*}_{1}, \eta^{*}_{2})^{T}, \quad  \eta_{2} = (\eta^{*}_{5}, \eta^{*}_{6}, \eta^{*}_{1}, \eta^{*}_{2}, \eta^{*}_{3}, \eta^{*}_{4})^{T}.\]
To see this, note that
  \begin{align*}
    &(\eps^{T}\eta_{0}, \eps^{T}\eta_{1}, \eps^{T}\eta_{2})^{T}  = \lb \begin{array}{cccccc}
      \eps_{1} & \eps_{2} & \eps_{3} & \eps_{4} & \eps_{5} & \eps_{6}\\
\eps_{5} & \eps_{6} & \eps_{1} & \eps_{2} & \eps_{3} & \eps_{4}\\
     \eps_{3} & \eps_{4} & \eps_{5} & \eps_{6} & \eps_{1} & \eps_{2}
    \end{array}
\rb \eta^{*}, \\
    &(\eps^{T}\eta_{1}, \eps^{T}\eta_{2}, \eps^{T}\eta_{0})^{T}  = \lb \begin{array}{cccccc}
\eps_{5} & \eps_{6} & \eps_{1} & \eps_{2} & \eps_{3} & \eps_{4}\\
     \eps_{3} & \eps_{4} & \eps_{5} & \eps_{6} & \eps_{1} & \eps_{2}\\
      \eps_{1} & \eps_{2} & \eps_{3} & \eps_{4} & \eps_{5} & \eps_{6}
    \end{array}
\rb \eta^{*}.
  \end{align*}
By assumption A1, 
\[\lb \begin{array}{cccccc}
      \eps_{1} & \eps_{2} & \eps_{3} & \eps_{4} & \eps_{5} & \eps_{6}\\
\eps_{5} & \eps_{6} & \eps_{1} & \eps_{2} & \eps_{3} & \eps_{4}\\
     \eps_{3} & \eps_{4} & \eps_{5} & \eps_{6} & \eps_{1} & \eps_{2}
    \end{array}
\rb\stackrel{d}{=}\lb \begin{array}{cccccc}
\eps_{5} & \eps_{6} & \eps_{1} & \eps_{2} & \eps_{3} & \eps_{4}\\
     \eps_{3} & \eps_{4} & \eps_{5} & \eps_{6} & \eps_{1} & \eps_{2}\\
      \eps_{1} & \eps_{2} & \eps_{3} & \eps_{4} & \eps_{5} & \eps_{6}
    \end{array}
    \rb\]
  \[\Longrightarrow (\eps^{T}\eta_{0}, \eps^{T}\eta_{1}, \eps^{T}\eta_{2})\stackrel{d}{=}(\eps^{T}\eta_{1}, \eps^{T}\eta_{2}, \eps^{T}\eta_{0}).\]
Using the same argument we can show $(\eps^{T}\eta_{0}, \eps^{T}\eta_{1}, \eps^{T}\eta_{2})\stackrel{d}{=}(\eps^{T}\eta_{2}, \eps^{T}\eta_{0}, \eps^{T}\eta_{1})$ and thus the invariance of $(\eps^{T}\eta_{0}, \eps^{T}\eta_{1}, \eps^{T}\eta_{2})$ under the cyclic permutation group. 

In general, if $n$ is divisible by $m + 1$ with $n = (m + 1)t$, then we can construct $\eta_{j}$ as a left shifted transform of a vector $\eta^{*}$, i.e.
\begin{equation}
  \label{eq:etaj_divisible}
  \eta_{j} = \pi_{L}^{tj}(\eta^{*})
\end{equation}
where $\pi_{L}$ is the left shifting operator defined in \eqref{eq:leftshifting}. More generally, if $n = (m + 1)t + s$ for some integers $t$ and $0\le s \le m$, we can leave the last $s$ components to be the same across the $\eta_{j}$s while shifting the first $(m + 1)t$ entries as in \eqref{eq:etaj_divisible}, as stated in the following condition. 
\begin{enumerate}[C2]
\item There exists $\eta_{*}\in \R^{n}$ such that
\[\eta_{j} = \left[\pi_{L}^{tj}\{(\eta_{1}^{*}, \ldots, \eta_{(m + 1)t}^{*})\}, \eta_{(m+1)t + 1}^{*}, \ldots, \eta_{n}^{*}\right]^{T},\]
where $t = \lfloor n / (m + 1)\rfloor$.
\end{enumerate}
\begin{proposition}\label{prop:typeI}
Under assumption A1, $(y^{T}\eta_{0}, \ldots, y^{T}\eta_{m})$ is distributionally invariant under the cyclic permutation group if $(\eta_{0}, \ldots, \eta_{m})$ satisfies C1 and C2.
\end{proposition}

Now we discuss the existence of $(\eta_{*}, \gamma_{[-r]})$. Note that $\eta_{j}$ is a linear transformation of $\eta^{*}$. Let $I_{p-r}$ denote the identity matrix of size $p-r$ and $\Pi_{j}\in \R^{n\times n}$ be the matrix such that $\eta_{j} = \Pi_{j}\eta^{*}$. Then C1 and C2 imply that
\begin{equation}
  \label{eq:validity_equation}
  \lb
  \begin{array}{ll}
    -I_{p-r} & X_{[-r]}^{T}\\
    -I_{p-r} & X_{[-r]}^{T}\Pi_{1}\\
    \vdots & \vdots\\
    -I_{p-r} & X_{[-r]}^{T}\Pi_{m}
  \end{array}
\rb \com{\gamma_{[-r]}}{\eta_{*}} = 0.
\end{equation}
The above linear system has $(m + 1)(p - r)$ equations and $n + p - r$ unknowns. Therefore, a non-zero solution always exists if $(m + 1)(p - r) < n + p - r$.
\begin{theorem}\label{thm:typeI}
Under assumption A1,
  \begin{enumerate}[(a)]
  \item The linear system \eqref{eq:validity_equation} always has a non-zero solution if 
    \begin{equation}
      \label{eq:np_cond_typeI}
      n / (p - r) > m.
    \end{equation}
  \item for any solution $(\gamma_{[-r]}, \eta^{*})$ of \eqref{eq:validity_equation}, 
\[(y^{T}\eta^{*}, y^{T}\Pi_{1}\eta^{*}, \cdots y^{T}\Pi_{m}\eta^{*})\]
is invariant under the cyclic permutation group under $H_{0}$, where $\Pi_{j}\in \R^{n\times n}$ is the coefficient matrix that maps $\eta^{*}$ to $\eta_{j}$ defined in C2.
  \end{enumerate}
\end{theorem}
Suppose $\alpha = 0.05$ for illustration and set $m = 1 / \alpha - 1 = 19$. Then the condition \eqref{eq:np_cond_typeI} reads 
\[n > 19(p - r).\]
Even when $r = 1$, this is satisfied in many applications. On the other hand, when $r$ is large but $p - r$ is small, then \eqref{eq:np_cond_typeI} can still be satisfied even if $p > n$. This is in sharp contrast to regression F-tests and permutation F-tests that require fitting the full model and thus $p \le n$. Furthermore, we emphasize that Theorem \ref{thm:typeI} allows arbitrary design matrices. This is fundamentally different from the asymptotically valid tests which always impose regularity conditions on $X$.

\subsection{Construction for high power when $r = 1$}

To guarantee reasonable power, we need $y^{T}\eta_{0}$ to be significantly different from the other statistics under the alternative. In this subsection we focus on the case where $r = 1$ to highlight the key idea. The general case with $r > 1$ is discussed in Appendix \ref{app:r>1}. 

When $\beta_{1} \not = 0$, \eqref{eq:etajy} implies that 
\[y^{T}\eta_{j} = (X_{1}^{T}\eta_{j})\beta_{1} + W_{j}\]
where $W_{j} = \eps^{T}\eta_{j} + (\one^{T}\eta_{*})\beta_{0} + (X_{[-1]}^{T}\eta_{*})^{T}\beta_{[-1]}$ and $(W_{1}, \ldots, W_{m})$ is invariant under the cyclic permutation group by Theorem \ref{thm:typeI}. To enhance power, it is desirable that $X_{1}^{T}\eta_{0}$ lies far from $\{X_{1}^{T}\eta_{1}, \ldots, X_{1}^{T}\eta_{m}\}$. In particular, we impose the following condition on the $\eta_{j}$s:
\begin{enumerate}[{C}1]
\setcounter{enumi}{2}
\item there exists $\gamma_{1}, \delta\in R$, such that 
\[X_{1}^{T}\eta_{j} = \gamma_{1}\,\,\, (j = 1, 2, \ldots, m), \quad X_{1}^{T}\eta_{0} = \gamma_{1} + \delta.\]
\end{enumerate}
Putting C1, C2 and C3 together, we obtain the following linear system,
\begin{equation}
  \label{eq:power_equation}
  \bigg(-e_{1, p(m+1)}\,\,\vdots\,\, A(X)^{T}\bigg)\lb \begin{array}{l}
\delta\\
\gamma\\
\eta
\end{array}\rb = 0,
\end{equation}
where $e_{1, p(m + 1)}$ is the first canonical basis in $\R^{p(m + 1)}$ and 
\begin{equation}
  \label{eq:AX}
  A(X) = \lb
  \begin{array}{cccc}
    -I_{p}& -I_{p}& \cdots & -I_{p}\\
    X & \Pi_{1}^{T}X & \cdots & \Pi_{m}^{T}X
  \end{array}
\rb\in \R^{(n + p)\times p(m+1)}.
\end{equation}
This linear system has $(m + 1)p$ equations and $n + p + 1$ variables. Thus it always has a non-zero solution if 
\[n + p + 1 > p(m + 1)\Longleftrightarrow n \ge pm.\]
When $\alpha = 0.05$ and $m = 19$, this condition is still reasonable in many problems. 

The normalized gap $\delta / \|\eta\|$ can be regarded as a proxy of power. Write $\gamma$ for $\com{\gamma_{1}}{\gamma_{[-1]}}$. It is natural to consider the following optimization:
\begin{align}
  &\max_{\delta\in \R, \gamma\in \R^{p}, \eta\in \R^{n}, \|\eta\|_{2} = 1}\,\, \delta, \quad   \mbox{s.t. }  \bigg(-e_{1, p(m+1)}\,\,\vdots\,\, A(X)^{T}\bigg)\lb \begin{array}{l}
\delta\\
\gamma\\
\eta
\end{array}\rb = 0.\label{eq:LP}
\end{align}
This linear programming problem can be solved by fitting a linear regression and it permits a closed-form solution. Let $O^{*}(X)$ denote the optimal value of the objective function, i.e. the maximum achievable value of $\delta$ in this case. Here we use the symbol $O^{*}(X)$ instead of $\delta(X)$ to distinguish the role of the objective value and the variable $\delta$. In spite of these coinciding when $r = 1$, they are distinct when $r > 1$; see Appendix \ref{app:r>1} for details.
\begin{theorem}\label{thm:deltastar}
Assume that $n \ge pm$. Let 
\begin{equation}\label{eq:BX}
  B(X) = \lb
  \begin{array}{cccc}
    (I - \Pi_{m})^{T}X & (\Pi_{1} - \Pi_{m})^{T}X & \cdots & (\Pi_{m-1} - \Pi_{m})^{T}X
  \end{array}
\rb\in \R^{n\times mp}.
\end{equation}
Partition $B(X)$ into $[B(X)_{1} \,\, B(X)_{[-1]}]$ where $B(X)_{1}$ is the first column of $B(X)$. Further let 
\[\td{\eta} = (I - H_{[-1]})B(X)_{1}, \quad \mbox{ where }H_{[-1]} = B(X)_{[-1]}(B(X)_{[-1]}^{T}B(X)_{[-1]})^{+}B(X)_{[-1]}^{T}\] 
where $+$ denotes the Moore-Penrose generalized inverse. Then $O^{*}(X) = \|\td{\eta}\|_{2}$ and one global maximizer of \eqref{eq:LP} is given by
\[\eta^{*}(X) = \td{\eta} / \|\td{\eta}\|_{2}, \quad \delta^{*}(X) = \|\td{\eta}\|_{2}.\]
\end{theorem}
\begin{remark}
  When $B(X)_{[-1]}$ has full column rank, $\td{\eta}$ is the residual vector obtained by regressing $B(X)_{1}$ on $B(X)_{[-1]}$ and $\|\td{\eta}\|_{2}^{2}$ is the residual sum of squares. Both quantities can be easily computed using standard software. If $B(X)_{[-1]}$ does not have full column rank, then $\td{\eta}$ is the residual from minimum-norm least squares solution obtained by regressing $B(X)_{1}$ on $B(X)_{[-1]}$, which is the limit of the ridge estimator with the penalty level tending to zero and is the limiting solution of standard gradient descent initialized at zero \citep[e.g.][]{hastie2019surprises}.
\end{remark}

\subsection{Pre-ordering rows of design Matrix}

Given any $X$, we can easily calculate the proxy of signal strength $O^{*}(X)$ by Theorem \ref{thm:deltastar}. However, the optimal value is not invariant to row permutations of $X$. That is, for any permutation matrix $\Pi\in \R^{n\times n}$, $O^{*}(X) \not= O^{*}(\Pi X)$ typically. Roughly speaking, this is because $\delta^{*}(X)$ involves the left shifting operator, which depends on the arrangement of the rows of $X$. Figure \ref{fig:illustrate} illustrates the variability of $O^{*}(\Pi X)$ as a function of $\Pi$ for a fixed matrix with $8$ rows and $3$ columns, generated with \iid ~Gaussian entries.

\begin{figure}
  \centering
  \begin{subfigure}[t]{0.36\textwidth}
  \includegraphics[width = 1\textwidth]{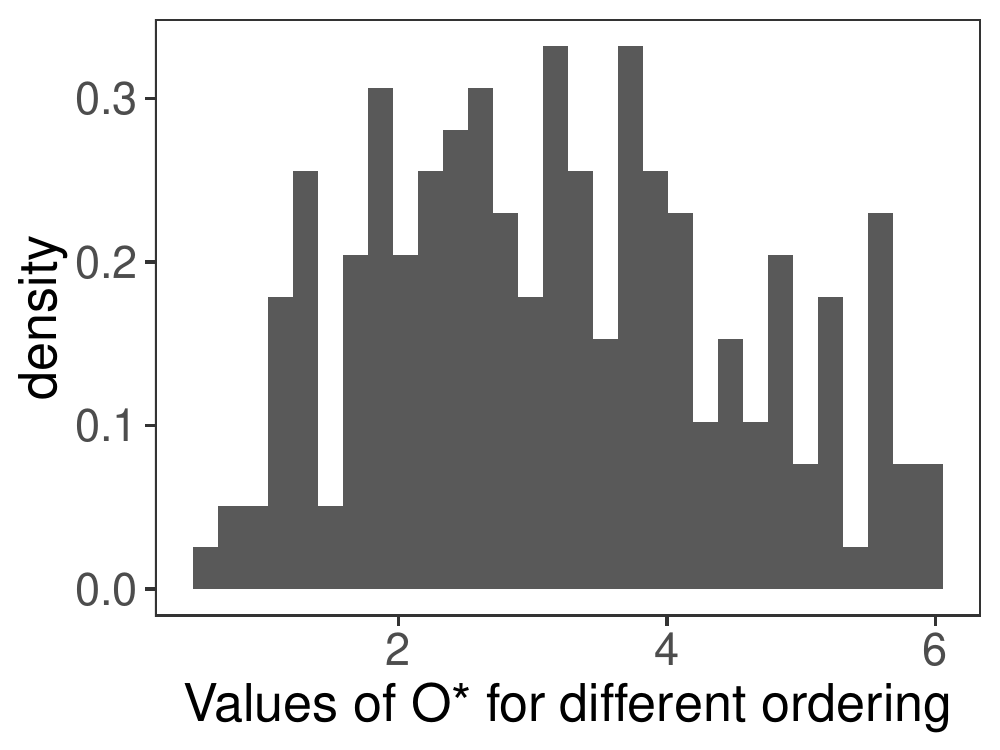}
  \caption{}\label{fig:illustrate}
  \end{subfigure}
  \begin{subfigure}[t]{0.63\textwidth}
  \includegraphics[width = 1\textwidth]{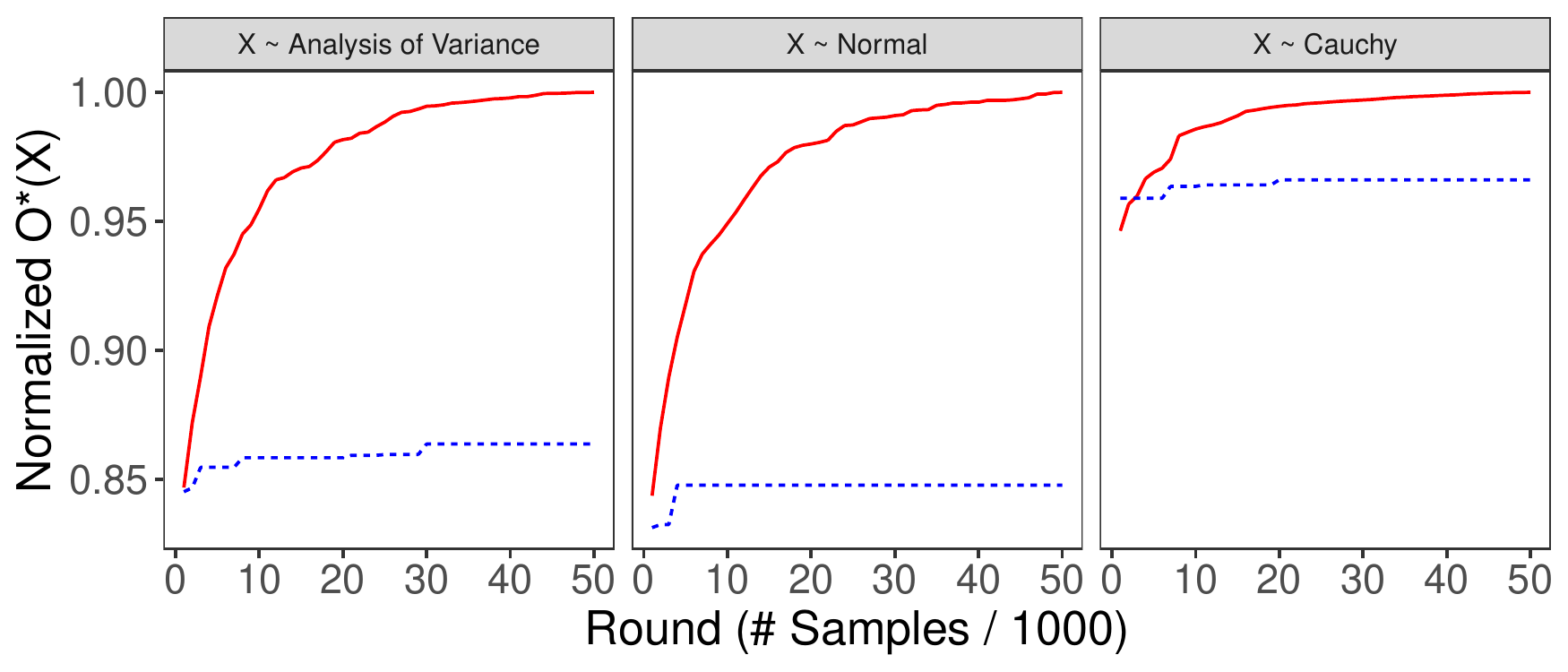}
  \caption{}\label{fig:GA_SS}
  \end{subfigure}
\caption{(a) Histograms of $O^{*}(\Pi X)$ for a realization of a random matrix with Gaussian entries; (b) Comparisons of the \GA ~(red solid line) and \SS ~(blue dotted line) for three matrices as realizations of random one-way \ANOVA ~matrices with one entry in each row at a uniformly random position (left), random matrices with standard normal entries (middle), and random matrices with standard Cauchy entries (right).}
\end{figure}
Notably, even in such regular cases variability is non-negligible. This motivates the following secondary combinatorial optimization problem:
\begin{equation}
  \label{eq:secondary_optimization}
  \max_{\Pi}O^{*}(\Pi X).
\end{equation}
This is a non-linear travelling salesman problem. Note that we aim at finding a solution with a reasonably large objective value instead of finding the global maximum of \eqref{eq:secondary_optimization}, which is NP-hard. For this reason, we solve \eqref{eq:secondary_optimization} by the \GA, which is generally efficient for moderate $n$ albeit without a worst-case convergence guarantee. In a nutshell, a \GA ~maintains a \emph{population} of permutations, generates new permutations by two operations: \emph{crossover} and \emph{mutation}, and \emph{evolves} the population via a mechanism called \emph{selection}, based on the objective value. We refer the readers to \cite{michalewicz2013genetic} for more details. 

We compare the \GA, implemented in \texttt{R} package \texttt{gaoptim}, with a simple competing algorithm that randomly selects ordering and keeps the one yielding the largest objective value. We refer to this method as \SS. Although this competitor is arguably too weak and more efficient algorithms may exist, our goal here is simply to illustrate the effectiveness of the \GA ~instead of to claim the superiority of the \GA. We compare the performance of the \GA ~and \SS ~on three matrices with $n = 1000$ and $p = 20$ as realizations generated from random one-way \ANOVA ~matrices with exactly one entry in each row at a unifromly random position, random matrices with \iid ~standard normal entries and random matrices with \iid ~standard Cauchy entries. The results are plotted in Figure \ref{fig:GA_SS} where the y-axis measures $O^{*}(\Pi X)$, scaled by the maximum achieved by the \GA ~and \SS ~for visualization, and the x-axis measures the number of random samples each algorithm accesses. The population size is set to be $10$ for the \GA ~in all scenarios. It is clear that the \GA ~consistently improves the solution while \SS ~gets trapped after a few iterations.

\subsection{Implementation of the \CPT}\label{subsec:CPT_summary}

We summarize the implementation of the \CPT ~below:
\begin{enumerate}[Step 1]
\item Compute a desirable pre-ordering $\Pi_{0}$ for the combinatorial optimization problem 
\[\max_{\Pi}O^{*}(\Pi X),\]
where $O^{*}(\cdot)$ is defined in Theorem \ref{thm:deltastar} when $r = 1$ or in Theorem \ref{thm:deltastar_r} when $r > 1$;
\item Replace $y$ and $X$ by $\Pi_{0}y$ and $\Pi_{0}X$;
\item Compute $\eta^{*}$ via Theorem \ref{thm:deltastar} when $r = 1$ or via Theorem \ref{thm:deltastar_r} when $r > 1$;
\item Compute $S_{j} = y^{T}\eta_{j}$ for $j = 0, 1, \ldots, m$ where 
\[\eta_{j}^{T} = \big[\pi_{L}^{tj}\{(\eta_{1}^{*}, \ldots, \eta_{(m + 1)t}^{*})\}, \eta_{(m+1)t + 1}^{*}, \ldots, \eta_{n}^{*}\big], \quad t = \lfloor n / (m + 1)\rfloor;\]
\item Compute $\td{S}_{j} = \big|S_{j} - \med\{(S_{j})_{j=0}^{m}\}\big|$; 
\item Compute the p-value $p = R_{0} / (m + 1)$ where $R_{0}$ is the rank of $\td{S}_{0}$ in the set $\{\td{S}_{0}, \td{S}_{1}, \ldots, \td{S}_{m}\}$ in descending order;
\item Reject the null hypothesis if $p\le \alpha$.
\end{enumerate}

The computational cost of Step 3 is the same as solving a linear regression with the sample size $n$ and dimension $pm$, as indicated by Theorem \ref{thm:deltastar}. As a result, the computational cost of Step 2 and Step 4-7 are negligible. If the computing budget is tight, a random ordering ~can be used for Step 1, for which the computational cost is negligible. Otherwise, a \GA ~can be used instead, of which the computational cost is the same as solving $M$ linear regressions of the same size as in Step 3, where $M$ is the total number of samples in the solution path. Admittedly, the latter option is computationally intensive compared to regression t- or F-tests and permutation tests -- the former involves solving a single linear regression with a smaller dimension $p$ and the latter involves solving $M$ linear regressions of the same size, where $M$ is the number of permutations. However, for moderate-sized problems, the computational time of our method is acceptable. On the other hand, if the \GA ~is replaced by a more efficient search algortihm, the computational cost can be drastically reduced. We discuss one potential algorithm in Section \ref{subsec:discussion_algorithm}.

\section{Experiments}\label{sec:experiments}

To assess the power of our procedure, we conduct extensive numerical experiments. In all the experiments we fix the sample size $n = 1000$ and consider three values $25, 33, 40$ for dimension $p$ such that the sample per parameter $n / p \approx 40, 30, 25$. Given a value of $p$, we consider the three types of design matrices considered in Figure \ref{fig:GA_SS}
. For each type of design matrices, we generate $50$ independent copies. Given each $X$, we generate 3000 copies of $\eps$ with independent entries from the standard normal distribution and standard Cauchy distribution. 

We consider two variants of the \CPT, one with random ordering and one with pre-ordering by the \GA, as well as five competing tests: (i) the t- or F-test; (ii) the permutation t- or F-test which approximates the null distribution of the t- or F-statistic by the permutation distribution with $X_{[r]}$ reshuffled; (iii) the Freedman-Lane test \citep{freedman83, anderson2001permutation, toulis2019life} which approximates the null distribution of the t- or F-statistic by the permutation distribution with reduced-form regression residuals reshuffled; (iv) the asymptotic z-test for least absolute deviation (\LAD) regression; (v) the GroupBound method \citep{meinshausen2015group}. For methods (ii) and (iii), we calculate the test based on 1000 random permutations. To further demonstrate the importance of the pre-ordering step in the \CPT, we consider a weaker pre-ordering with 1000 random samples and a stronger pre-ordering with 10000 random samples for the \GA. 
All tests will be performed with level $\alpha = 0.05$ and the number of statistics $m + 1$ is set to be $20$ for the \CPT. All programs to replicate the results in this article can be found in $\texttt{https://github.com/lihualei71/CPT}$. 

\begin{figure}
  \centering
  \includegraphics[width = 0.8\textwidth]{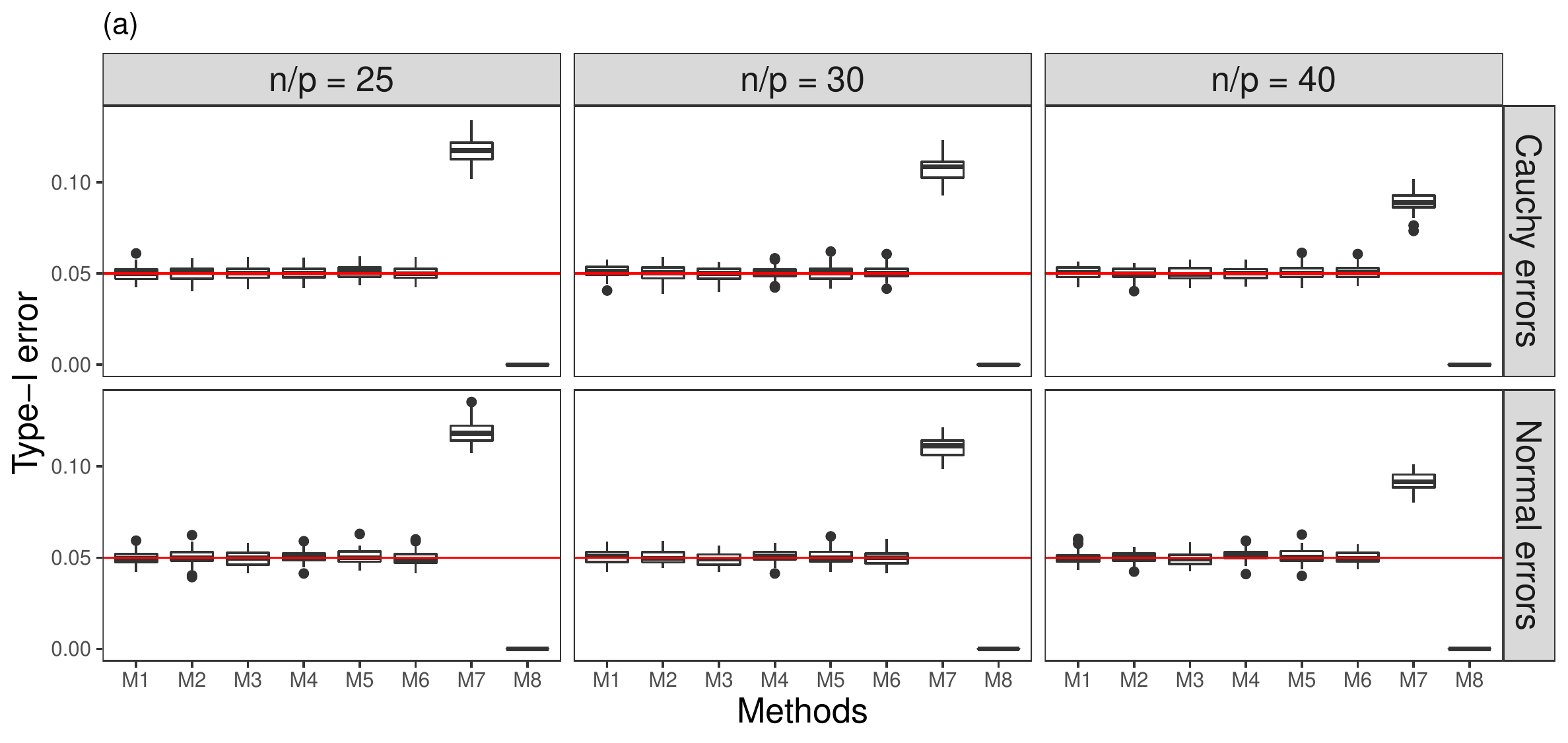}
  \includegraphics[width = 0.8\textwidth]{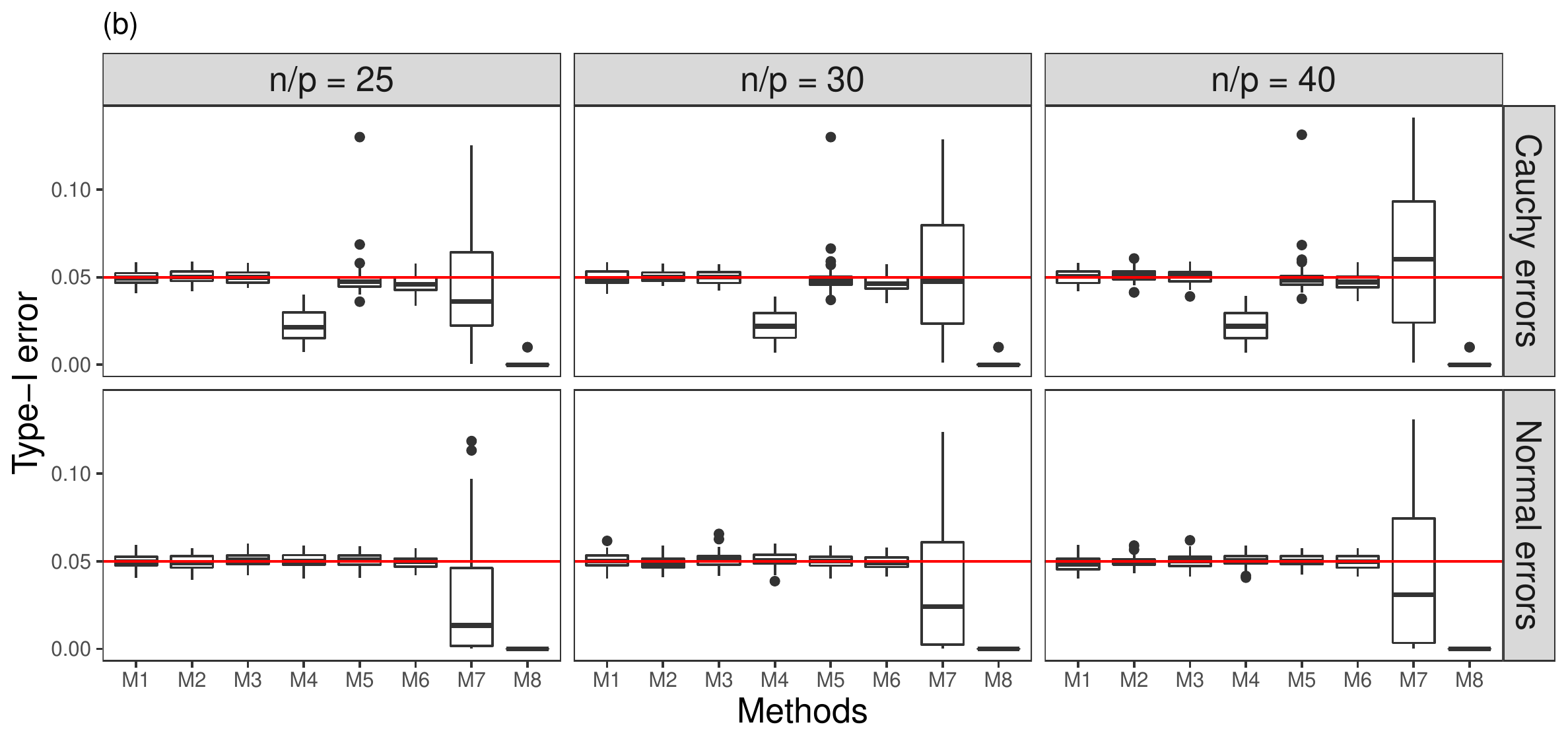}
  \includegraphics[width = 0.8\textwidth]{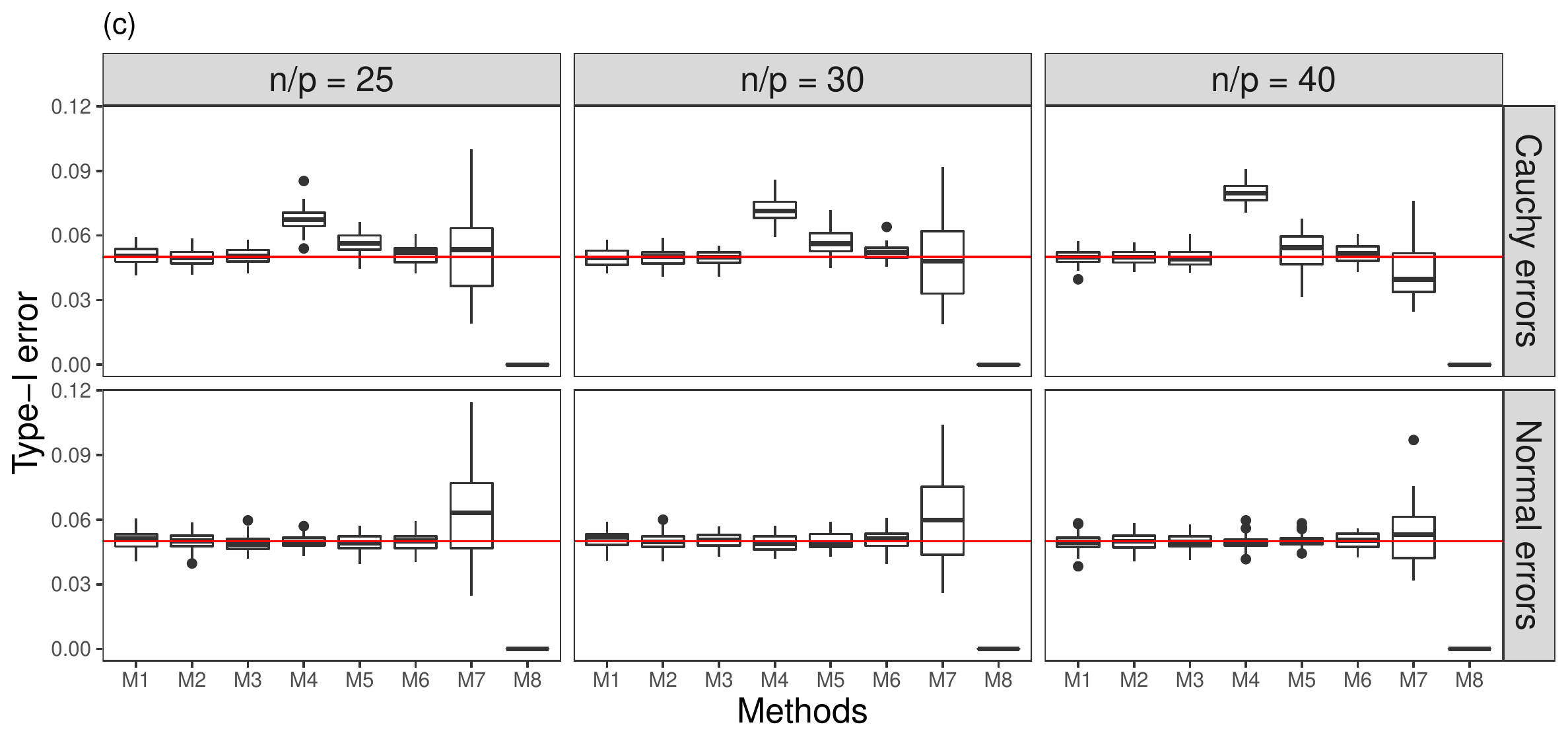}    
\caption{Monte-Carlo Type I error for testing a single coordinate with three types of $X$’s which are realizations of (a) random matrices with standard normal entries; (b) random matrices with standard Cauchy
entries; (c) random one-way \ANOVA ~design matrices. Eight methods are compared: M1, \CPT ~with
stronger ordering via the Genetic Algorithm; M2, \CPT ~with weaker ordering via the Genetic Algorithm; M3,
\CPT ~with random ordering via the Stochastic Search; M4, t- or F-test; M5, permutation test; M6, Freedman-Lane test; M7, test based on LAD; M8, GroupBound.}\label{fig:size_r1}
\end{figure}

Owing to the space constraint, here we only present the results for testing a single coordinate, i.e. $H_{0}: \beta_{1} = 0$, while leaving other results to Appendix \ref{app:experiments}. Since all tests considered here are invariant with respect to $\beta_{[-1]}$, we assume $\beta_{[-1]} = 0$ without loss of generality. Given a design matrix $X$ and an error distribution $F$, we start by computing a benchmark signal-to-noise ratio $\beta_{1}$ such that the t- or F-test has approximately 20\% power, using Monte-Carlo simulation, where $y$ is generated from
\[y = X_{1}\beta_{1} + \eps, \quad \mbox{where }\eps_{i}\sim F.\]
Then all tests are performed on $X$ and the following $18000$ outcome vectors $y_{s}^{(b)}$, respectively:
\[y_{s}^{(b)}\triangleq X_{1}(s\beta_{1}) + \eps^{(b)}, \quad \mbox{where }s = 0, 1, \ldots, 5, \,\,b = 1,\ldots, 3000.\]
For each $s$, the proportion of rejections among 3000 $\eps$'s is computed. When $s = 0$, this proportion serves as an approximation of the Type I error and should be closed to or below $\alpha$ for a valid test; when $s > 0$, it serves as an approximation of power and should be large for a powerful test. For each of the three types of design matrices, the above experiments are repeated on 50 independent copies of $X$'s. 


Figure \ref{fig:size_r1} presents the Type I error of all tests for three types of design matrices. The boxplots present the variation among 50 independent copies of design matrices. In all cases, the three variants of the \CPT ~are valid, as guaranteed by our theory, while GroupBound is overly conservative. The permutation test and Freedman-Lane test appear to be valid in our simulation settings even though there is no theoretical guarantee for heavy-tailed errors. When errors are Gaussian, the t-test is valid, as guaranteed by theory, but can be conservative or anti-conseravative with heavy-tailed errors depending on the design matrix. 
On the other hand, the test based on the \LAD ~regression is anti-conservative when $X$ is a realization of Gaussian matrices and the error distribution is Gaussian or Cauchy, although validity can be proved asymptotically under regularity conditions that are satisfied by realizations of Gaussian matrices with high probability \citep[e.g.][]{pollard1991asymptotics}. This makes a case for the fragility of some asymptotic guarantees.

\begin{figure}[H]
  \centering
  \includegraphics[width = 0.48\textwidth]{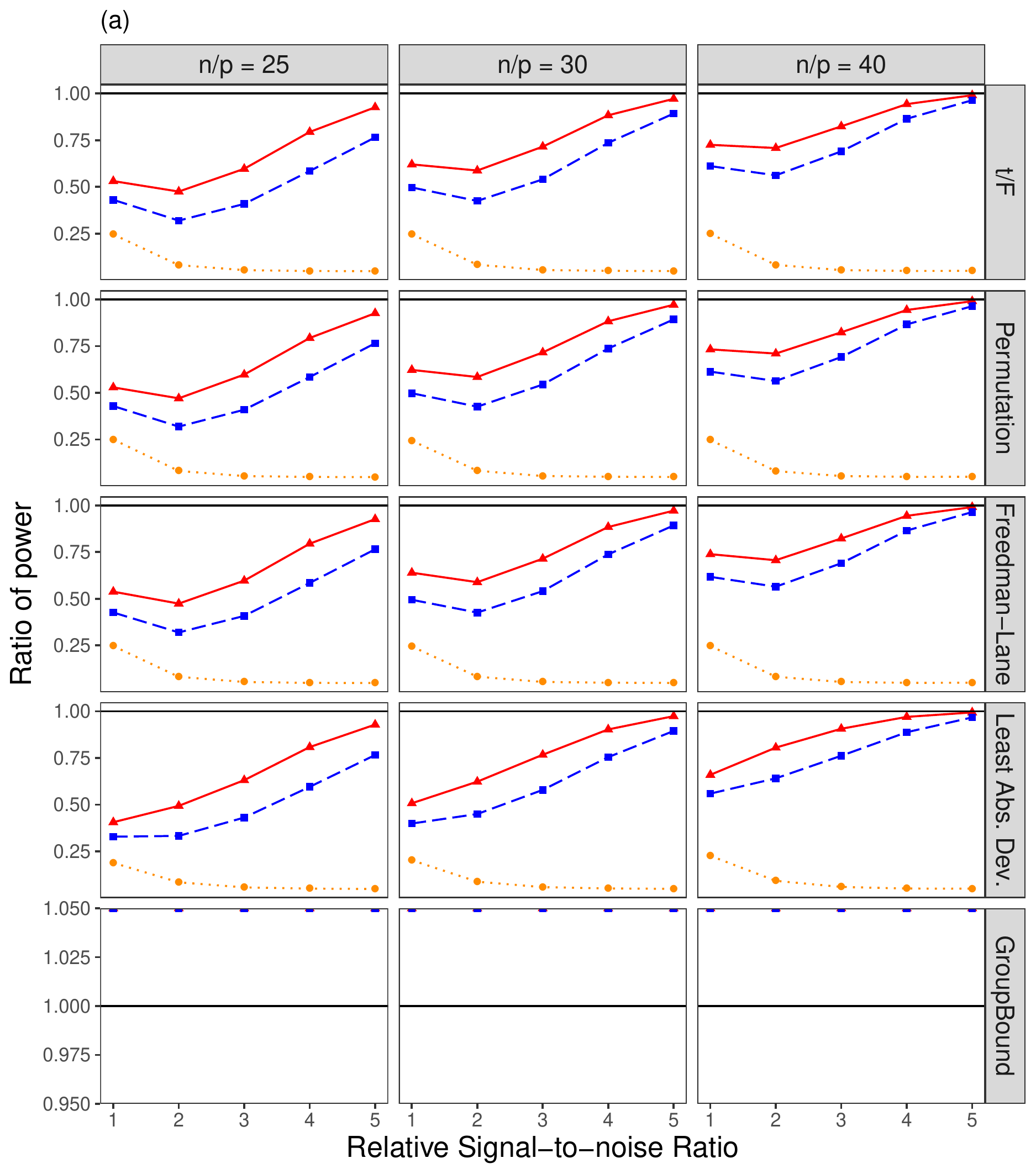}
  \includegraphics[width = 0.48\textwidth]{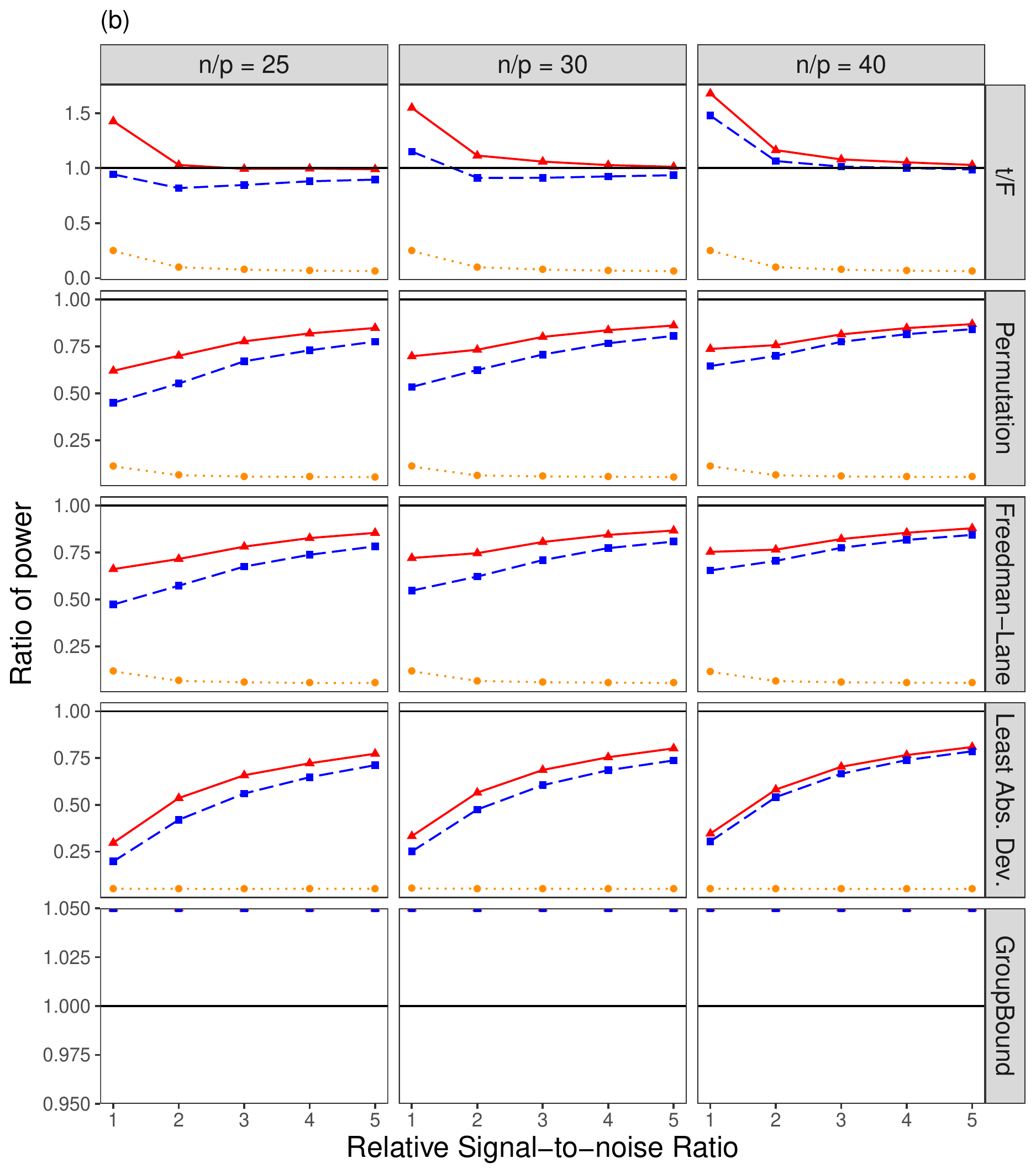}
\caption{Median power ratio between each variant of the CPT, one with stronger ordering via the Genetic
Algorithm (red solid line), one with weaker ordering via the Genetic Algorithm (blue dashed line) and one
with random ordering via Stochastic Search (orange dotted line), to each competing test displayed in each
row, for testing a single coordinate in the case with (a) realizations of Gaussian matrices and Gaussian
errors, (b) realizations of Cauchy matrices and Cauchy errors. The black solid line indicates equal power.
The missing values in the bottom row correspond to infinite ratios.}\label{fig:power_1}
\end{figure}

For power comparison, we only show results for the case where the design matrices are realizations of Gaussian (resp. Cauchy) matrices and errors are Gaussian (resp. Cauchy) in Figure \ref{fig:power_1}; the results for other cases will be presented in Appendix \ref{app:experiments}. All figures plot the median power ratio, obtained from $50$ independent copies of $X$'s, between each variant of the \CPT ~and each competing test. First we see that GroupBound has zero power in all scenarios, so the power ratios are infinite, and hence missing in the plots. Second, the pre-ordering step plays an important role in raising the power of the \CPT. Third, the relative power of the \CPT, with ordering via the \GA, improves as $n / p$ increases. Furthermore, in the Gaussian case, it is not surprising that the t-test is the most powerful one because it is provably the uniformly most powerful unbiased test for linear models with Gaussian errors. The efficiency loss of the \CPT ~against the t-test, permutation t-test and the test based on \LAD ~regression is moderate in general and is low when the sample size per parameter and the signal-to-noise ratio is large. In the Cauchy case, the \CPT ~is more powerful than the t-test.

\section{Discussion}\label{sec:discussion}


\subsection{Confidence interval/region by inverting the test}

It is straightforward to deduce a confidence region for $\beta_{[r]}$ by inverting the \CPT. Specifically, the inverted confidence region is given by $\mathcal{I} \triangleq \left\{\beta_{[r]}: p(y - X\beta; X) > \alpha\right\}$, where $p(y; X)$ is the p-value produced by the \CPT ~with a design matrix $X$ and an outcome vector $y$. Under the construction C3, 
\[(y - X\beta)^{T}\eta_{j} = y^{T}\eta_{j} - \gamma^{T}\beta - \delta^{T}\beta_{[r]} I(j = 0).\]
Thus, 
\[\med\left[(y - X\beta)^{T}\eta_{j}\}_{j=0}^{m}\right] = \med\left[\{y^{T}\eta_{j} - \delta^{T}\beta_{[r]} I(j = 0)\}_{j=0}^{m}\right] - \gamma^{T}\beta.\]
Then $\mathcal{I}$ can be simplified as 
\begin{equation}\label{eq:CI}
\mathcal{I} = \left\{\beta_{[r]}: \delta^{T}\beta_{[r]}\in [x_{\min}, x_{\max}]\right\}
\end{equation}
where $x_{\min}$ and $x_{\max}$ are the infimum and the superimum of $x$ such that
\begin{equation}
  \label{eq:inverting}
  \frac{1}{m+1}\left\{ 1 +\sum_{j=1}^{m}I\bigg(|y^{T}\eta_{0} - x - m(x)|\ge |y^{T}\eta_{j} - m(x)|\bigg)\right\} > \alpha,
\end{equation}
and $m(x) = \med\left[\{y^{T}\eta_{j} - x I(j = 0)\}_{j=0}^{m}\right]$. When $r = 1$, the confidence interval \eqref{eq:CI} gives a useful confidence interval simply as 
\[\mathcal{I} = [x_{\min} / \delta, x_{\max} / \delta],\]
where $x_{\min}$ and $x_{\max}$ are the smallest and the largest solutions of \eqref{eq:inverting}. When $r > 1$, the confidence region \eqref{eq:CI} may not be useful because it is unbounded. More precisely, $\beta_{[r]}\in \mathcal{I}$ implies that $\beta_{[r]} + \xi \in \mathcal{I}$ for any $\xi$ orthogonal to $\delta$. We leave the construction of more efficient confidence regions to future research. 

\subsection{Connection to knockoff based inference}

Our test is implicitly connected to the novel idea of knockoffs, proposed by \cite{barber2015controlling} to control the false discovery rate for variable selection in linear models. Specifically, they assume a Gaussian linear model and aim at detecting a subset of variables that control the false discovery rate in finite samples. Unlike the single hypothesis testing considered in this paper, multiple inference requires dealing with the dependence between test statistics for each hypothesis carefully. They proposed an innovative idea of constructing a pseudo design matrix $\td{X}$ such that the joint distribution of $(X_{1}^{T}y, \ldots, X_{p}^{T}y, \td{X}_{1}^{T}y, \ldots, \td{X}_{p}^{T}y)$ is invariant to the pairwise swapping of $X_{j}^{T}y$ and $\td{X}_{j}^{T}y$ all for $j$ with $\beta_{j} = 0$. Then the test statistic for testing $H_{0j}: \beta_{j} = 0$ is constructed by comparing $X_{j}^{T}y$ and $\td{X}_{j}^{T}y$ in an appropriate way, thereby obtaining a valid binary p-value $p_{j}$ that is uniformly distributed on $\{1/2, 1\}$ under $H_{0j}$. The knockoffs-induced p-values marginally resemble the construction of statistics in the \CPT ~with $m = 2, \eta_{0} = X_{j}, \eta_{1} = \td{X}_{j}$. On the other hand, the validity of knockoffs essentially rests on the distributional invariance of $\eps$ under the rotation group while the validity of the \CPT ~relies on the distributional invariance of $\eps$ under the cyclic permutation group. This coincidence illustrates the charm of group invariance in statistical inference.

\subsection{More efficient algorithm for pre-ordering}\label{subsec:discussion_algorithm}

Although a \GA ~is able to solve \eqref{eq:secondary_optimization} efficiently for moderate-sized problems, it is not scalable enough to handle big data. Since the exact minimizer is not required, we can resort to other heuristic algorithms. One heuristic strategy is proposed by \cite{fogel2013convex} by relaxing permutation matrice into doubly stochastic matrices, with $\Pi \one = \Pi^{T}\one = 0$ and $\Pi_{ij}\ge 0$, and optimizing the objective using continuous optimization algorithms. 
This may suggest an efficient gradient based algorithm. We leave this as a future direction.




\section*{Acknowledgement}
Peter J. Bickel was supported by the National Science Foundation (DMS 82978). The authors are grateful
to Peng Ding, William Fithian, editors and reviewers for their constructive feedback.

\begin{appendices}

\section{Technical Proofs}\label{sec:proofs}

\begin{proof}[Proof of Proposition \ref{prop:MRT}]
Let $R_{j}$ be the rank of $S_{j}$ in descending order as defined in \eqref{eq:validity_pvalue}. Then the invariance of $S$ implies the invariance of $(R_{0}, R_{1}, \ldots, R_{m})$. As a result, 
\[R_{0} \stackrel{d}{=} R_{1} \stackrel{d}{=} \cdots \stackrel{d}{=} R_{m}.\]
Then for any $k$, 
\[\P(R_{0} \ge k) = \frac{1}{m + 1}\sum_{j=0}^{m}\P(R_{j} \ge k) = \frac{1}{m + 1}\sum_{j=0}^{m}\E I(R_{j}\ge k) = \frac{1}{m + 1}\E\big|\{j\ge 0: R_{j}\ge k\}\big|.\]
Let $S_{(1)}\ge S_{(2)}\ge \cdots \ge S_{(m+1)}$ be the ordered statistics of $(S_{0}, \ldots, S_{m})$, which may involve ties. Then by definition, $R_{j}\ge k$ whenever $S_{j}\le S_{(k)}$ and thus,
\[\big|\{j\ge 0: R_{j}\ge k\}\big|\ge m - k + 2,\]
implying that $\mathcal{L}(R_{0})\succeq \mathrm{Unif}([m + 1])$. When there is no tie, the set $\{R_{0}, R_{1}, \ldots, R_{m}\}$ is always $\{1, 2, \ldots, m+1\}$ and thus 
\[\P(R_{0}\ge k) = \frac{m - k + 2}{m + 1}.\]
\end{proof}

\begin{proof}[Proof of Proposition \ref{prop:exact_MRT}]
  By Proposition \ref{prop:MRT}, $\L(p)\succeq \mathrm{Unif}([0, 1])$. Thus $\P_{H_{0}}(p\le \alpha)\le \alpha$. If $S$ has no ties almost suresly and $1 / \alpha = (m + 1) / b$ for some integer $b$, Proposition \ref{prop:MRT} implies that $R_{0}\sim \mathrm{Unif}([m+1])$ and thus
  \[\P(p\le \alpha) = \P\lb R_{0} \le \frac{b}{m + 1}\rb = \frac{b}{m + 1} = \alpha.\]
\end{proof}

\begin{proof}[Proof of Proposition \ref{prop:transform_S}]
  By definition,
  \begin{align*}
    \pi_{L}(\td{S}) &= \{g(S_{1}; S), \ldots, g(S_{m}; S), g(S_{0}; S)\}\\
                    &= \{g(\pi_{L}(S)_{0}; S), g(\pi_{L}(S)_{1}; S), \ldots, g(\pi_{L}(S)_{m}; S)\}\\
                    &= \{g(\pi_{L}(S)_{0}; \pi_{L}(S)), g(\pi_{L}(S)_{1}; \pi_{L}(S)), \ldots, g(\pi_{L}(S)_{m}; \pi_{L}(S))\},
  \end{align*}
  where the last line uses the invariance of $g$. The proof is completed by noting that  $S\stackrel{d}{=}\pi_{L}(S)$.
\end{proof}

\begin{proof}[Proof of Proposition \ref{prop:typeI}]
It is left to prove the invariance of $(\eps^{T}\eta_{0}, \ldots, \eps^{T}\eta_{m})$ under the cyclic permutation group. Further, since the last $n - (m + 1)t$ terms are the same for all $j$, it is left to prove the case where $n$ is divisible by $m + 1$. Let $\td{\Pi}$ be the permutation matrix corresponding to $\pi_{L}^{t}$. Then C2 implies that 
\begin{align}
\pi_{L}(\eps^{T}\eta_{0}, \eps^{T}\eta_{1}, \ldots, \eps^{T}\eta_{m}) &= (\eps^{T}\eta_{1}, \ldots, \eps^{T}\eta_{m}, \eps^{T}\eta_{0})\nonumber\\
& = (\eps^{T}\td{\Pi}\eta_{*}, \ldots, \eps^{T}\td{\Pi}^{m}\eta_{*}, \eps^{T}\eta_{*})\nonumber\\
& = (\eps^{T}\td{\Pi}\eta_{*}, \ldots, \eps^{T}\td{\Pi}^{m}\eta_{*}, \eps^{T}\td{\Pi}^{m+1}\eta_{*}) \qquad (\mbox{Since }\td{\Pi}^{m+1} = \mathrm{Id})\nonumber\\
& \stackrel{d}{=} (\eps^{T}\eta_{*}, \ldots, \eps^{T}\td{\Pi}^{m-1}\eta_{*}, \eps^{T}\td{\Pi}^{m}\eta_{*}) \qquad (\mbox{Since }\td{\Pi}\eps \stackrel{d}{=}\eps)\nonumber\\
& = (\eps^{T}\eta_{0},\eps^{T}\eta_{1}, \ldots, \eps^{T}\eta_{m}).\label{eq:invariance}
\end{align}
Repeating \eqref{eq:invariance} for $m - 1$ times, we prove the invariance of $(\eps^{T}\eta_{1}, \ldots, \eps^{T}\eta_{m})$ under the cyclic permutation group.
\end{proof}

\begin{proof}[Proof of Theorem \ref{thm:typeI}]
  When $n / (p - r) > m$, the number of variables $n + p - r$ of \eqref{eq:validity_equation} is larger than the number of equations $(m + 1)(p - r)$. Part (a) is then proved. For any solution of \eqref{eq:validity_equation}, by \eqref{eq:etajy_H0}, 
  \[y^{T}\eta_{j} = (\one^{T}\eta_{*})\beta_{0} + \gamma_{[-r]}^{T}\beta_{[-r]} + \eps^{T}\gamma_{j}.\]
  The proof is completed by noting that the deterministic parts are identical for all $j$ and the stochastic parts are invariant under the cyclic permutation group by Proposition \ref{prop:typeI}.
\end{proof}

\begin{proof}[Proof of Theorem \ref{thm:deltastar}]
First, \eqref{eq:power_equation} can be equivalently formulated as
\[B(X)^{T}\eta = \delta e_{1, pm}.\]
This can be further rewritten as
\begin{equation}
  \label{eq:power_equation2}
  \delta = B(X)_{1}^{T}\eta, \quad B(X)_{[-1]}^{T}\eta = 0.
\end{equation}
For any $\eta$ satisfying the second constraint,
\[H_{[-1]}\eta = 0,\]
and thus
\[B(X)_{1}^{T}\eta = B(X)_{1}^{T}(I - H_{[-1]})\eta = \td{\eta}^{T}\eta.\]
As a result, 
\[\max_{B(X)_{[-1]}^{T}\eta = 0, \|\eta\|_{2} = 1} B(X)_{1}^{T}\eta \le \max_{\|\eta\|_{2} = 1} \td{\eta}^{T}\eta = \|\td{\eta}\|_{2}.\]
In other words, we have shown that $\delta^{*}(X)\le \|\td{\eta}\|_{2}$. On the other hand, the vector $\td{\eta} / \|\td{\eta}\|_{2}$ satisfies the constraint \eqref{eq:power_equation2} and 
\[B(X)_{1}^{T}\td{\eta} / \|\td{\eta}\|_{2} = \|\td{\eta}\|_{2}.\]
This shows that $\delta^{*}(X) \ge \|\td{\eta}\|_{2}$. In this case, it is obvious that $O^{*}(X) = \delta^{*}(X)$. Therefore, $O^{*}(X) = \|\td{\eta}\|_{2}$ and one maximizer is $\eta^{*}(X) = \td{\eta} / \|\td{\eta}\|_{2}$.
\end{proof}

\section{1908-2018: A Selective Review of The Century-Long Effort}\label{sec:epitome}

The linear model is fundamental in the history of statistics and has been developed for over a century. Nowadays it is still a widely-used model for data analysts to demystify complex data, as well as a powerful tool for statisticians to understand complicated methods and expand the toolbox for advanced tasks. It is impossible to exhaust the literature for this long-standing problem. We thus provide a selective yet extensive review to highlight milestones in the past century. In particular, we focus on the linear hypothesis testing problem, as well as the estimation problem which can yield the former, for vanilla linear models with general covariates, and briefly discuss special cases such as location problems and \ANOVA ~problems when necessary. However, we exclude the topics such as Bayesian linear models, high dimensional sparse linear models, selective inference for linear models, linear models with dependent errors, high breakdown regression methods, linear time series, and generalized linear models. We should emphasize that these topics are at least equally important as those discussed in this section; they are excluded simply to avoid digression. 

\subsection{Normal theory-based tests}

Motivated by the seminal work by \cite{student1908probable} and \cite{student1908probable2} which propose the one-sample and two-sample t-test, Ronald A. Fisher derived the well-known t-distribution \citep{fisher1915frequency} and applied it to testing a single regression coefficient in homoskedastic Gaussian linear models \citep{fisher1922goodness}. In his 1922 paper, he also derived a test that is equivalent to the F-test for testing the global null under the same setting. Later he derived the F-distribution \citep{fisher1924036} which he characterized through ``z'', the half logarithm of F-statistics, and proposed the F-test for \ANOVA ~problems. Both tests were elaborated in his impactful book \citep{fisher1925statistical}, and the term ``F-test'' was coined by George W. Snedecor \citep{snedecor1934calculation}. 

This line of foundational work established the first generation of rigorous statistical tests for linear models. They are {\em exact tests} of linear hypotheses in linear models with \iid ~normal errors and arbitrary fixed-design matrices. Although the exactness of the tests requires no assumption on the design matrices, the normality assumption can rarely be justified in practice. Early investigations of the test validity with non-normal errors can be dated back to Egon S. Pearson \citep{pearson1929some, pearson1929distribution, pearson1931analysis}. Unlike the large-sample theory that is standard nowadays, the early works took an approximation perspective to improve the validity in small samples. It was furthered in the next few decades \citep[e.g.][]{eden1933validity, bartlett1935effect, geary1947testing, gayen1949distribution, gayen1950distribution, david1951effect, david1951method, box1953non, box1962robustness, pearson1975relation} and it was mostly agreed that the regression t-test is extremely robust to non-normal errors with a moderately large sample size (e.g. $>30$) while the regression F-test is more sensitive to the deviation from normality. It is worth emphasizing that these results were either based on mathematically unrigorous approximation or based on the Edgeworth expansion theory that could be justified rigorously \citep[e.g.][]{esseen1945fourier, wallace1958asymptotic, bhattacharya1978validity} in the asymptotic regime that the sample size tends to infinity while the dimension of the parameters stays relatively low (e.g. a small constant). 

Later on, due to the popularization of rigorous large-sample theories in 1950s \citep[e.g.][]{lecam1953some, chernoff1956large}, pioneered by \cite{doob1935limiting}, \cite{wilks1938large}, \cite{mann1943stochastic}, and \cite{wald1949note}, statisticians started to investigate the validity of regression t- and F-tests in certain asymptotic regimes. This can be dated back to Friedhelm Eicker \citep{eicker1963asymptotic, eicker1967limit}, to the best of our knowledge, and developed by Peter J. Huber in his well-known and influential paper \citep{huber1973robust}, which showed that the least squares estimate is jointly asymptotically normal if and only if the maximum leverage score tends to zero. This clean and powerful result laid the foundation to asymptotic analyses for the t- and F-test \citep[e.g.][]{arnold1980asymptotic}. Notably these early works did not assume that the dimension $p$ stays fixed, as opposed to the simplified arguments in standard textbooks. Before 1990s, the large-sample theory for least squares estimators were well established in the regime where the sample size per parameter $n / p$ grows to infinity, under regularity conditions on the design matrices and on the errors, typically with \iid ~elements and finite moments. It shows that both the t- and F-test are asymptotically valid and can be approximated by the z- and $\chi^2$-test, respectively. For the t-test, the robustness to non-normality was proved even without typical regularity conditions (e.g. \cite{zellner1976bayesian, jensen1979linear} for spherically invariant errors, \cite{efron69, cressie1980relaxing, benjamini1983t, pinelis1994extremal} for orthant symmetric errors) or beyond the aforementioned regime \citep[e.g.][]{lei2018asymptotics}. In contrast, though similar results exist for the F-test \citep[e.g.][]{zellner1976bayesian}, more non-robustness results were established. For instance, a line of work \citep[e.g.][]{boos1995anova, akritas2000asymptotics, calhoun2011hypothesis, anatolyev2012inference} showed that the F-test is asymptotically invalid, unless the errors are normal, in the moderate dimensional regime where $n / p$ stays bounded as $n$ approaches infinity, although correction is available under much stronger assumptions on the design matrix or the coefficient vectors. Even with normal errors, \cite{zhong2011tests} showed that the power of the F-test diminishes as $n / p$ approaches $1$. In sum, there have been tremendous efforts over the past century put into the robustness of the regression t- and F-test and it was agreed that the t-test is insensitive to non-normality, high dimensions and irregularity of design matrices to certain extent while the F-test is less robust in general.

\subsection{Permutation tests}

Despite tremendous attentions on the regression t- and F-test, other methodologies were developed in parallel as well. The earliest alternative is the permutation test, which justifies the significance of the test through the so-called ``permutation distribution''. However, the early attempts to justify permutation tests were based on the ``randomization model'' in contrast to the ``population model'' that we considered in \eqref{eq:lm_matrix}. The ``randomization model'' was introduced by Jerzy S. Neyman in his master thesis \citep{neyman1923application} and coined by Ronald A. Fisher in 1926 \citep{fisher1926arrangement}. It is also known as the Neyman-Rubin model \citep{rubin1974estimating}, or design-based inference (\cite{sarndal1978design}, in contrast to the model-based inference), or ``conditional-on-errors'' model (\cite{kennedy1995randomization}, in contrast to the ``conditional-on-treatment'' model). The theoretical foundation of permutation tests was laid by Edwin J. G. Pitman in his three seminal papers \citep{pitman1937significance1, pitman1937significance, pitman1938significance}, with the last two focusing on regression problems, albeit under the ``randomization model''. The early works viewed the permutation test as a better machinery in terms of the logical coherence and robustness to non-normality \citep[e.g.][]{geary1927some, eden1933validity, fisher1935logic}. They found that the permutation distribution under the ``randomization model'' mostly agree with the normality-based distribution under the ``population model'', until 1937 when Li B. Welch disproved the agreement for Latin-squares designs \citep{welch1937z}. In the next few decades, most works on permutation tests were established under the ``randomization model'' without being justified under the ``population model''. We will skip the discussion of this period and refer to \cite{berry2013} for a thorough literature review on this line of work, because our paper focuses on the ``population model'' like \eqref{eq:lm_matrix}. 

The general theory of permutation tests under the ``population model'' can be dated back to \cite{hoeffding1952large} and \cite{box1955permutation}, and was further developed by e.g. \cite{romano1989bootstrap}, \cite{romano1990behavior}, \cite{chung2013exact}. For regression problems, early studies investigated special cases in \ANOVA ~problems \citep{mehra1969class, brown1982distribution, welch1990construction}. For testing a single regression coefficient, \cite{oja1987permutation} and \cite{collins1987permutation} proposed permutation tests on a linear statistic and the F-statistic by permuting the covariates of interest. Whereas the procedure can be easily validated for univariate regressions, the validity was only justified under the ``randomization model'' when $p > 1$. \cite{manly1991randomization} proposed permuting the response vector $y$, which is valid for testing the global null $\beta = 0$ but not for general linear hypotheses. \cite{freedman83}, \cite{ter1992permutation} and \cite{kennedy1996randomization} proposed three different permutation tests on regression residuals. The theoretical guarantees of the aforementioned tests were established in a later review paper by \cite{anderson2001permutation}. The main take-away message is that the permutation test should be performed on asymptotically pivotal statistics. For instance, to test for a single coefficient, the permutation t-test is asymptotically valid. This was further confirmed and extended by \cite{diciccio2017robust} to heteroscedastic linear models with random designs.

\subsection{Rank-based tests}

Rank-based methods for linear models can be dated back to 1936, when \cite{hotelling1936rank} established the hypothesis testing theory for rank correlation, nowadays known as the Spearman's correlation. This work can be regarded as an application of rank-based methods for univariate linear models. Appealed by the normality-free nature of rank-based tests, Milton Friedman extended the idea to one-way \ANOVA ~problems \citep{friedman1937use}. It can be identified as the first application of rank-based methods for multivariate linear models and was further developed by \cite{kendall1939problem} and \cite{friedman1940comparison}. Friedman's test transforms continuous or ordinal outcomes into ranks. It was widely studied for \ANOVA ~problems, started by the famous Kruskal-Wallis test for one-way \ANOVA ~\citep{kruskal1952use} and extended to two-way \ANOVA ~problems and factorial designs \citep{hodges1962rank, puri1966class, sen1968class, conover1976some, conover1981rank, akritas1990rank, akritas1994fully, brunner1994rank, akritas1997nonparametric}. As of 90s, motivated by the advances of high dimensional asymptotic theories, further progresses have been made to refine the procedures in presence of large number of factors or treatments \citep{brownie1994type, boos1995anova, wang2004rank, bathke2005rank, bathke2008nonparametric}. 

However the aforementioned works are restricted to \ANOVA ~problems, with a few exceptions \citep[e.g.][]{sen1968estimates, sen1969class}, and fundamentally different from the modern rank tests based on regression R-estimates, themselves based on ranks of regression residuals. The first R-estimate-based test can be dated back to \cite{hajek1962asymptotically}, which derived the asymptotically most powerful rank test for univariate regressions when the error distribution is known. \cite{adichie1967asymptotic} extended the idea to testing the intercept and the regression coefficient simultaneously. It was further extended to testing the global null for multivariate regressions \citep{koul1969asymptotic}. Tests for general sub-hypotheses were first proposed by \cite{koul1970class} and \cite{puri1973note} for bivariate regressions. The general theory of testing sub-hypotheses were independently developed by \cite{srivastava1972asymptotically}, \cite{mckean1976tests} and \cite{adichie1978rank}. The underlying theory is based on the seminal work by Jana \Ju ~\citep{jureckova1969asymptotic}, as a significant generalization of \cite{hodges1963estimates} for location problems and \cite{adichie1967estimates} for univariate regressions. Her work was further extended by \cite{jureckova1971nonparametric} and \cite{van1972analogue}. However, these approaches are computationally intensive due to the discreteness of ranks. A one-step estimator was proposed by \cite{kraft1972linearized}, which is asymptotically equivalent to the maximum likelihood estimators if the error distribution is known. Another one-step rank-based estimator, motivated by \cite{bickel1975one} for M-estimators, was proposed by \cite{mckean1978robust}. On the other hand, \cite{jaeckel1972estimating} proposed a rank-based objective function, later known as the Jaeckel's dispersion function, that is convex in $\beta$ whose minimizer is asymptotically equivalent to \Ju 's score-based estimators. \cite{hettmansperger1978statistical} found an equivalent but mathematically more tractable formulation of the Jaeckel's dispersion function as the sum of pairwise differences of regression residuals. A weighted generalization of the dispersion function was introduced by \cite{sievers1983weighted}, which unifies the Jaeckel's dispersion function and Kendall's tau-based dispersion function \citep{sen1968estimates, sievers1978weighted}. Three nice survey papers were written by \cite{adichie198411}, \cite{aubuchon198412}, and \cite{draper1988rank}. In 90s, motivated by the development of quantile regressions \citep{koenker1978regression}, \cite{gutenbrunner1992regression} found an important coincidence between the dual problem of the quantile regression and the ``rank-score process'', which generalizes the notion introduced by \cite{hajek1967theory} to linear models. \cite{gutenbrunner1993tests} then developed a rank-score test for linear hypotheses; see also \cite{koenker19978} for a review. 
In the past two decades, there were much fewer works on rank-based tests for linear models \citep[e.g.][]{feng2013rank}.

\subsection{Tests based on regression M-estimates}

Regression M-estimates were introduced by Peter J. Huber in 1964 for location problems \citep{huber1964robust}. The idea was soon extended to linear models by \cite{relles68}, who proved the asymptotic theory for Huber's loss with $p$ fixed and $n$ tending to infinity. The theory was further extended to general convex loss functions by \cite{yohai72}. Despite the appealing statistical properties, the computation remained challenging in 1970s. \cite{bickel1975one} proposed one-step M-estimates that are computational tractable with the same asymptotic property as full M-estimates. In addition, he proved the uniform asymptotic linearity of M-estimates, which is a fundamental theoretical result that laid the foundation for later works. Based on \cite{bickel1975one}'s technique, \cite{jureckova1977asymptotic} established the relation between regression M- and R-estimates. The asymptotic normality of M-estimates directly yields an asymptotically valid Wald-type test for general linear hypotheses. \cite{schrader1980robust} developed an analogue of the likelihood-ratio test based on M-estimators for sub-hypotheses. It was further extended to general linear hypotheses by \cite{silvapulle1992robust}. However, both Wald-type tests and likelihood-ratio-type tests involve unknown nuisance parameters. To get rid of them, \cite{sen1982m} proposed the M-test as an analogue of the studentized score test, which is able to test general linear hypotheses with merely an estimate of regression coefficients under the null hypothesis. It is known that the Rao's score test may not be efficient in presence of nuisance parameters. \cite{singer1985m} discussed an efficient test, which is essentially the analogue of Neyman's $C(\alpha)$ test based on projected scores \citep{neyman1959optimal}, although it brings back nuisance parameters. M-tests were later investigated and generalized in a general framework based on influence functions \citep[e.g.][]{boos1992generalized, markatou19973}.

As with the regression t- and F-test, the robustness to high dimensionality was investigated extensively for M-estimators in general linear models. In Huber's 1972 Wald Lectures \citep{huber72}, he conjectured that the asymptotic normality of M-estimates proved by \cite{relles68} can be extended to the asymptotic regime where $p$ grows with $n$. The conjecture was proved one year later in the regime $\kappa p^{2} = o(1)$, where $\kappa$ is the maximum leverage score, which implies $p = o(n^{1/3})$ \citep{huber1973robust}. This was improved to $\kappa p^{3/2} = o(1)$ by \cite{yohai1979asymptotic}, which implies that $p = o(n^{2/5})$, to $p = o(n^{2/3} / \log n)$ by \cite{portnoy85} under further regularity conditions on the design matrix, and to $\kappa n^{1/3}(\log n)^{2/3} = o(1)$, which implies that $p = o(n^{2/3} / (\log n)^{2/3})$. All aforementioned results are derived for smooth loss functions. For non-smooth loss functions, \cite{welsh1989m} obtained the first asymptotic result in the regime $p = o(n^{1/3} / (\log n)^{2/3})$. It was improved to $p = o(n^{1/2})$ by \cite{bai1994limiting}. For a single coordinate, \cite{bai1994limiting} showed the asymptotic normality in the regime $p = o(n^{2/3})$. These works prove that the classical asymptotic theory holds if $p <\!\! < n^{2/3}$. However, in moderate dimensions where $p$ grows linear with $n$, the M-estimates are no longer consistent in $L_{2}$ metric. For certain random designs, the estimation error $\|\hat{\beta} - \beta\|_{2}^{2}$ converges to a non-vanishing quantity determined by $p / n$, the loss function and the error distribution through a complicated system of non-linear equations \citep{elkaroui11,bean11, elkaroui13,donoho16,elkaroui2018impact}. This surprising phenomenon marks the failure of the classical asymptotic theory for M-estimators. For least squares estimators, \cite{lei2018asymptotics} showed that the classical t-test with appropriate studentization is still asymptotically valid under regularity conditions on the design matrix. \cite{cattaneo2018inference} proposed a refined test for heteroscedastic linear models. However it is unclear how to test general linear hypotheses with general M-estimators in this regime, even for a single coordinate. \cite{lei2018asymptotics} provided the only fixed-design result for the asymptotic property of a single coordinate of general M-estimates in this regime. For the purpose of hypothesis testing, the null variance needs to be estimated but no consistent variance estimator is known at this moment, except for special random designs \citep[e.g.][]{bean11}.

\subsection{Tests based on regression L-estimates}

L-estimators constitute an important class of robust statistics based on linear combination of order statistics. Frederick Mosteller proposed the first L-estimator for Gaussian samples \citep{mosteller1946some}. This was further developed in the following two decades \citep[e.g.][]{hastings1947low, lloyd1952least, evans1955atomic, jung1956linear, tukey1960survey, bickel1965some, gastwirth1966robust}. In particular, John W. Tukey advocated the trimmed mean and Winsorized mean in his far-reaching paper \citep{tukey1962future}, which he attributed to Charles P. Winsor based on their personal communication in 1941. One year later, the well-known Hodges-Lehmann estimator was developed \citep{hodges1963estimates}, which established the first connection between R- and L-estimates. For location problems, \cite{bickel1975descriptive} found the superiority of L-estimates over M- and R-estimates. 

Despite the simplicity and the nice theoretical property of L-statistics, they are not easy to be generalized to linear models. The first attempt was made by \cite{bickel1973some}, which proposed a one-step L-estimate for general linear models. However, this estimator is not equivariant to affine transformations of the design matrix. Motivated by this paper, \cite{welsh1987one} proposed a class of one-step L-estimators that are equivariant to reparametrization of the design matrix. \cite{welsh1991asymptotically} further extended the idea to construct an adaptive L-estimator. Another line of thoughts were motivated by the pinoneering work of \cite{koenker1978regression}, which introduced the notion of regression quantiles as a natural analogue of sample quantiles for linear models. Although the quantile regression yields an M-estimator, it had been the driving force for the development of regression L-estimators since 80s. In this paper, they proposed another class of L-estimators by a discrete weighted average of regression quantiles and derived its asymptotic distribution. This idea was furthered by \cite{koenker1987estimation} to L-estimators with continuous weights, by \cite{portnoy1989adaptive} to adaptive L-estimators, and by \cite{koenker1994estimatton} to heteroscedastic linear models. Another notable strategy of contructing L-statistics is based on weighted least squares with ``outliers'' removed. \cite{ruppert1980trimmed} developed two equivariant one-step estimators as analogues of the trimmed mean. Both estimators can be formulated in the form of weighted least squares with extreme residuals removed. As with \cite{ruppert1980trimmed}, \cite{jureckova1983winsorized} proposed an analogue of the winsorized mean. The Bahadur representation of the trimmed mean least squares estimator was derived by \cite{jurevckova1984regression}. A nice review article of regression L-estimators was written by \cite{alimoradi19989}. The asymptotic results of L-estimators induce asymptotically valid Wald-type tests with consistent estimates of the asymptotic variance. Unlike M-estimators, we are not aware of other types of tests based on L-estimates.

\subsection{Resampling-based tests}

Resampling, marked by the jackknife \citep{quenouille1949problems, quenouille1956notes, tukey1958bias} and bootstrap \citep{efron1979bootstrap}, is a generic technique to assess the uncertainty of an estimator. Although both involving resampling, resampling-based tests are fundamentally different from permutation tests. The former approximates the sampling distribution under the truth while the latter approximates the sampling distribution under the null hypothesis, though they are asymptotically equivalent in many cases \citep[e.g.][]{romano1989bootstrap}. \cite{miller1974unbalanced} proposed the first jackknife-based estimator for general linear models. He showed that the estimator is asymptotically normal, the jackknife variance estimator is consistent, and thus the Wald-type test is asymptotically valid. \cite{hinkley1977jackknifing} pointed out that Miller's estimator is less efficient than the least squares estimator and proposed a weighted jackknife estimator to achieve efficiency. \cite{wu1986jackknife} proposed a general class of delete-$d$ jackknife estimators for estimating the covariance matrix of the least squares estimator. This was extended by \cite{shao1987heteroscedasticity}, \cite{shao1988resampling}, \cite{shao1989jackknifing}, \cite{peddada1992jackknife}, and \cite{liu1992efficiency}. 

On the other hand, David A. Freedman first studied the bootstrapping procedures for linear models \citep{freedman1981bootstrapping}. He studied two types of bootstrap: the residual bootstrap, where the regression residuals are resampled and added back to the fitted values, and the pair bootstrap, where the outcome and the covariates are resampled together. In the fixed-$p$ regime, he showed the consistency of the residual bootstrap under homoscedastic linear models and that of the pair bootstrap under general ``correlation models'' including heteroscedastic linear models. \cite{navidi1989edgeworth}, \cite{hall1989unusual} and \cite{qumsiyeh1994bootstrapping} established the higher order accuracy of the pair bootstrap for linear models and the results were then presented under a broader framework in the influential monograph by Peter Hall \citep{hall1992bootstrap}. \cite{wu1986jackknife} found that the residual bootstrap fails in heteroscedastic linear models because its sampling process is essentially homoscedastic. To overcome this, he introduced another type of bootstrapping method based on random rescalings of regression residuals that match the first and second moment. \cite{liu1988bootstrap} introduced a further requirement to match the third moment and improved the rate of convergence. Later \cite{mammen1993bootstrap} coined this procedure the ``wild bootstrap'' and proved the consistency for least squares estimators under random-design homoscedastic and heteroscedastic linear models. \cite{hu1995bootstrap} proposed an alternative bootstrap procedure for heteroscedastic linear models that resamples the score function instead of the residuals. A wild bootstrap analogue of the score-based bootstrap was proposed by \cite{kline2012score}. In particular, they developed the bootstrap Wald tests and the boostrap score tests for general linear hypotheses. 

The bootstrap techniques were also widely studied for regression M-estimates. The residual bootstrap was extended to M-estimators with smooth loss functions by \cite{shorack1982bootstrapping}. Unlike the least squares estimator, it requires a debiasing step to obtain distributional consistency. \cite{lahiri1992bootstrapping} proposed a weighted residual bootstrap that does not require debiasing. He additionally showed the higher order accuracy of the weighted bootstrap and Shorack's bootstrap for studentized M-estimators. However, this weighted bootstrap is hard to be implemented in general. On the other hand, motivated by Bayesian bootstrap \citep{rubin1981bayesian}, \cite{rao1992approximation} proposed a bootstrapping procedure by randomly reweighting the objective function. This idea was extended by \cite{chatterjee1999generalised} in a broader framework called ``generalized bootstrap''. It was later reinvented by \cite{jin2001simple} and referred to as ``perturbation bootstrap''. The higher order accuracy of perturbation bootstrap was established by \cite{das2019second}. It was pointed out by \cite{das2019second} that the perturbation bootstrap coincides with the wild bootstrap for least squares estimators. \cite{hu2000estimating} proposed another estimating function based bootstrap, as essentially a resampling version of \cite{sen1982m}'s M-tests. The wild bootstrap was introduced for quantile regressions by \cite{feng2011wild}.

The robustness of bootstrap methods against high dimensions was widely studied in the literature. \cite{bickel1983bootstrapping} proved the distributional consistency of the residual bootstrap for least squares estimators in the regime $p = o(n)$ for linear contrasts of $\beta$ and in the regime $p = o(n^{1/2})$ for the vector $\beta$, under fixed-design linear models with vanishing maximum leverage scores. They further showed the failure of bootstrap in moderate dimensions where $p / n\rightarrow c \in (0, 1)$ and the usual variance rescaling does not help because the bootstrap distribution is no longer asymptotically normal. For M-estimators, \cite{shorack1982bootstrapping} showed that the debiased residual bootstrap is distributionally consistent in the regime $p = o(n^{1/3})$ for linear contrasts of $\beta$. The results were extended by \cite{mammen89} to the regime $p = o(n^{2/3} / (\log n)^{2/3})$ for linear contrasts of $\beta$, and to the regime $p = o(n^{1/2})$ for the vector $\beta$. For random designs with \iid ~observations, \cite{mammen1993bootstrap} proved the distributional consistency of both the pair bootstrap and wild bootstrap for linear contrasts of $\beta$ in the regime $p = o(n^{a})$ for arbitrary $a < 1$. He also proved the consistency under heteroscedastic linear models in the regime $p = o(n^{3/4})$ for the pair bootstrap, and in the regime $p = o(n^{1/2})$ for the wild bootstrap. This was further extended by \cite{chatterjee1999generalised} to the generalized bootstrap, including the perturbation bootstrap \citep{rao1992approximation}, $m$-out-of-$n$ bootstrap \citep{bickel2008choice} and delete-d jackknife \citep{wu1990asymptotic}. On the other hand, extending \cite{bickel1983bootstrapping}'s negative result, \cite{el2018can} showed the failure of various bootstrap procedures for M-estimators in moderate dimensions, including the pair bootstrap, residual bootstrap, wild bootstrap and jackknife. 

\subsection{Other tests}

A generic strategy for hypothesis testing is through pivotal statistics. Specifically, if there exists a statistic $S$ whose distribution is fully known, then the rejection rule $S \in \mathcal{R}^{c}$ for any region $\mathcal{R}$ with $P(S \in \mathcal{R}) \ge 1 - \alpha$ yields an exact test. For linear models, it is extremely hard to find a pivotal statistic under general linear hypotheses, except for Gaussian linear models under which the t- and F-statistics are pivotal. However, if the goal is to test all coefficients plus the intercept, i.e. $H_{0}: \beta_{0} = \gamma_{0}, \beta = \gamma$, then one can recover the stochastic errors as $\eps_{i} = y_{i} - \gamma_{0} - x_{i}^{T}\gamma$ under the null and construct pivotal statistics based on $\eps$. Taking one step further, given a pivotal statistic, one can invert the above test to obtain a finite-sample valid confidence region $\mathcal{I}$ for $(\beta_{0}, \beta)$, by collecting all $(\gamma_{0}, \gamma)$'s at which the corresponding null hypothesis fails to be rejected. 
This induces a confidence region for $R^{T}\beta$ as $\mathcal{I}' = \{R^{T}\beta: (\beta_{0}, \beta)\in \mathcal{I}\}$. Finally, using the duality between the confidence interval and hypothesis testing again, the test which rejects the null hypothesis is valid for the linear hypothesis $H_{0}: R^{T}\beta = 0$ in finite samples. If $r <\!\!< p$, this seemingly ``omnibus test'' is in general conservative and inferior to the tests discussed in previous subsections. Nonetheless, it stimulates several non-standard but interesting tests that are worth discussions.

The most popular strategy to construct pivotal statistics is based on quantiles of $\eps_{i}$s, especially the median. Assuming $\eps_{i}$s have zero median, \cite{fisher1925statistical} first introduced the sign test for location problems, which was investigated and formalized later by \cite{cochran1937efficiencies}. Thirteen years later, Henri Theil proposed an estimator for univariate linear models \citep{theil1950rank, theil1950rank2, theil1950rank3}, later known as the Theil-Sen estimator \citep{sen1968estimates}. \cite{brown1951median} proposed a median test for general linear models by reducing the problem into a contingency table and applying the $\chi^{2}$-tests. The theoretical property of the Brown-Mood test was studied by \cite{kildea1981brown} and \cite{johnstone1985resistant}. \cite{daniels1954distribution} proposed a geometry-based test for univariate linear models, which can be regarded as a generalization of the Brown-Mood test. It was later connected to the notion of regression depth \citep{rousseeuw1999regression} and applied in deepest regression methods \citep{van2002deepest}. The idea of inverting the sign test was exploited in \cite{quade1979regression}. An analogue incorporating Kendall's tau between the residuals and covariates was proposed by \cite{lancaster1985nonparametric}. The idea also attracted some attention in signal processing \citep[e.g.][]{campi2005guaranteed, campi2009non} and econometrics \citep[e.g.][]{chernozhukov2009finite}. It should be noted that the approach is computationally infeasible even in low dimensions. Assuming further the symmetry of $\eps_{i}$s, \cite{hartigan1970exact} proposed a non-standard test based on an interesting notion of typical values. It was designed for location problems but can be applied to certain \ANOVA ~problems. Furthermore, \cite{siegel1982robust} proposed the repeated median estimator and \cite{rousseeuw1984least} proposed the least median squares estimators to achieve a high breakdown point.

The pivotal statistics can also be constructed in other ways. \cite{parzen1994resampling} proposed a bootstrap procedure based on inverting a pivotal estimating function at a random point. This procedure mimics the Fisher's fiducial inference but can be justified under the frequentist framework. Recently \cite{meinshausen2015group} proposed the GroupBound test for sub-hypotheses, which even works for high-dimensional settings where $p >\!\!> n$. However, the validity is only guaranteed for rotationally invariant errors with a known noise level. This assumption is extremely strong as shown by \cite{maxwell1860v}: a rotationally invariant random vector with \iid ~coordinates must be multivariate Gaussian. 

\section{Construction of $\eta$'s When $r > 1$}\label{app:r>1}

Similar to C3, we impose the following restriction on $\eta$.
\begin{enumerate}[{C}1']
\setcounter{enumi}{2}
\item there exists $\gamma_{[r]}, \delta\in \R^{r}$, such that 
\[X_{[r]}^{T}\eta_{j} = \gamma_{[r]}\,\,\, (j = 1, 2, \ldots, m), \quad X_{[r]}^{T}\eta_{0} = \gamma_{[r]} + \delta.\]
\end{enumerate}
Combining with \eqref{eq:validity_equation}, we obtain an analogue of \eqref{eq:power_equation} as follows.
\begin{equation}
  \label{eq:power_equation_r}
  \bigg(-e_{1, p(m+1)}, \ldots, -e_{r, p(m+1)}\,\,\vdots\,\, A(X)^{T}\bigg)\lb \begin{array}{l}
\delta\\
\gamma\\
\eta
\end{array}\rb = 0,
\end{equation}
where $A(X)$ is defined in \eqref{eq:AX} and $\gamma = \com{\gamma_{[r]}}{\gamma_{[-r]}}$. This linear system involves $p(m + 1)$ equations and $n + p + r$ variables. Therefore it always has a non-zero solution if 
\[n + p + r > p(m + 1) \Longleftrightarrow n \ge pm - r + 1.\]

Unlike the univariate case, there are infinite ways to characterize the signal strength since $\delta$ is multivariate. A sensible class of criteria is to maximize a quadratic form
\begin{equation}
  \label{eq:QP}
  \max_{\delta\in \R^{r}, \gamma\in \R^{p}, \eta\in \R^{n}, \|\eta\|_{2} = 1}\,\, \delta^{T}M\delta \quad \mbox{s.t. }   \bigg(-e_{1, p(m+1)}, \ldots, -e_{r, p(m+1)}\,\,\vdots\,\, A(X)^{T}\bigg)\lb \begin{array}{l}
\delta\\
\gamma\\
\eta
\end{array}\rb = 0.
\end{equation}
The following theorem gives the optimal solution given any weighting matrix $M$. Let $O^{*}(X)$ denote the optimal value of the objective function. 
\begin{theorem}\label{thm:deltastar_r}
Assume that $n \ge pm - r + 1$. Let $B(X)$ be defined in \eqref{eq:BX} in the main text. Partition $B(X)$ into $(B(X)_{[r]} \,\, B(X)_{[-r]})$ where $B(X)_{[r]}$ is the matrix formed by the first $r$ columns of $B(X)$. Let 
\[M_{r}(X) = (I - H_{[-r]})B(X)_{[r]}MB(X)_{[r]}^{T}(I - H_{[-r]}),\]
where
\[H_{[-r]} = B(X)_{[-r]}(B(X)_{[-r]}^{T}B(X)_{[-r]})^{+}B(X)_{[-r]}^{T}\] 
Further let $\lambda_{\max}(M_{r}(X))$ denote the maximum eigenvalue, $u$ denote any eigenvector corresponding to it and $\td{\eta} = (I - H_{[-r]})u$. Then $O^{*}(X) = \lambda_{\max}(M_{r}(X))$ and 
\[\eta^{*}(X) = \td{\eta} / \|\td{\eta}\|_{2}, \quad \delta^{*}(X) = B(X)_{[r]}^{T}\eta^{*}(X)\] 
is an optimal solution of \eqref{eq:QP}.
\end{theorem}

\begin{proof}[Proof of Theorem \ref{thm:deltastar_r}]
Similar to the proof of Theorem \ref{thm:deltastar}, we first rewrite \eqref{eq:power_equation_r} as 
\[B(X)_{[r]}^{T}\eta = \delta, \quad B(X)_{[-r]}^{T}\eta = 0.\]
As a result, $\eta$ lies in the row null space of $B(X)_{[-r]}$ and thus
\[H_{[-r]}\eta = 0.\]
Then 
\[\delta^{T}M\delta = \eta^{T}(I - H_{[-r]})B(X)_{[r]}MB(X)_{[r]}^{T}(I - H_{[-r]})\eta = \eta^{T}M_{r}(X)\eta.\]
Since $\|\eta\|_{2}\le 1$, 
\[\delta^{T}M\delta \le \lambda_{\max}(M_{r}(X)).\]
On the other hand, for any eigenvector $u$ of $M_{r}(X)$ corresponding to its largest eigenvalue, let $\td{\eta} = (I - H_{[-r]})u$ and $\eta = \td{\eta} / \|\td{\eta}\|_{2}$, then
\[\eta^{T}M_{r}(X) \eta = \lambda_{\max}(M_{r}(X)), \quad B(X)_{[-r]}\eta = 0, \quad \|\eta\|_{2} = 1.\]
Thus, $\eta^{*}(X) = \td{\eta} / \|\td{\eta}\|_{2}$ is an optimal solution. As a result, $\delta^{*}(X) = B(X)_{[r]}^{T}\eta^{*}(X)$ and $O^{*}(X) = \lambda_{\max}(M_{r}(X))$.
\end{proof}

Although Theorem \ref{thm:deltastar_r} gives the solution of \eqref{eq:QP} for arbitrary weight matrix $M$, it is not clear which $M$ is the best choice. Note that
\[\eta_{j}^{T}y = \delta^{T}\beta_{[r]}I(j = 0) + \td{W}_{j},\]
where $\td{W}_{j} = \gamma^{T}\beta + \eta_{j}^{T}\eps$ is invariant under the cyclic permutation group. Thus, $\delta^{T}\beta_{[r]}$ characterizes the signal strength. In principle, the ``optimal'' weight matrix should depend on the prior knowledge of $\beta_{[r]}$. For instance, for a Bayesian hypothesis testing problem with a prior distribution $Q$ on $\beta_{[r]}$ under the alternative, the optimal weight matrix is $M = \E_{Q}\left[\beta_{[r]}\beta_{[r]}^{T}\right]$.

\section{Complementary Experimental Results}\label{app:experiments}

\subsection{Testing for a single coordinate}

In this appendix we present experimental results that complement Section \ref{sec:experiments}. Figure \ref{fig:power_normal_t1_1} - \ref{fig:power_ANOVA1_t1_1} present the power comparison for testing a single coordinate under the same setting as considered in Section \ref{sec:experiments} for four extra scenarios with realizations of Gaussian matrices + Cauchy errors, realizations of Cauchy matrices + Gaussian errors and realizations of random one-way \ANOVA ~matrices + Gaussian or Cauchy errors, respectively.

\subsection{Testing for multiple coordinates}\label{subsec:multiple_experiment}

Next we consider testing the first five coordinates $H_{0}: \beta_{1} = \ldots = \beta_{5} = 0$, with a Bayesian alternative hypothesis 
\[\beta_{[5]} \sim N(s\one_{5}, \Sigma), \quad \Sigma = \diag(0.2, 0.4, 0.6, 0.8, 1), \quad s\in \{0, 1, \ldots, 5\}\]
All other settings are exactly the same as Section \ref{sec:experiments}, except that the t-test and permutation t-test are replaced by the F-test and permutation F-test. For the \CPT, we choose the weight matrix $M = \E [\beta_{[5]}\beta_{[5]}^{T}]$. Figure \ref{fig:size_r5} presents the Monte-Carlo Type I error of all tests. The results are qualitatively the same as those in Section \ref{sec:experiments}, though the F-test and \LAD -based test become more invalid. Figure \ref{fig:power_normal_normal_5} - \ref{fig:power_ANOVA1_t1_5} present the power results under the same setting as Section \ref{subsec:multiple_experiment} for six scenarios with realizations of Gaussian matrices + Gaussian or Cauchy errors, realizations of Cauchy matrices + Gaussian or Cauchy errors and realizations of random one-way \ANOVA ~matrices + Gaussian or Cauchy errors, respectively.

\newpage
\begin{figure}[h]
  \centering
  \includegraphics[width = 0.9\textwidth]{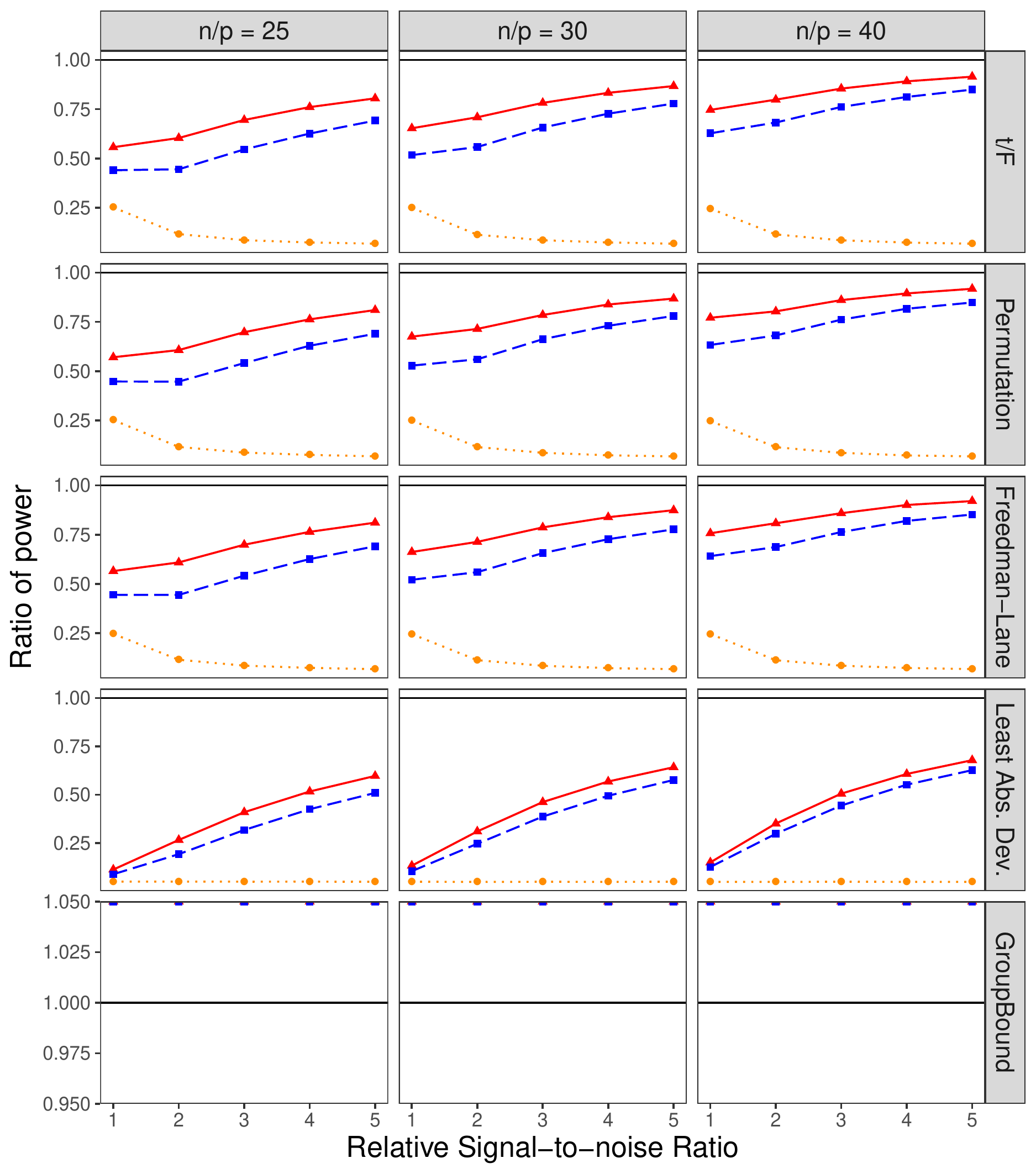}
\caption{Median power ratio of each variant of the cyclic permutation test, one with stronger ordering via
the genetic algorithm (solid), one with weaker ordering via the genetic algorithm (dashed) and one with
random ordering via stochastic search (dotted), to the competing test displayed in each row, for testing a
single coordinate in the case with realizations of Gaussian matrices and Cauchy errors. The black solid line
marks equal power. The missing values in the last row correspond to infinite ratios.}\label{fig:power_normal_t1_1}
\end{figure}

\begin{figure}[h]
  \centering
  \includegraphics[width = 0.9\textwidth]{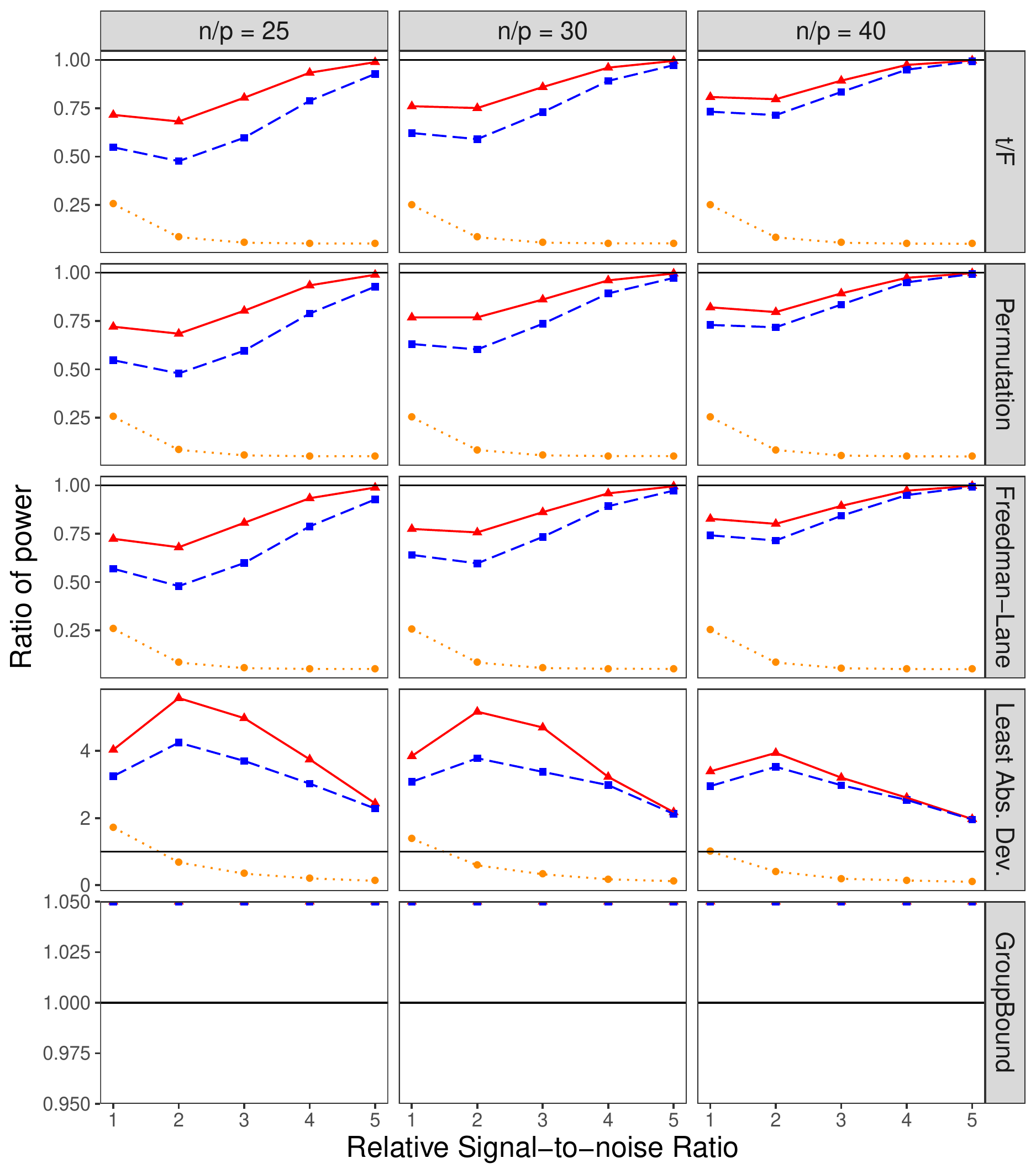}
\caption{Median power ratio of each variant of the cyclic permutation test, one with stronger ordering via
the genetic algorithm (solid), one with weaker ordering via the genetic algorithm (dashed) and one with
random ordering via stochastic search (dotted), to the competing test displayed in each row, for testing a
single coordinate in the case with realizations of Cauchy matrices and Gaussian errors. The black solid line
marks equal power. The missing values in the last row correspond to infinite ratios.}\label{fig:power_t1_normal_1}
\end{figure}

\begin{figure}[h]
  \centering
  \includegraphics[width = 0.9\textwidth]{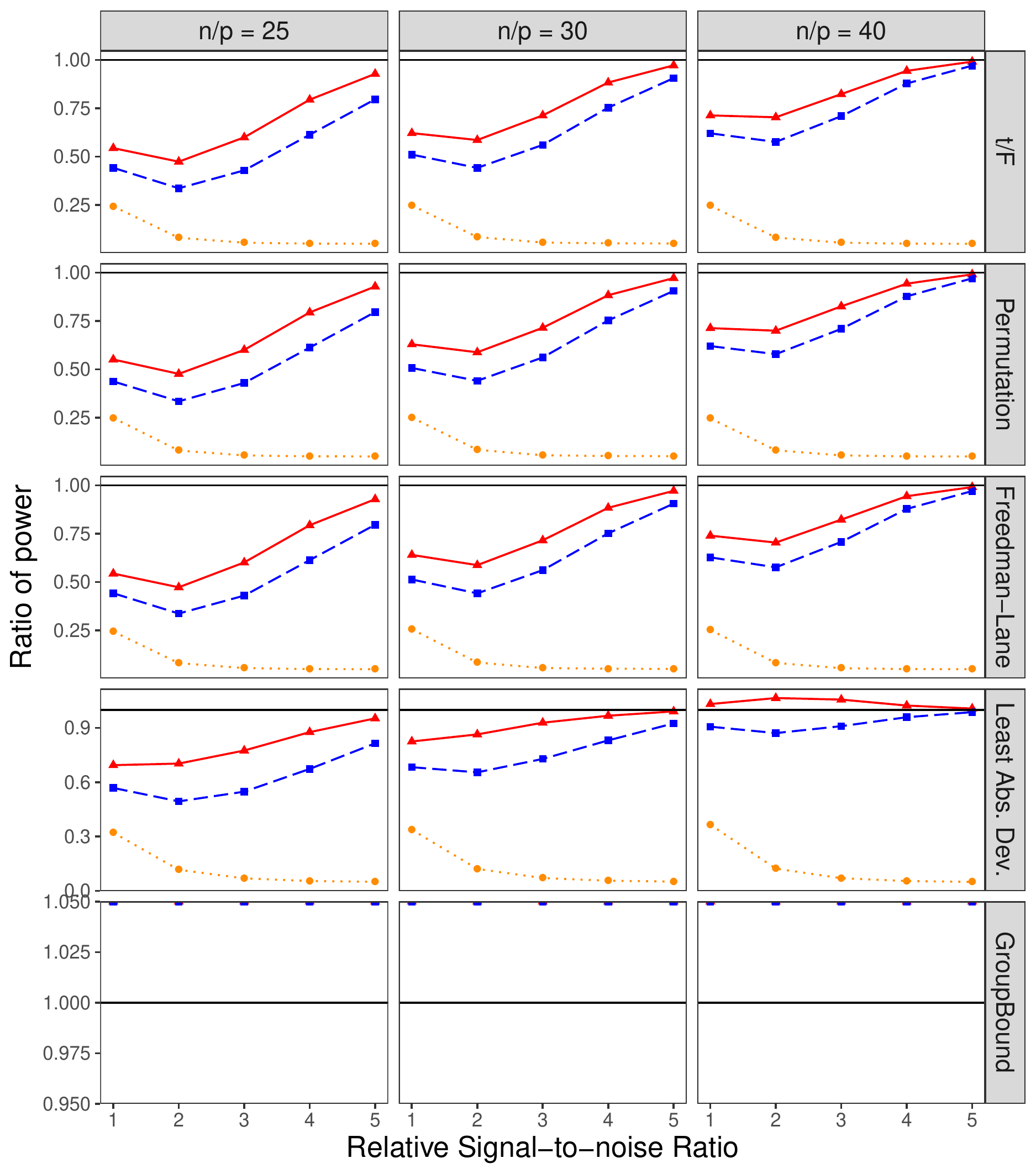}
\caption{Median power ratio of each variant of the cyclic permutation test, one with stronger ordering via
the genetic algorithm (solid), one with weaker ordering via the genetic algorithm (dashed) and one with
random ordering via stochastic search (dotted), to the competing test displayed in each row, for testing a
single coordinate in the case with realizations of random one-way ANOVA matrices and Gaussian errors.
The black solid line marks equal power. The missing values in the last row correspond to infinite ratios.}\label{fig:power_ANOVA1_normal_1}
\end{figure}

\begin{figure}[h]
  \centering
  \includegraphics[width = 0.9\textwidth]{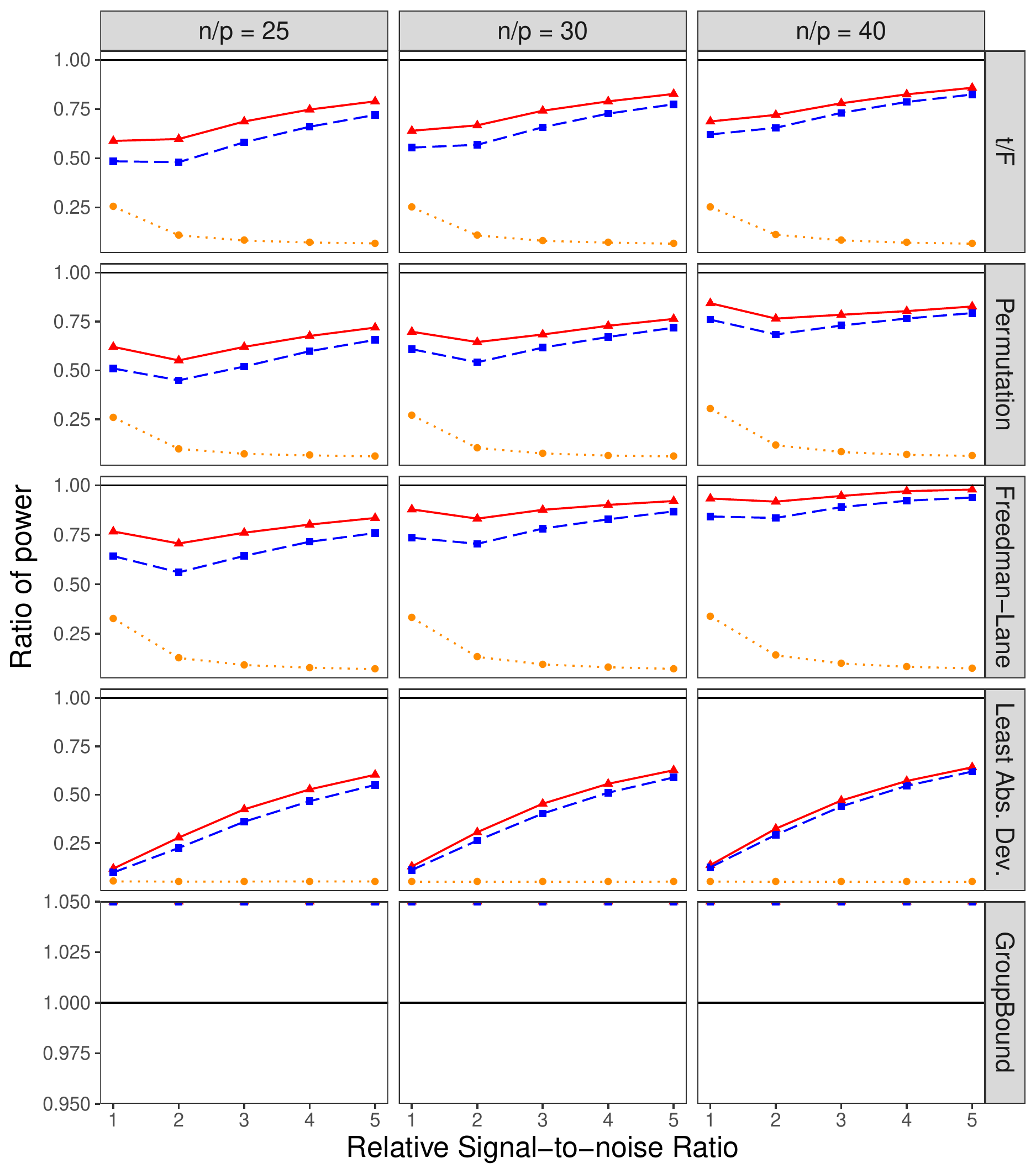}
\caption{Median power ratio of each variant of the cyclic permutation test, one with stronger ordering via
the genetic algorithm (solid), one with weaker ordering via the genetic algorithm (dashed) and one with
random ordering via stochastic search (dotted), to the competing test displayed in each row, for testing a
single coordinate in the case with realizations of random one-way ANOVA matrices and Cauchy errors. The
black solid line marks equal power. The missing values in the last row correspond to infinite ratios.}\label{fig:power_ANOVA1_t1_1}
\end{figure}

\begin{figure}[h]
  \centering
  \includegraphics[width = 0.8\textwidth]{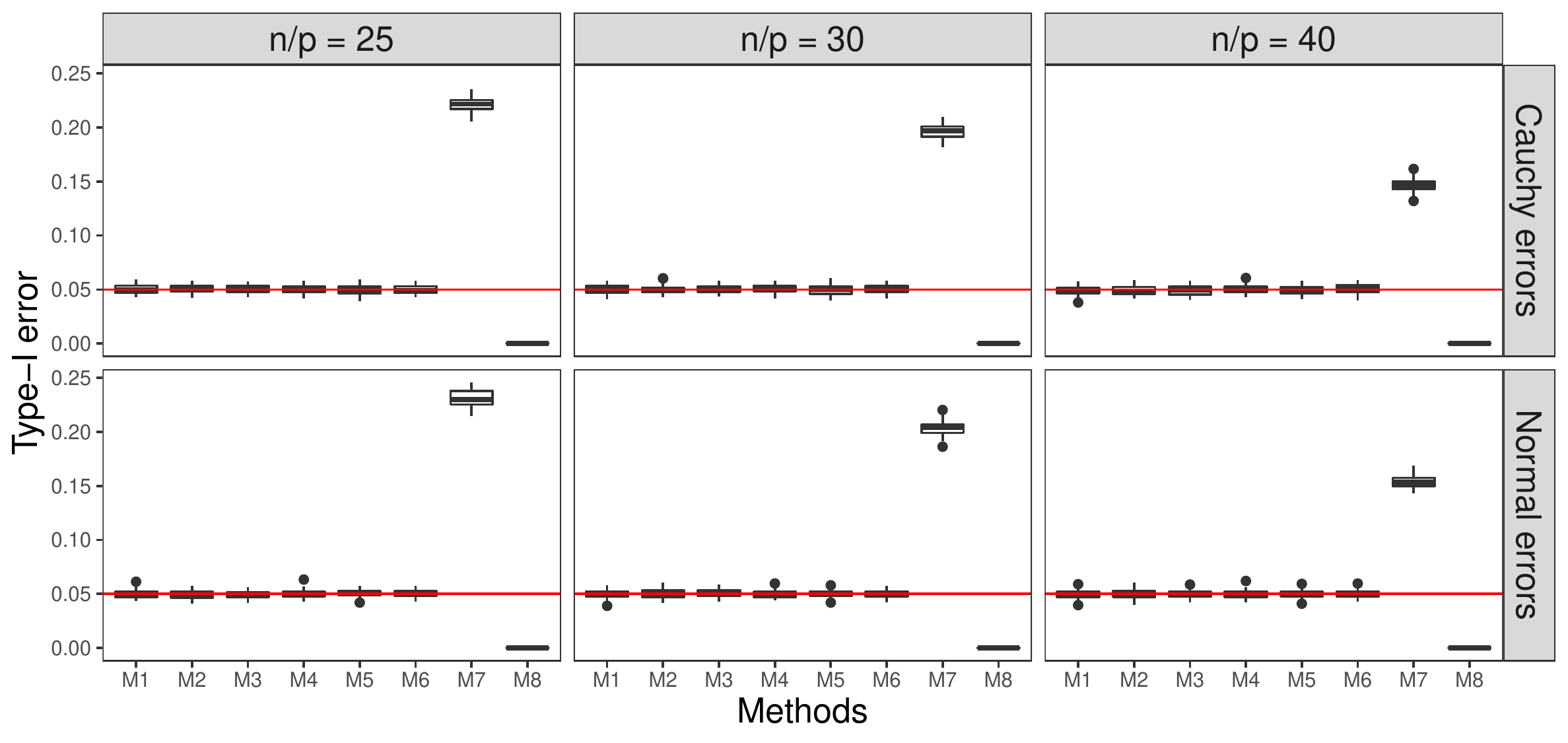}
  \includegraphics[width = 0.8\textwidth]{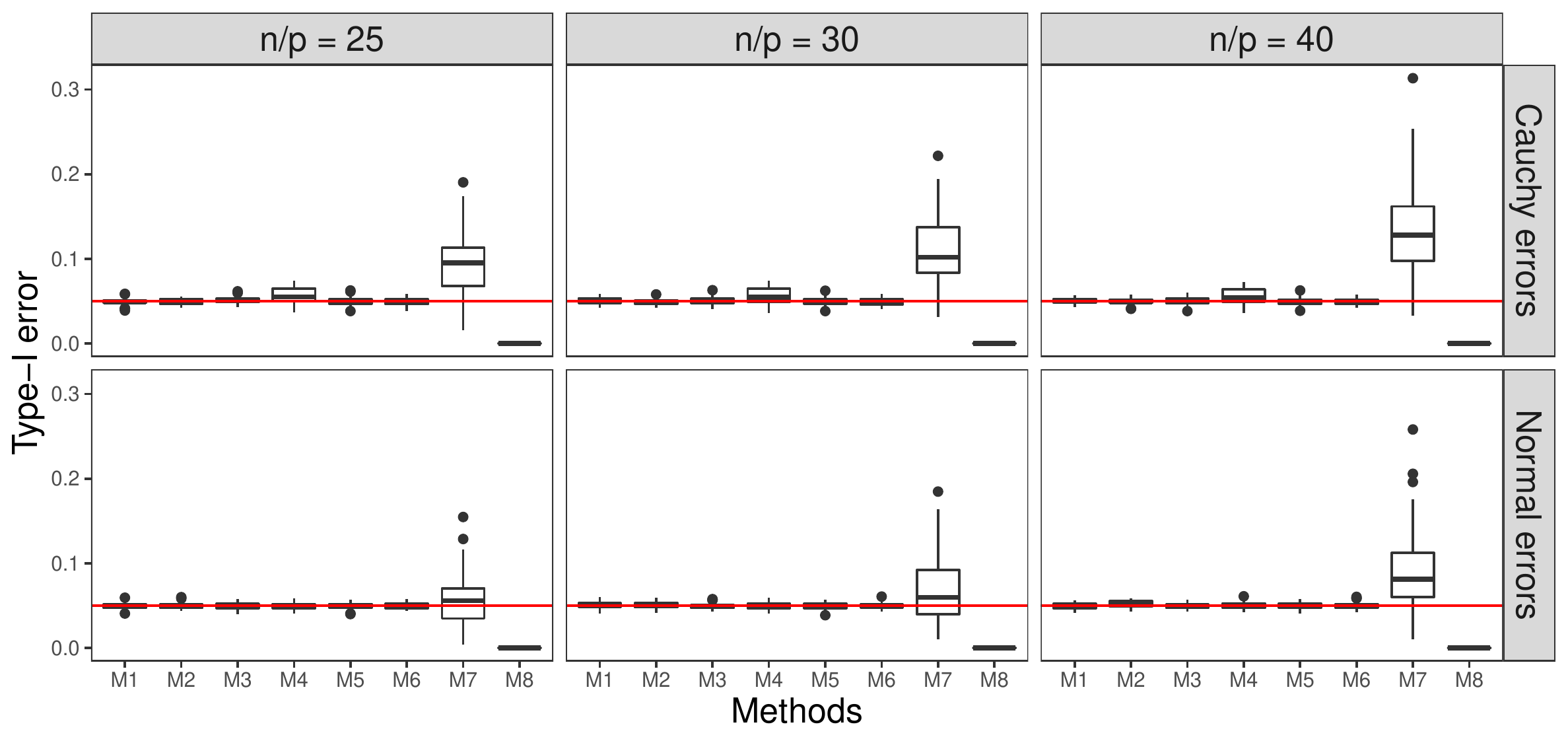}
  \includegraphics[width = 0.8\textwidth]{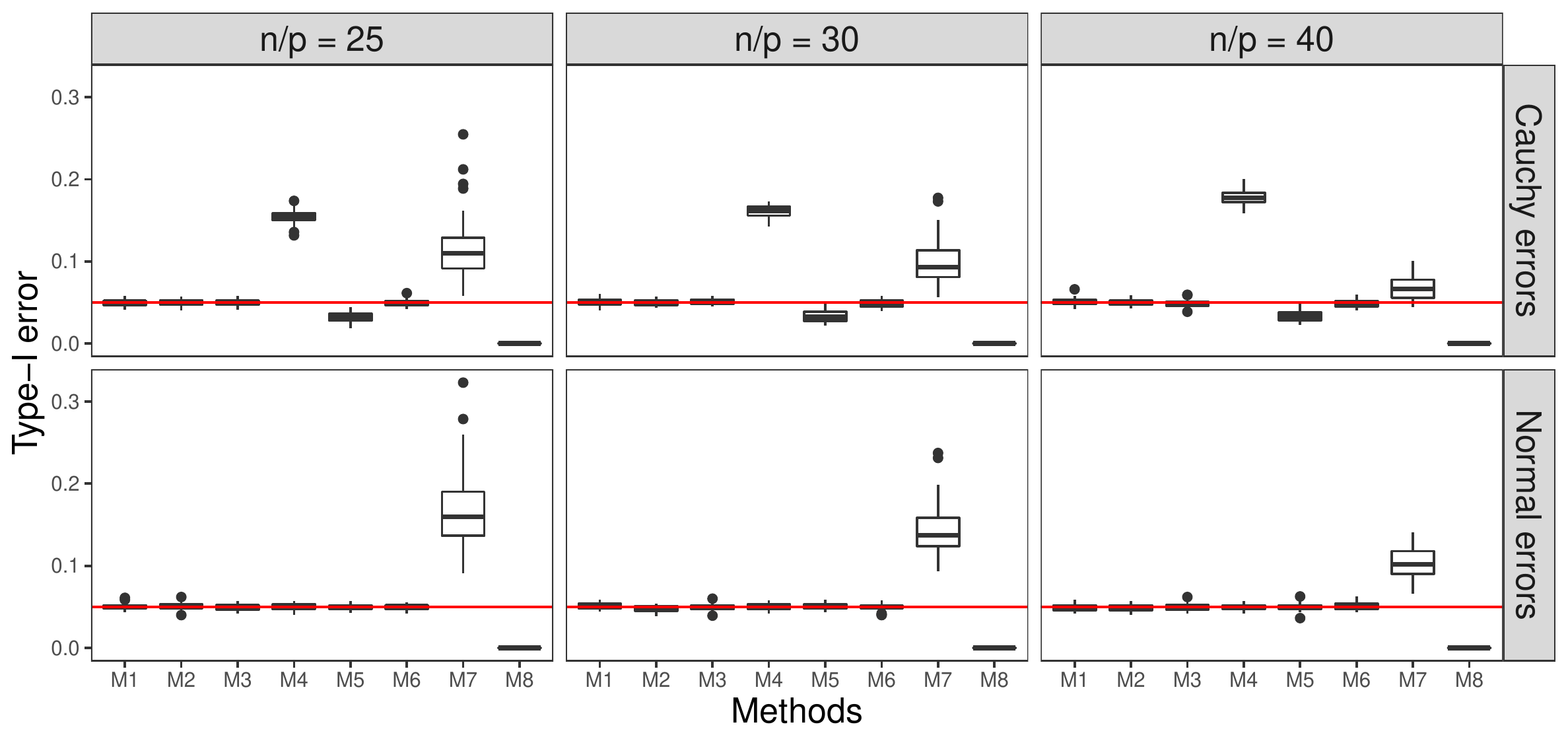}    
\caption{Monte-Carlo Type I error for testing five coordinate with three types of $X$’s which are realizations of (a) random matrices with standard normal entries; (b) random matrices with standard Cauchy
entries; (c) random one-way \ANOVA ~design matrices. Eight methods are compared: M1, \CPT ~with
stronger ordering via the Genetic Algorithm; M2, \CPT ~with weaker ordering via the Genetic Algorithm; M3,
\CPT ~with random ordering via the Stochastic Search; M4, t- or F-test; M5, permutation test; M6, Freedman-Lane test; M7, test based on LAD; M8, GroupBound.}\label{fig:size_r5}
\end{figure}

\begin{figure}[h]
  \centering
  \includegraphics[width = 0.9\textwidth]{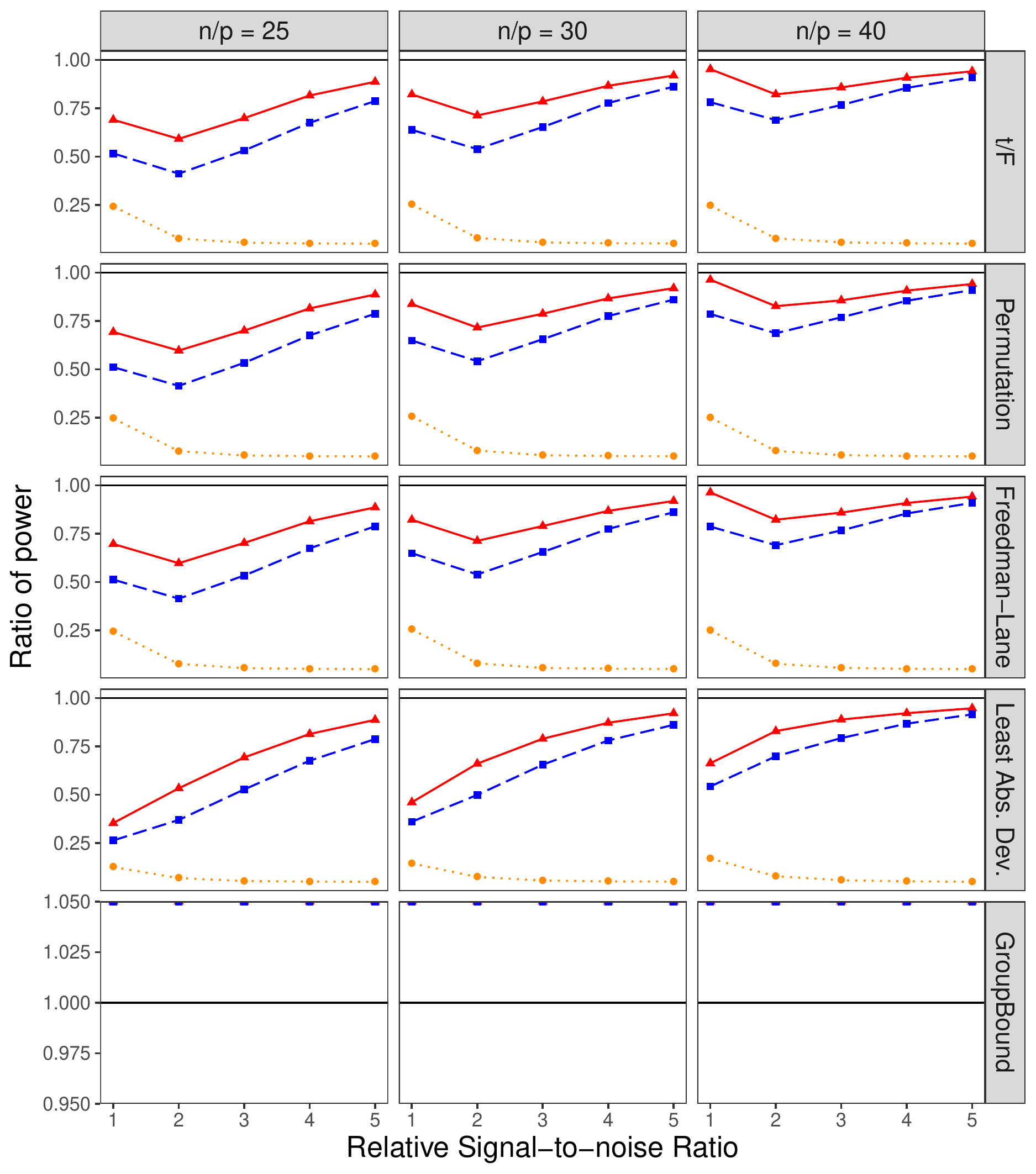}
\caption{Median power ratio of each variant of the cyclic permutation test, one with stronger ordering via
the genetic algorithm (solid), one with weaker ordering via the genetic algorithm (dashed) and one with
random ordering via stochastic search (dotted), to the competing test displayed in each row, for testing five
coordinates in the case with realizations of Gaussian matrices and Gaussian errors. The black solid line
marks equal power. The missing values in the last row correspond to infinite ratios.}\label{fig:power_normal_normal_5}
\end{figure}

\begin{figure}[h]
  \centering
  \includegraphics[width = 0.9\textwidth]{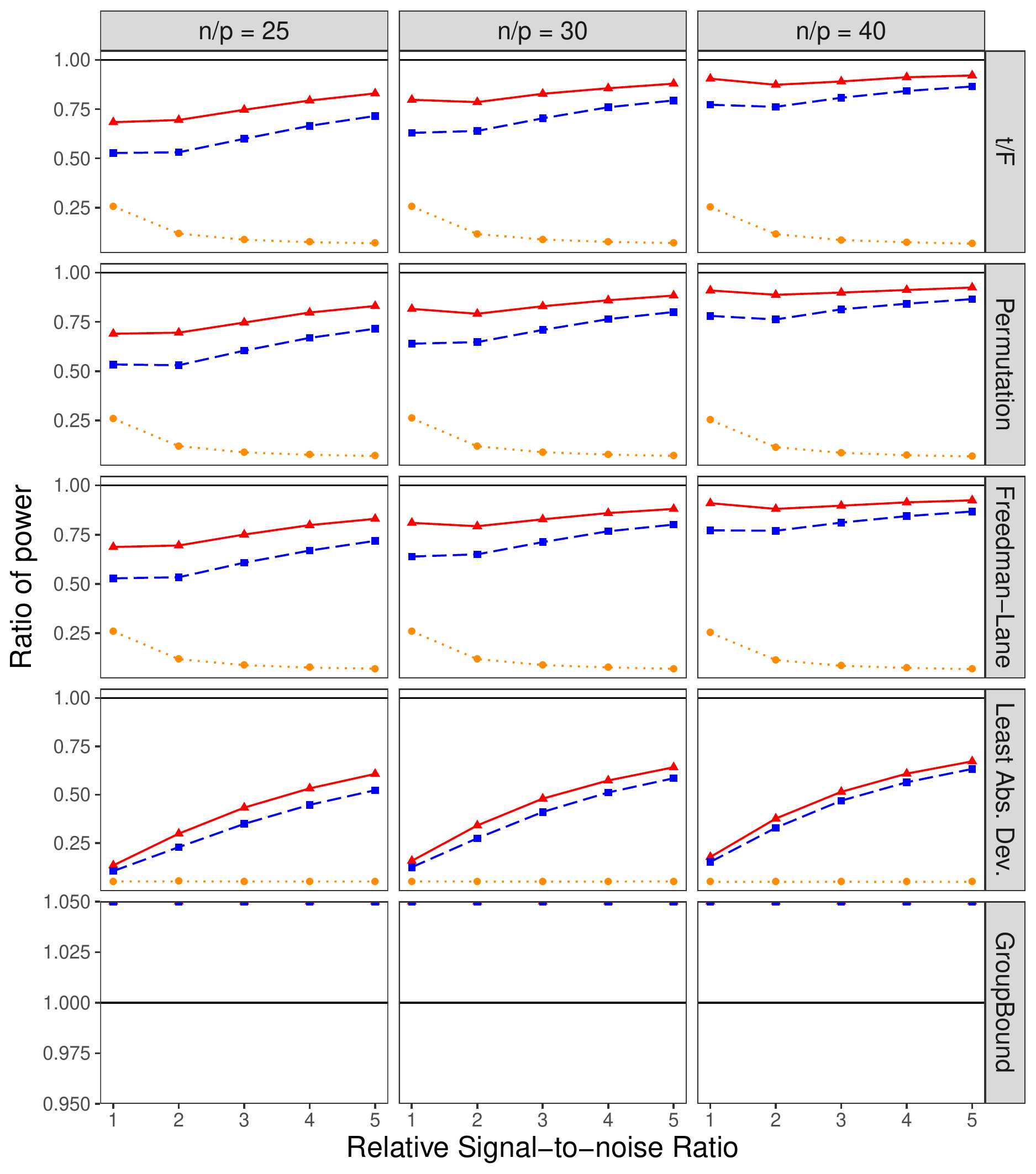}
\caption{Median power ratio of each variant of the cyclic permutation test, one with stronger ordering
via the genetic algorithm (solid), one with weaker ordering via the genetic algorithm (dashed) and one with
random ordering via stochastic search (dotted), to the competing test displayed in each row, for testing five
coordinates in the case with realizations of Gaussian matrices and Cauchy errors. The black solid line marks
equal power. The missing values in the last row correspond to infinite ratios.}\label{fig:power_normal_t1_5}
\end{figure}

\begin{figure}[h]
  \centering
  \includegraphics[width = 0.9\textwidth]{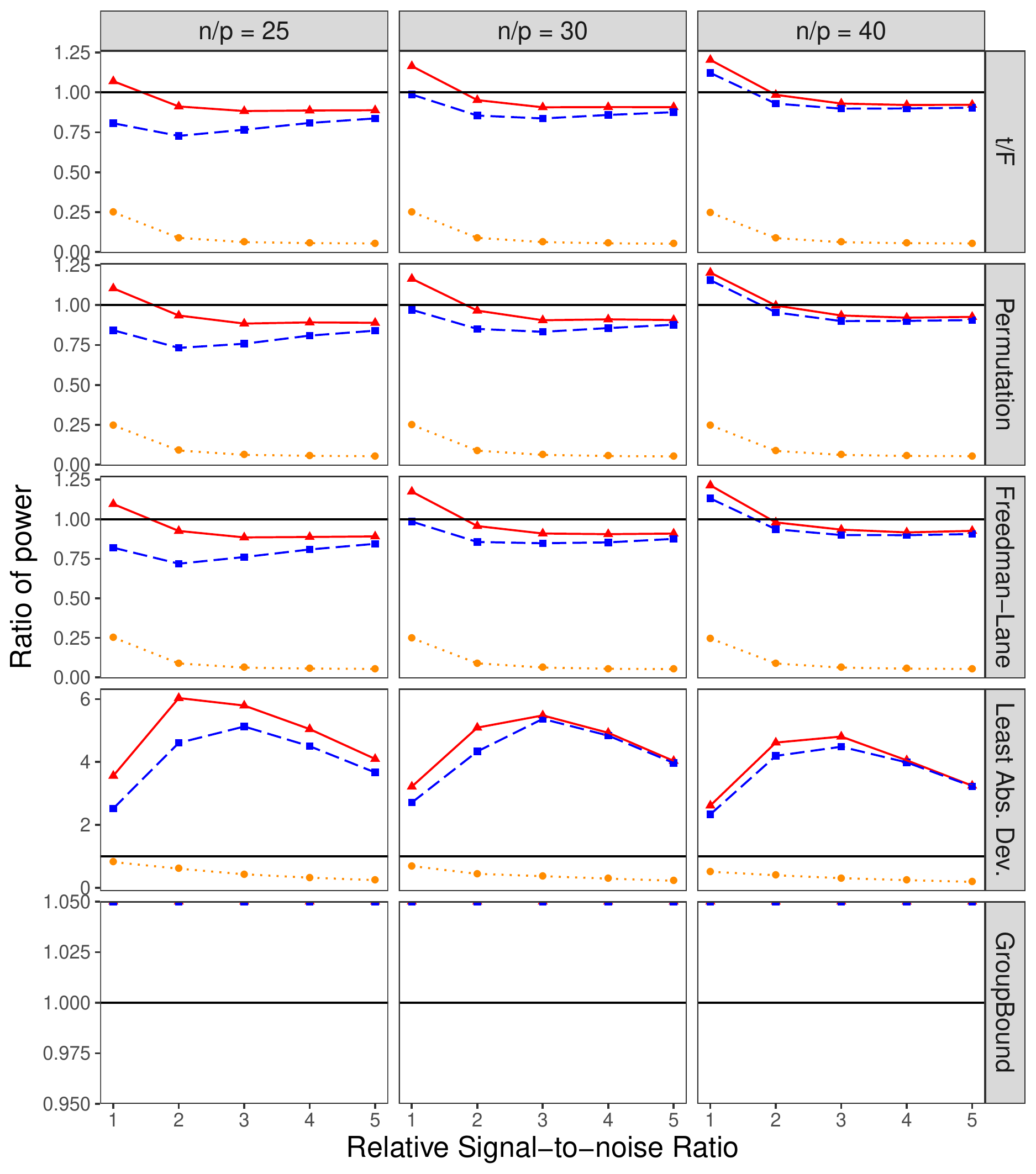}
\caption{Median power ratio of each variant of the cyclic permutation test, one with stronger ordering
via the genetic algorithm (solid), one with weaker ordering via the genetic algorithm (dashed) and one with
random ordering via stochastic search (dotted), to the competing test displayed in each row, for testing five
coordinates in the case with realizations of Cauchy matrices and Gaussian errors. The black solid line marks
equal power. The missing values in the last row correspond to infinite ratios.}\label{fig:power_t1_normal_5}
\end{figure}

\begin{figure}[h]
  \centering
  \includegraphics[width = 0.9\textwidth]{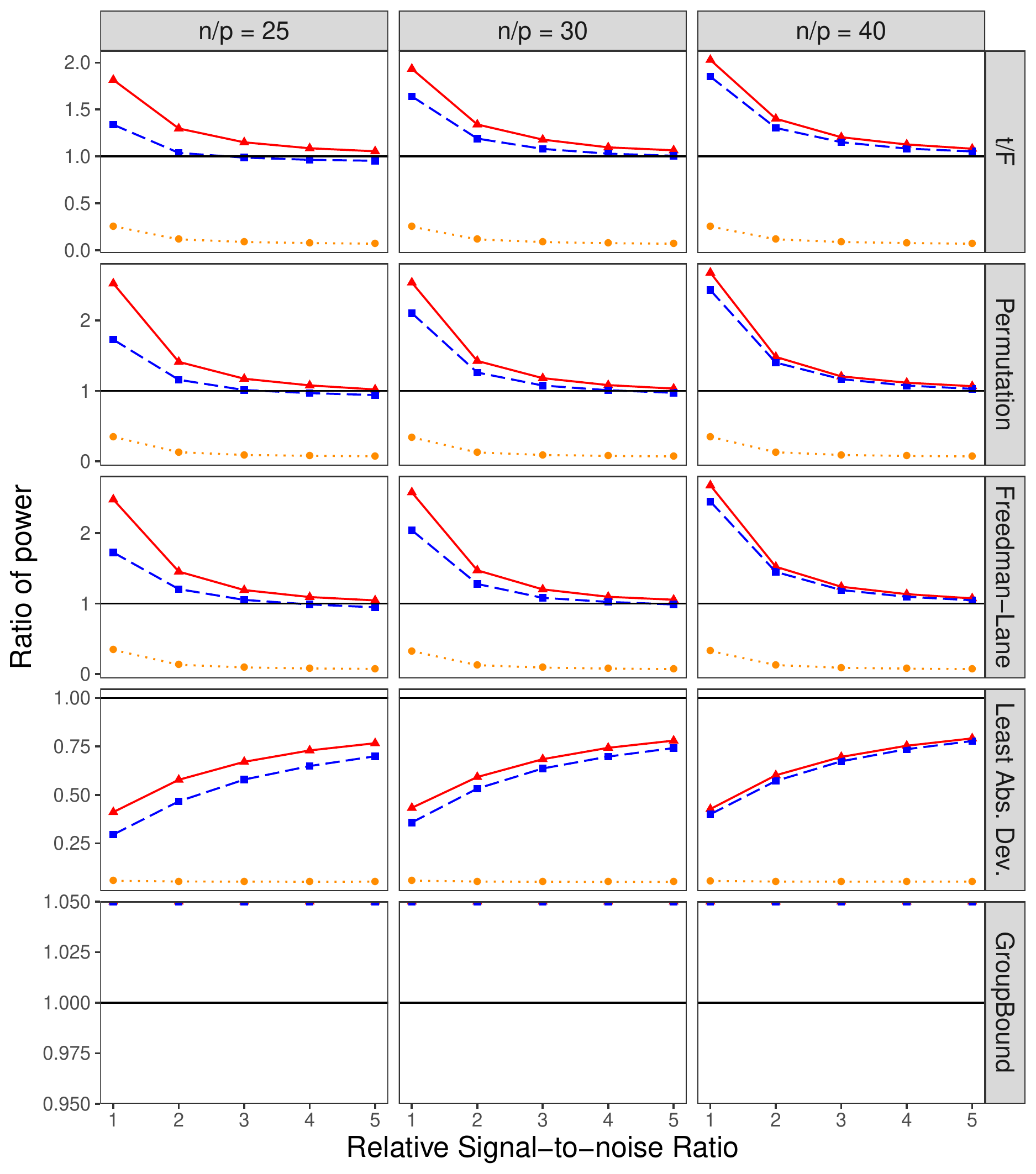}
\caption{Median power ratio of each variant of the cyclic permutation test, one with stronger ordering
via the genetic algorithm (solid), one with weaker ordering via the genetic algorithm (dashed) and one with
random ordering via stochastic search (dotted), to the competing test displayed in each row, for testing five
coordinates in the case with realizations of Cauchy matrices and Cauchy errors. The black solid line marks
equal power. The missing values in the last row correspond to infinite ratios.}\label{fig:power_t1_t1_5}
\end{figure}

\begin{figure}[h]
  \centering
  \includegraphics[width = 0.9\textwidth]{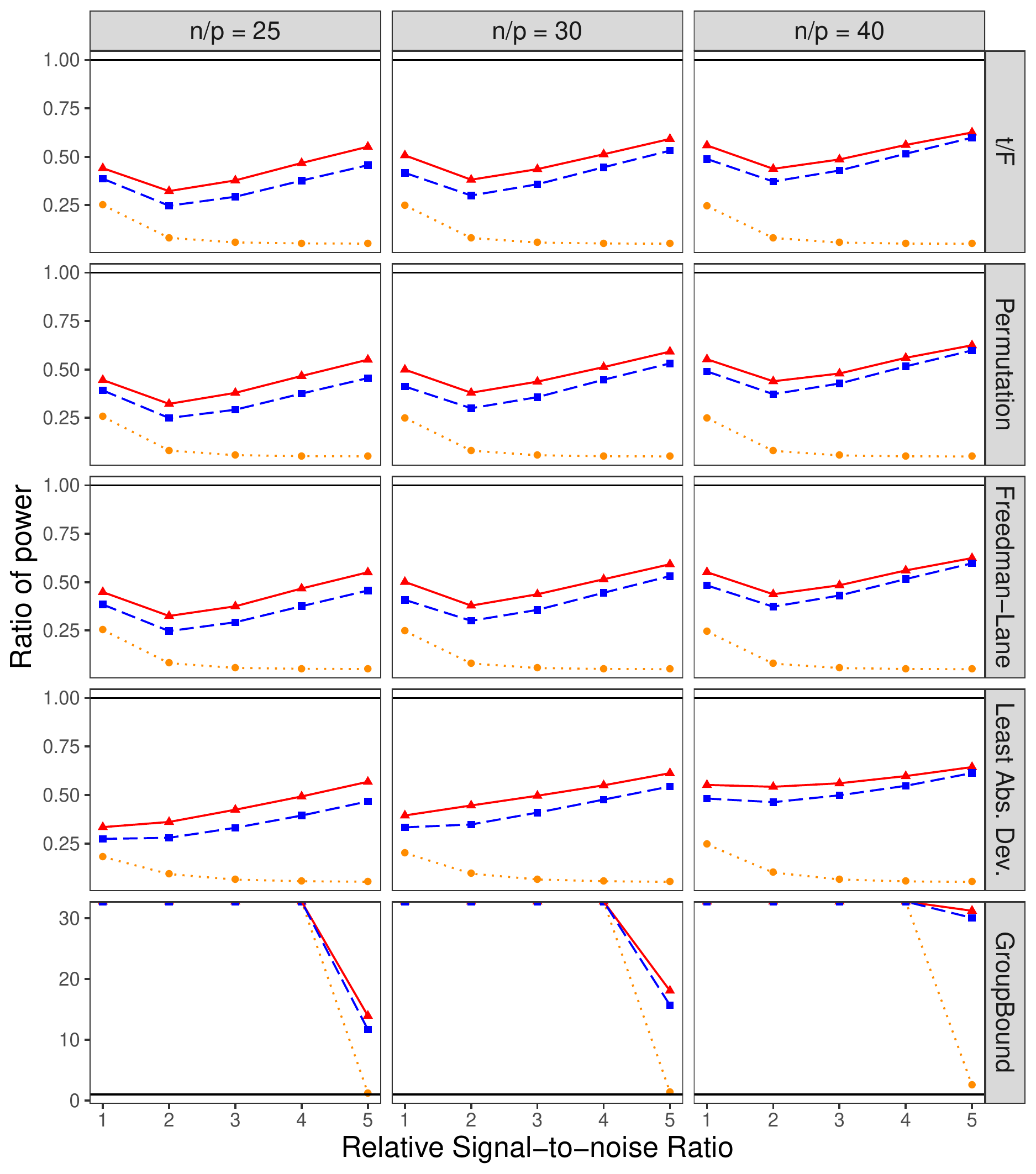}
\caption{Median power ratio of each variant of the cyclic permutation test, one with stronger ordering
via the genetic algorithm (solid), one with weaker ordering via the genetic algorithm (dashed) and one with
random ordering via stochastic search (dotted), to the competing test displayed in each row, for testing five
coordinates in the case with realizations of random one-way ANOVA matrices and Gaussian errors. The
black solid line marks equal power. The missing values in the last row correspond to infinite ratios.}\label{fig:power_ANOVA1_normal_5}
\end{figure}

\begin{figure}[h]
  \centering
  \includegraphics[width = 0.9\textwidth]{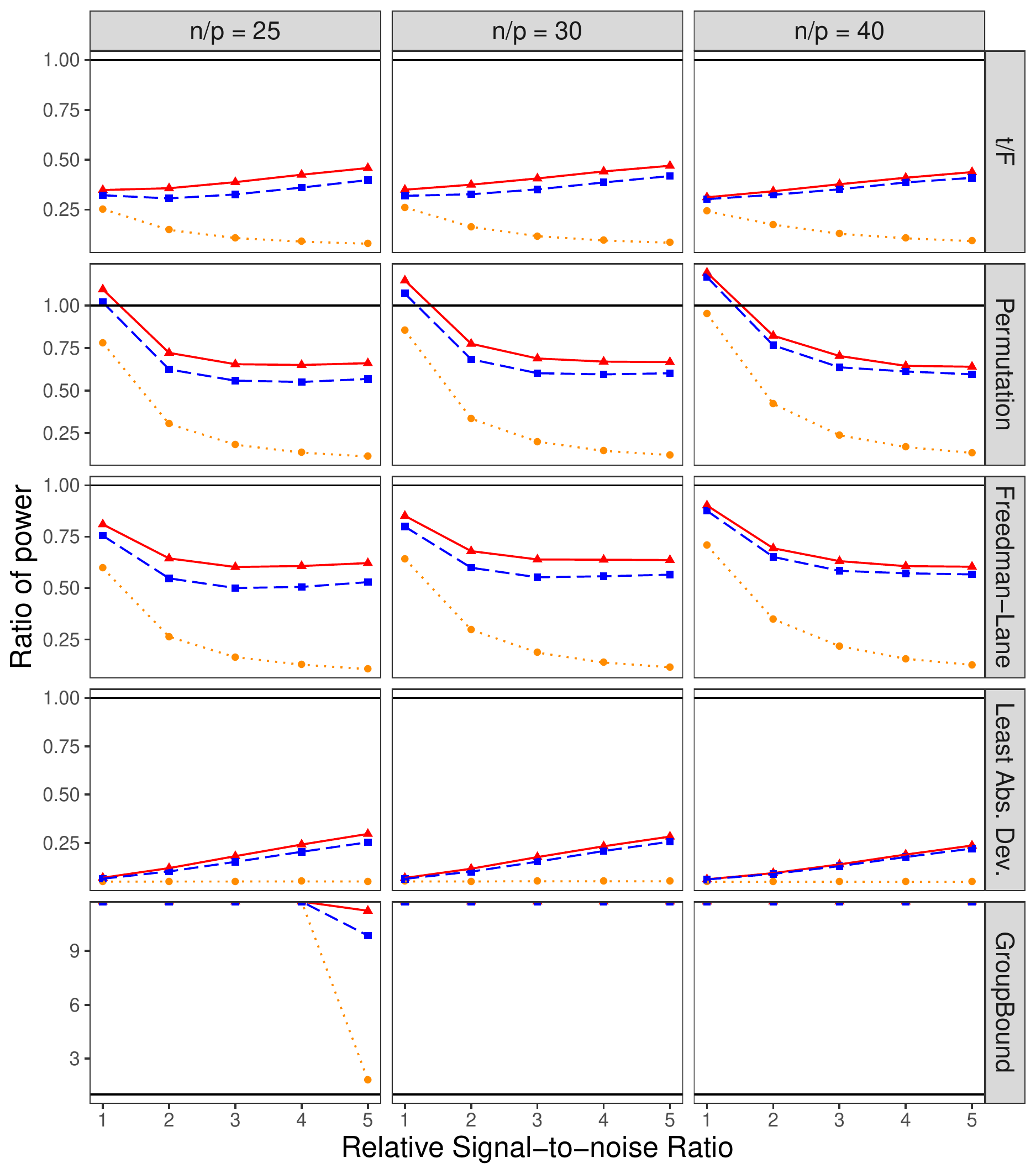}
\caption{Median power ratio of each variant of the cyclic permutation test, one with stronger ordering
via the genetic algorithm (solid), one with weaker ordering via the genetic algorithm (dashed) and one with
random ordering via stochastic search (dotted), to the competing test displayed in each row, for testing five
coordinates in the case with realizations of random one-way ANOVA matrices and Cauchy errors. The black
solid line marks equal power. The missing values in the last row correspond to infinite ratios.}\label{fig:power_ANOVA1_t1_5}
\end{figure}
\end{appendices}

\clearpage
\bibliography{CPT}

\begin{thebibliography}{227}
\expandafter\ifx\csname natexlab\endcsname\relax\def\natexlab#1{#1}\fi

\bibitem[{Adichie(1967{\natexlab{a}})}]{adichie1967asymptotic}
\textsc{Adichie, J.~N.} (1967{\natexlab{a}}).
\newblock Asymptotic efficiency of a class of non-parametric tests for
  regression parameters.
\newblock \textit{The Annals of Mathematical Statistics} , 884--893.

\bibitem[{Adichie(1967{\natexlab{b}})}]{adichie1967estimates}
\textsc{Adichie, J.~N.} (1967{\natexlab{b}}).
\newblock Estimates of regression parameters based on rank tests.
\newblock \textit{The Annals of Mathematical Statistics} , 894--904.

\bibitem[{Adichie(1978)}]{adichie1978rank}
\textsc{Adichie, J.~N.} (1978).
\newblock Rank tests of sub-hypotheses in the general linear regression.
\newblock \textit{The Annals of Statistics} \textbf{6}, 1012--1026.

\bibitem[{Adichie(1984)}]{adichie198411}
\textsc{Adichie, J.~N.} (1984).
\newblock 11 rank tests in linear models.
\newblock \textit{Handbook of statistics} \textbf{4}, 229--257.

\bibitem[{Akritas(1990)}]{akritas1990rank}
\textsc{Akritas, M.~G.} (1990).
\newblock The rank transform method in some two-factor designs.
\newblock \textit{Journal of the American Statistical Association} \textbf{85},
  73--78.

\bibitem[{Akritas \& Arnold(2000)}]{akritas2000asymptotics}
\textsc{Akritas, M.~G.} \& \textsc{Arnold, S.} (2000).
\newblock Asymptotics for analysis of variance when the number of levels is
  large.
\newblock \textit{Journal of the American Statistical association} \textbf{95},
  212--226.

\bibitem[{Akritas \& Arnold(1994)}]{akritas1994fully}
\textsc{Akritas, M.~G.} \& \textsc{Arnold, S.~F.} (1994).
\newblock Fully nonparametric hypotheses for factorial designs i: Multivariate
  repeated measures designs.
\newblock \textit{Journal of the American Statistical Association} \textbf{89},
  336--343.

\bibitem[{Akritas et~al.(1997)Akritas, Arnold \&
  Brunner}]{akritas1997nonparametric}
\textsc{Akritas, M.~G.}, \textsc{Arnold, S.~F.} \& \textsc{Brunner, E.} (1997).
\newblock Nonparametric hypotheses and rank statistics for unbalanced factorial
  designs.
\newblock \textit{Journal of the American Statistical Association} \textbf{92},
  258--265.

\bibitem[{Alimoradi \& Saleh(1998)}]{alimoradi19989}
\textsc{Alimoradi, S.} \& \textsc{Saleh, A. M.~E.} (1998).
\newblock 9 on some {L}-estimation in linear regression models.
\newblock \textit{Handbook of Statistics} \textbf{17}, 237--280.

\bibitem[{Anatolyev(2012)}]{anatolyev2012inference}
\textsc{Anatolyev, S.} (2012).
\newblock Inference in regression models with many regressors.
\newblock \textit{Journal of Econometrics} \textbf{170}, 368--382.

\bibitem[{Anderson \& Robinson(2001)}]{anderson2001permutation}
\textsc{Anderson, M.~J.} \& \textsc{Robinson, J.} (2001).
\newblock Permutation tests for linear models.
\newblock \textit{Australian \& New Zealand Journal of Statistics} \textbf{43},
  75--88.

\bibitem[{Arnold(1980)}]{arnold1980asymptotic}
\textsc{Arnold, S.~F.} (1980).
\newblock Asymptotic validity of f tests for the ordinary linear model and the
  multiple correlation model.
\newblock \textit{Journal of the American Statistical Association} \textbf{75},
  890--894.

\bibitem[{Aubuchon \& Hettmansperger(1984)}]{aubuchon198412}
\textsc{Aubuchon, J.~C.} \& \textsc{Hettmansperger, T.~P.} (1984).
\newblock 12 on the use of rank tests and estimates in the linear model.
\newblock \textit{Handbook of statistics} \textbf{4}, 259--274.

\bibitem[{Bai \& Wu(1994)}]{bai1994limiting}
\textsc{Bai, Z.} \& \textsc{Wu, Y.} (1994).
\newblock Limiting behavior of {M}-estimators of regression coefficients in
  high dimensional linear models i. scale dependent case.
\newblock \textit{Journal of Multivariate Analysis} \textbf{51}, 211--239.

\bibitem[{Barber \& Cand{\`e}s(2015)}]{barber2015controlling}
\textsc{Barber, R.~F.} \& \textsc{Cand{\`e}s, E.~J.} (2015).
\newblock Controlling the false discovery rate via knockoffs.
\newblock \textit{The Annals of Statistics} \textbf{43}, 2055--2085.

\bibitem[{Bartlett(1935)}]{bartlett1935effect}
\textsc{Bartlett, M.} (1935).
\newblock The effect of non-normality on the t distribution.
\newblock In \textit{mathematical proceedings of the cambridge philosophical
  society}, vol.~31. Cambridge University Press.

\bibitem[{Bathke \& Lankowski(2005)}]{bathke2005rank}
\textsc{Bathke, A.} \& \textsc{Lankowski, D.} (2005).
\newblock Rank procedures for a large number of treatments.
\newblock \textit{Journal of statistical planning and inference} \textbf{133},
  223--238.

\bibitem[{Bathke \& Harrar(2008)}]{bathke2008nonparametric}
\textsc{Bathke, A.~C.} \& \textsc{Harrar, S.~W.} (2008).
\newblock Nonparametric methods in multivariate factorial designs for large
  number of factor levels.
\newblock \textit{Journal of Statistical planning and Inference} \textbf{138},
  588--610.

\bibitem[{Bean et~al.(2012)Bean, Bickel, El~Karoui, Lim \& Yu}]{bean11}
\textsc{Bean, D.}, \textsc{Bickel, P.~J.}, \textsc{El~Karoui, N.}, \textsc{Lim,
  C.} \& \textsc{Yu, B.} (2012).
\newblock Penalized robust regression in high-dimension.
\newblock \textit{Technical Report 813, Department of Statistics, UC Berkeley}
  .

\bibitem[{Benjamini(1983)}]{benjamini1983t}
\textsc{Benjamini, Y.} (1983).
\newblock Is the t test really conservative when the parent distribution is
  long-tailed?
\newblock \textit{Journal of the American Statistical Association} \textbf{78},
  645--654.

\bibitem[{Berry et~al.(2013)Berry, Johnston \& Mielke}]{berry2013}
\textsc{Berry, K.~J.}, \textsc{Johnston, J.~E.} \& \textsc{Mielke, P.~W.}
  (2013).
\newblock \textit{A chronicle of permutation statistical methods. 1920-2000,
  and beyond}.
\newblock Springer.

\bibitem[{Bhattacharya \& Ghosh(1978)}]{bhattacharya1978validity}
\textsc{Bhattacharya, R.~N.} \& \textsc{Ghosh, J.~K.} (1978).
\newblock On the validity of the formal {Edgeworth} expansion.
\newblock \textit{Ann. Statist} \textbf{6}, 434--451.

\bibitem[{Bickel(1965)}]{bickel1965some}
\textsc{Bickel, P.~J.} (1965).
\newblock On some robust estimates of location.
\newblock \textit{The Annals of Mathematical Statistics} \textbf{36}, 847--858.

\bibitem[{Bickel(1973)}]{bickel1973some}
\textsc{Bickel, P.~J.} (1973).
\newblock On some analogues to linear combinations of order statistics in the
  linear model.
\newblock \textit{The Annals of Statistics} , 597--616.

\bibitem[{Bickel(1975)}]{bickel1975one}
\textsc{Bickel, P.~J.} (1975).
\newblock One-step huber estimates in the linear model.
\newblock \textit{Journal of the American Statistical Association} \textbf{70},
  428--434.

\bibitem[{Bickel \& Freedman(1983)}]{bickel1983bootstrapping}
\textsc{Bickel, P.~J.} \& \textsc{Freedman, D.~A.} (1983).
\newblock Bootstrapping regression models with many parameters.
\newblock \textit{Festschrift for Erich L. Lehmann} , 28--48.

\bibitem[{Bickel \& Lehmann(1975)}]{bickel1975descriptive}
\textsc{Bickel, P.~J.} \& \textsc{Lehmann, E.~L.} (1975).
\newblock Descriptive statistics for nonparametric models {II}. location.
\newblock \textit{The Annals of Statistics} \textbf{3}, 1045--1069.

\bibitem[{Bickel \& Sakov(2008)}]{bickel2008choice}
\textsc{Bickel, P.~J.} \& \textsc{Sakov, A.} (2008).
\newblock On the choice of $m$ in the $m$ out of $n$ bootstrap and confidence
  bounds for extrema.
\newblock \textit{Statistica Sinica} \textbf{18}, 967--985.

\bibitem[{Boos(1992)}]{boos1992generalized}
\textsc{Boos, D.~D.} (1992).
\newblock On generalized score tests.
\newblock \textit{The American Statistician} \textbf{46}, 327--333.

\bibitem[{Boos \& Brownie(1995)}]{boos1995anova}
\textsc{Boos, D.~D.} \& \textsc{Brownie, C.} (1995).
\newblock {ANOVA} and rank tests when the number of treatments is large.
\newblock \textit{Statistics \& Probability Letters} \textbf{23}, 183--191.

\bibitem[{Box(1953)}]{box1953non}
\textsc{Box, G.~E.} (1953).
\newblock Non-normality and tests on variances.
\newblock \textit{Biometrika} \textbf{40}, 318--335.

\bibitem[{Box \& Andersen(1955)}]{box1955permutation}
\textsc{Box, G.~E.} \& \textsc{Andersen, S.~L.} (1955).
\newblock Permutation theory in the derivation of robust criteria and the study
  of departures from assumption.
\newblock \textit{Journal of the Royal Statistical Society: Series B
  (Methodological)} \textbf{17}, 1--26.

\bibitem[{Box \& Watson(1962)}]{box1962robustness}
\textsc{Box, G.~E.} \& \textsc{Watson, G.~S.} (1962).
\newblock Robustness to non-normality of regression tests.
\newblock \textit{Biometrika} \textbf{49}, 93--106.

\bibitem[{Brown \& Maritz(1982)}]{brown1982distribution}
\textsc{Brown, B.} \& \textsc{Maritz, J.} (1982).
\newblock Distribution-free methods in regression.
\newblock \textit{Australian Journal of Statistics} \textbf{24}, 318--331.

\bibitem[{Brown \& Mood(1951)}]{brown1951median}
\textsc{Brown, G.~W.} \& \textsc{Mood, A.~M.} (1951).
\newblock On median tests for linear hypotheses.
\newblock In \textit{Proceedings of the Second Berkeley Symposium on
  Mathematical Statistics and Probability}. The Regents of the University of
  California.

\bibitem[{Brownie \& Boos(1994)}]{brownie1994type}
\textsc{Brownie, C.} \& \textsc{Boos, D.~D.} (1994).
\newblock Type {I} error robustness of {ANOVA} and {ANOVA} on ranks when the
  number of treatments is large.
\newblock \textit{Biometrics} , 542--549.

\bibitem[{Brunner \& Denker(1994)}]{brunner1994rank}
\textsc{Brunner, E.} \& \textsc{Denker, M.} (1994).
\newblock Rank statistics under dependent observations and applications to
  factorial designs.
\newblock \textit{Journal of Statistical planning and Inference} \textbf{42},
  353--378.

\bibitem[{Calhoun(2011)}]{calhoun2011hypothesis}
\textsc{Calhoun, G.} (2011).
\newblock Hypothesis testing in linear regression when $k/n$ is large.
\newblock \textit{Journal of econometrics} \textbf{165}, 163--174.

\bibitem[{Campi et~al.(2009)Campi, Ko \& Weyer}]{campi2009non}
\textsc{Campi, M.~C.}, \textsc{Ko, S.} \& \textsc{Weyer, E.} (2009).
\newblock Non-asymptotic confidence regions for model parameters in the
  presence of unmodelled dynamics.
\newblock \textit{Automatica} \textbf{45}, 2175--2186.

\bibitem[{Campi \& Weyer(2005)}]{campi2005guaranteed}
\textsc{Campi, M.~C.} \& \textsc{Weyer, E.} (2005).
\newblock Guaranteed non-asymptotic confidence regions in system
  identification.
\newblock \textit{Automatica} \textbf{41}, 1751--1764.

\bibitem[{Cattaneo et~al.(2018)Cattaneo, Jansson \&
  Newey}]{cattaneo2018inference}
\textsc{Cattaneo, M.~D.}, \textsc{Jansson, M.} \& \textsc{Newey, W.~K.} (2018).
\newblock Inference in linear regression models with many covariates and
  heteroscedasticity.
\newblock \textit{Journal of the American Statistical Association}
  \textbf{113}, 1350--1361.

\bibitem[{Chatterjee(1999)}]{chatterjee1999generalised}
\textsc{Chatterjee, S.~B.} (1999).
\newblock \textit{Generalised bootstrap techniques}.
\newblock Ph.D. thesis, Indian Statistical Institute, Kolkata.

\bibitem[{Chernoff(1956)}]{chernoff1956large}
\textsc{Chernoff, H.} (1956).
\newblock Large-sample theory: Parametric case.
\newblock \textit{The Annals of Mathematical Statistics} \textbf{27}, 1--22.

\bibitem[{Chernozhukov et~al.(2009)Chernozhukov, Hansen \&
  Jansson}]{chernozhukov2009finite}
\textsc{Chernozhukov, V.}, \textsc{Hansen, C.} \& \textsc{Jansson, M.} (2009).
\newblock Finite sample inference for quantile regression models.
\newblock \textit{Journal of Econometrics} \textbf{152}, 93--103.

\bibitem[{Chung \& Romano(2013)}]{chung2013exact}
\textsc{Chung, E.} \& \textsc{Romano, J.~P.} (2013).
\newblock Exact and asymptotically robust permutation tests.
\newblock \textit{The Annals of Statistics} \textbf{41}, 484--507.

\bibitem[{Cochran(1937)}]{cochran1937efficiencies}
\textsc{Cochran, W.~G.} (1937).
\newblock The efficiencies of the binomial series tests of significance of a
  mean and of a correlation coefficient.
\newblock \textit{Journal of the Royal Statistical Society} \textbf{100},
  69--73.

\bibitem[{Collins(1987)}]{collins1987permutation}
\textsc{Collins, M.~F.} (1987).
\newblock A permutation test for planar regression.
\newblock \textit{Australian Journal of Statistics} \textbf{29}, 303--308.

\bibitem[{Conover \& Iman(1976)}]{conover1976some}
\textsc{Conover, W.} \& \textsc{Iman, R.~L.} (1976).
\newblock On some alternative procedures using ranks for the analysis of
  experimental designs.
\newblock \textit{Communications in Statistics-Theory and Methods} \textbf{5},
  1349--1368.

\bibitem[{Conover \& Iman(1981)}]{conover1981rank}
\textsc{Conover, W.~J.} \& \textsc{Iman, R.~L.} (1981).
\newblock Rank transformations as a bridge between parametric and nonparametric
  statistics.
\newblock \textit{The American Statistician} \textbf{35}, 124--129.

\bibitem[{Cressie(1980)}]{cressie1980relaxing}
\textsc{Cressie, N.} (1980).
\newblock Relaxing assumptions in the one sample t-test.
\newblock \textit{Australian Journal of Statistics} \textbf{22}, 143--153.

\bibitem[{Daniels(1954)}]{daniels1954distribution}
\textsc{Daniels, H.} (1954).
\newblock A distribution-free test for regression parameters.
\newblock \textit{The Annals of Mathematical Statistics} , 499--513.

\bibitem[{Das \& Lahiri(2019)}]{das2019second}
\textsc{Das, D.} \& \textsc{Lahiri, S.~N.} (2019).
\newblock Second order correctness of perturbation bootstrap {M}-estimator of
  multiple linear regression parameter.
\newblock \textit{Bernoulli} \textbf{25}, 654--682.

\bibitem[{David \& Johnson(1951{\natexlab{a}})}]{david1951method}
\textsc{David, F.} \& \textsc{Johnson, N.} (1951{\natexlab{a}}).
\newblock A method of investigating the effect of nonnormality and
  heterogeneity of variance on tests of the general linear hypothesis.
\newblock \textit{The Annals of Mathematical Statistics} , 382--392.

\bibitem[{David \& Johnson(1951{\natexlab{b}})}]{david1951effect}
\textsc{David, F.~N.} \& \textsc{Johnson, N.} (1951{\natexlab{b}}).
\newblock The effect of non-normality on the power function of the {F}-test in
  the analysis of variance.
\newblock \textit{Biometrika} \textbf{38}, 43--57.

\bibitem[{DiCiccio \& Romano(2017)}]{diciccio2017robust}
\textsc{DiCiccio, C.~J.} \& \textsc{Romano, J.~P.} (2017).
\newblock Robust permutation tests for correlation and regression coefficients.
\newblock \textit{Journal of the American Statistical Association}
  \textbf{112}, 1211--1220.

\bibitem[{Donoho \& Montanari(2016)}]{donoho16}
\textsc{Donoho, D.} \& \textsc{Montanari, A.} (2016).
\newblock High dimensional robust {m}-estimation: Asymptotic variance via
  approximate message passing.
\newblock \textit{Probability Theory and Related Fields} \textbf{166},
  935--969.

\bibitem[{Doob(1935)}]{doob1935limiting}
\textsc{Doob, J.~L.} (1935).
\newblock The limiting distributions of certain statistics.
\newblock \textit{The Annals of Mathematical Statistics} \textbf{6}, 160--169.

\bibitem[{Draper(1988)}]{draper1988rank}
\textsc{Draper, D.} (1988).
\newblock Rank-based robust analysis of linear models. {I}. exposition and
  review.
\newblock \textit{Statistical Science} , 239--257.

\bibitem[{Eden \& Yates(1933)}]{eden1933validity}
\textsc{Eden, T.} \& \textsc{Yates, F.} (1933).
\newblock On the validity of {Fisher's} z test when applied to an actual
  example of non-normal data.
\newblock \textit{The Journal of Agricultural Science} \textbf{23}, 6--17.

\bibitem[{Efron(1969)}]{efron69}
\textsc{Efron, B.} (1969).
\newblock Student's t-test under symmetry conditions.
\newblock \textit{Journal of the American Statistical Association} \textbf{64},
  1278--1302.

\bibitem[{Efron(1979)}]{efron1979bootstrap}
\textsc{Efron, B.} (1979).
\newblock Bootstrap methods: Another look at the jackknife.
\newblock \textit{The Annals of Statistics} \textbf{7}, 1--26.

\bibitem[{Eicker(1963)}]{eicker1963asymptotic}
\textsc{Eicker, F.} (1963).
\newblock Asymptotic normality and consistency of the least squares estimators
  for families of linear regressions.
\newblock \textit{The Annals of Mathematical Statistics} \textbf{34}, 447--456.

\bibitem[{Eicker(1967)}]{eicker1967limit}
\textsc{Eicker, F.} (1967).
\newblock Limit theorems for regressions with unequal and dependent errors.
\newblock In \textit{Proceedings of the fifth Berkeley symposium on
  mathematical statistics and probability}, vol.~1.

\bibitem[{El~Karoui(2013)}]{elkaroui13}
\textsc{El~Karoui, N.} (2013).
\newblock Asymptotic behavior of unregularized and ridge-regularized
  high-dimensional robust regression estimators: rigorous results.
\newblock \textit{arXiv preprint arXiv:1311.2445} .

\bibitem[{El~Karoui(2018)}]{elkaroui2018impact}
\textsc{El~Karoui, N.} (2018).
\newblock On the impact of predictor geometry on the performance on
  high-dimensional ridge-regularized generalized robust regression estimators.
\newblock \textit{Probability Theory and Related Fields} \textbf{170}, 95--175.

\bibitem[{El~Karoui et~al.(2011)El~Karoui, Bean, Bickel, Lim \&
  Yu}]{elkaroui11}
\textsc{El~Karoui, N.}, \textsc{Bean, D.}, \textsc{Bickel, P.~J.}, \textsc{Lim,
  C.} \& \textsc{Yu, B.} (2011).
\newblock On robust regression with high-dimensional predictors.
\newblock \textit{Technical Report 811, Department of Statistics, UC Berkeley}
  .

\bibitem[{El~Karoui \& Purdom(2018)}]{el2018can}
\textsc{El~Karoui, N.} \& \textsc{Purdom, E.} (2018).
\newblock Can we trust the bootstrap in high-dimensions? the case of linear
  models.
\newblock \textit{The Journal of Machine Learning Research} \textbf{19},
  170--235.

\bibitem[{Esseen(1945)}]{esseen1945fourier}
\textsc{Esseen, C.-G.} (1945).
\newblock Fourier analysis of distribution functions. a mathematical study of
  the laplace-gaussian law.
\newblock \textit{Acta Mathematica} \textbf{77}, 1--125.

\bibitem[{Evans \& Evans(1955)}]{evans1955atomic}
\textsc{Evans, R.~D.} \& \textsc{Evans, R.} (1955).
\newblock \textit{Appendix {G}: The atomic nucleus}.
\newblock McGraw-Hill New York.

\bibitem[{Feng et~al.(2013)Feng, Zou, Wang \& Chen}]{feng2013rank}
\textsc{Feng, L.}, \textsc{Zou, C.}, \textsc{Wang, Z.} \& \textsc{Chen, B.}
  (2013).
\newblock Rank-based score tests for high-dimensional regression coefficients.
\newblock \textit{Electronic Journal of Statistics} \textbf{7}, 2131--2149.

\bibitem[{Feng et~al.(2011)Feng, He \& Hu}]{feng2011wild}
\textsc{Feng, X.}, \textsc{He, X.} \& \textsc{Hu, J.} (2011).
\newblock Wild bootstrap for quantile regression.
\newblock \textit{Biometrika} \textbf{98}, 995--999.

\bibitem[{Fisher(1915)}]{fisher1915frequency}
\textsc{Fisher, R.~A.} (1915).
\newblock Frequency distribution of the values of the correlation coefficient
  in samples from an indefinitely large population.
\newblock \textit{Biometrika} \textbf{10}, 507--521.

\bibitem[{Fisher(1922)}]{fisher1922goodness}
\textsc{Fisher, R.~A.} (1922).
\newblock The goodness of fit of regression formulae, and the distribution of
  regression coefficients.
\newblock \textit{Journal of the Royal Statistical Society} \textbf{85},
  597--612.

\bibitem[{Fisher(1924)}]{fisher1924036}
\textsc{Fisher, R.~A.} (1924).
\newblock 036: On a distribution yielding the error functions of several well
  known statistics.
\newblock \textit{Proceedings of the International Congress of Mathematics}
  \textbf{2}, 805--813.

\bibitem[{Fisher(1925)}]{fisher1925statistical}
\textsc{Fisher, R.~A.} (1925).
\newblock \textit{Statistical methods for research workers.}
\newblock Oliver and Boyd, Edinburgh and London.

\bibitem[{Fisher(1926)}]{fisher1926arrangement}
\textsc{Fisher, R.~A.} (1926).
\newblock The arrangement of field experiments.
\newblock \textit{Journal of the Ministry of Agriculture} \textbf{33},
  503--513.

\bibitem[{Fisher(1935)}]{fisher1935logic}
\textsc{Fisher, R.~A.} (1935).
\newblock The logic of inductive inference.
\newblock \textit{Journal of the royal statistical society} \textbf{98},
  39--82.

\bibitem[{Fogel et~al.(2013)Fogel, Jenatton, Bach \&
  d'Aspremont}]{fogel2013convex}
\textsc{Fogel, F.}, \textsc{Jenatton, R.}, \textsc{Bach, F.} \&
  \textsc{d'Aspremont, A.} (2013).
\newblock Convex relaxations for permutation problems.
\newblock In \textit{Advances in Neural Information Processing Systems}.

\bibitem[{Freedman(1981)}]{freedman1981bootstrapping}
\textsc{Freedman, D.~A.} (1981).
\newblock Bootstrapping regression models.
\newblock \textit{The Annals of Statistics} \textbf{9}, 1218--1228.

\bibitem[{Freedman \& Lane(1983)}]{freedman83}
\textsc{Freedman, D.~A.} \& \textsc{Lane, D.} (1983).
\newblock A nonstochastic interpretation of reported significance levels.
\newblock \textit{Journal of Business \& Economic Statistics} \textbf{1},
  292--298.

\bibitem[{Friedman(1937)}]{friedman1937use}
\textsc{Friedman, M.} (1937).
\newblock The use of ranks to avoid the assumption of normality implicit in the
  analysis of variance.
\newblock \textit{Journal of the american statistical association} \textbf{32},
  675--701.

\bibitem[{Friedman(1940)}]{friedman1940comparison}
\textsc{Friedman, M.} (1940).
\newblock A comparison of alternative tests of significance for the problem of
  m rankings.
\newblock \textit{The Annals of Mathematical Statistics} \textbf{11}, 86--92.

\bibitem[{Gastwirth(1966)}]{gastwirth1966robust}
\textsc{Gastwirth, J.~L.} (1966).
\newblock On robust procedures.
\newblock \textit{Journal of the American Statistical Association} \textbf{61},
  929--948.

\bibitem[{Gayen(1949)}]{gayen1949distribution}
\textsc{Gayen, A.~K.} (1949).
\newblock The distribution of {Student's} t in random samples of any size drawn
  from non-normal universes.
\newblock \textit{Biometrika} \textbf{36}, 353--369.

\bibitem[{Gayen(1950)}]{gayen1950distribution}
\textsc{Gayen, A.~K.} (1950).
\newblock The distribution of the variance ratio in random samples of any size
  drawn from non-normal universes.
\newblock \textit{Biometrika} \textbf{37}, 236--255.

\bibitem[{Geary(1927)}]{geary1927some}
\textsc{Geary, R.} (1927).
\newblock Some properties of correlation and regression in a limited universe.
\newblock \textit{Metron} \textbf{7}, 83--119.

\bibitem[{Geary(1947)}]{geary1947testing}
\textsc{Geary, R.~C.} (1947).
\newblock Testing for normality.
\newblock \textit{Biometrika} \textbf{34}, 209--242.

\bibitem[{Gutenbrunner \& Jureckova(1992)}]{gutenbrunner1992regression}
\textsc{Gutenbrunner, C.} \& \textsc{Jureckova, J.} (1992).
\newblock Regression quantile and regression rank score process in the linear
  model and derived statistics.
\newblock \textit{Annals of Statistics} \textbf{20}, 305--330.

\bibitem[{Gutenbrunner et~al.(1993)Gutenbrunner, Jureckova, Koenker \&
  Portnoy}]{gutenbrunner1993tests}
\textsc{Gutenbrunner, C.}, \textsc{Jureckova, J.}, \textsc{Koenker, R.} \&
  \textsc{Portnoy, S.} (1993).
\newblock Tests of linear hypotheses based on regression rank scores.
\newblock \textit{Journaltitle of Nonparametric Statistics} \textbf{2},
  307--331.

\bibitem[{H{\'a}jek(1962)}]{hajek1962asymptotically}
\textsc{H{\'a}jek, J.} (1962).
\newblock Asymptotically most powerful rank-order tests.
\newblock \textit{The Annals of Mathematical Statistics} , 1124--1147.

\bibitem[{Hajek \& Sidak(1967)}]{hajek1967theory}
\textsc{Hajek, J.} \& \textsc{Sidak, Z.} (1967).
\newblock Theory of rank tests. academia.

\bibitem[{Hall(1989)}]{hall1989unusual}
\textsc{Hall, P.} (1989).
\newblock Unusual properties of bootstrap confidence intervals in regression
  problems.
\newblock \textit{Probability Theory and Related Fields} \textbf{81}, 247--273.

\bibitem[{Hall(1992)}]{hall1992bootstrap}
\textsc{Hall, P.} (1992).
\newblock \textit{The bootstrap and {Edgeworth} expansion}.
\newblock Springer Science \& Business Media.

\bibitem[{Hartigan(1970)}]{hartigan1970exact}
\textsc{Hartigan, J.} (1970).
\newblock Exact confidence intervals in regression problems with independent
  symmetric errors.
\newblock \textit{The Annals of Mathematical Statistics} \textbf{41},
  1992--1998.

\bibitem[{Hastie et~al.(2019)Hastie, Montanari, Rosset \&
  Tibshirani}]{hastie2019surprises}
\textsc{Hastie, T.}, \textsc{Montanari, A.}, \textsc{Rosset, S.} \&
  \textsc{Tibshirani, R.~J.} (2019).
\newblock Surprises in high-dimensional ridgeless least squares interpolation.
\newblock \textit{arXiv preprint arXiv:1903.08560} .

\bibitem[{Hastings et~al.(1947)Hastings, Mosteller, Tukey \&
  Winsor}]{hastings1947low}
\textsc{Hastings, C.}, \textsc{Mosteller, F.}, \textsc{Tukey, J.~W.} \&
  \textsc{Winsor, C.~P.} (1947).
\newblock Low moments for small samples: a comparative study of order
  statistics.
\newblock \textit{The Annals of Mathematical Statistics} \textbf{18}, 413--426.

\bibitem[{Hettmansperger \& McKean(1978)}]{hettmansperger1978statistical}
\textsc{Hettmansperger, T.~P.} \& \textsc{McKean, J.~W.} (1978).
\newblock Statistical inference based on ranks.
\newblock \textit{Psychometrika} \textbf{43}, 69--79.

\bibitem[{Hinkley(1977)}]{hinkley1977jackknifing}
\textsc{Hinkley, D.~V.} (1977).
\newblock Jackknifing in unbalanced situations.
\newblock \textit{Technometrics} \textbf{19}, 285--292.

\bibitem[{Hodges \& Lehmann(1962)}]{hodges1962rank}
\textsc{Hodges, J.~L.} \& \textsc{Lehmann, E.~L.} (1962).
\newblock Rank methods for combination of independent experiments in analysis
  of variance.
\newblock \textit{The Annals of Mathematical Statistics} \textbf{33}, 482--497.

\bibitem[{Hodges \& Lehmann(1963)}]{hodges1963estimates}
\textsc{Hodges, J.~L.} \& \textsc{Lehmann, E.~L.} (1963).
\newblock Estimates of location based on rank tests.
\newblock \textit{The Annals of Mathematical Statistics} , 598--611.

\bibitem[{Hoeffding(1952)}]{hoeffding1952large}
\textsc{Hoeffding, W.} (1952).
\newblock The large-sample power of tests based on permutations of
  observations.
\newblock \textit{The Annals of Mathematical Statistics} \textbf{23}, 169--192.

\bibitem[{Hotelling \& Pabst(1936)}]{hotelling1936rank}
\textsc{Hotelling, H.} \& \textsc{Pabst, M.~R.} (1936).
\newblock Rank correlation and tests of significance involving no assumption of
  normality.
\newblock \textit{The Annals of Mathematical Statistics} \textbf{7}, 29--43.

\bibitem[{Hu \& Kalbfleisch(2000)}]{hu2000estimating}
\textsc{Hu, F.} \& \textsc{Kalbfleisch, J.~D.} (2000).
\newblock The estimating function bootstrap.
\newblock \textit{Canadian Journal of Statistics} \textbf{28}, 449--481.

\bibitem[{Hu \& Zidek(1995)}]{hu1995bootstrap}
\textsc{Hu, F.} \& \textsc{Zidek, J.~V.} (1995).
\newblock A bootstrap based on the estimating equations of the linear model.
\newblock \textit{Biometrika} \textbf{82}, 263--275.

\bibitem[{Huber(1964)}]{huber1964robust}
\textsc{Huber, P.~J.} (1964).
\newblock Robust estimation of a location parameter.
\newblock \textit{Anii. Math} .

\bibitem[{Huber(1972)}]{huber72}
\textsc{Huber, P.~J.} (1972).
\newblock The 1972 wald lecture robust statistics: A review.
\newblock \textit{The Annals of Mathematical Statistics} , 1041--1067.

\bibitem[{Huber(1973)}]{huber1973robust}
\textsc{Huber, P.~J.} (1973).
\newblock Robust regression: asymptotics, conjectures and {Monte Carlo}.
\newblock \textit{The Annals of Statistics} \textbf{1}, 799--821.

\bibitem[{Jaeckel(1972)}]{jaeckel1972estimating}
\textsc{Jaeckel, L.~A.} (1972).
\newblock Estimating regression coefficients by minimizing the dispersion of
  the residuals.
\newblock \textit{The Annals of Mathematical Statistics} , 1449--1458.

\bibitem[{Jensen(1979)}]{jensen1979linear}
\textsc{Jensen, D.} (1979).
\newblock Linear models without moments.
\newblock \textit{Biometrika} \textbf{66}, 611--617.

\bibitem[{Jin et~al.(2001)Jin, Ying \& Wei}]{jin2001simple}
\textsc{Jin, Z.}, \textsc{Ying, Z.} \& \textsc{Wei, L.} (2001).
\newblock A simple resampling method by perturbing the minimand.
\newblock \textit{Biometrika} \textbf{88}, 381--390.

\bibitem[{Johnstone \& Velleman(1985)}]{johnstone1985resistant}
\textsc{Johnstone, I.~M.} \& \textsc{Velleman, P.~F.} (1985).
\newblock The resistant line and related regression methods.
\newblock \textit{Journal of the American Statistical Association} \textbf{80},
  1041--1054.

\bibitem[{Jung(1956)}]{jung1956linear}
\textsc{Jung, J.} (1956).
\newblock On linear estimates defined by a continuous weight function.
\newblock \textit{Arkiv f{\"o}r matematik} \textbf{3}, 199--209.

\bibitem[{Jureckova(1969)}]{jureckova1969asymptotic}
\textsc{Jureckova, J.} (1969).
\newblock Asymptotic linearity of a rank statistic in regression parameter.
\newblock \textit{The Annals of Mathematical Statistics} \textbf{40},
  1889--1900.

\bibitem[{Jureckova(1971)}]{jureckova1971nonparametric}
\textsc{Jureckova, J.} (1971).
\newblock Nonparametric estimate of regression coefficients.
\newblock \textit{The Annals of Mathematical Statistics} , 1328--1338.

\bibitem[{Jureckova(1977)}]{jureckova1977asymptotic}
\textsc{Jureckova, J.} (1977).
\newblock Asymptotic relations of $ m $-estimates and $ r $-estimates in linear
  regression model.
\newblock \textit{The Annals of Statistics} \textbf{5}, 464--472.

\bibitem[{Jureckova(1983)}]{jureckova1983winsorized}
\textsc{Jureckova, J.} (1983).
\newblock Winsorized least squares estimator and its {M}-estimator counterpart.
\newblock \textit{Contributions to Statistics: Essays in Honour of Norman L.
  Johnson} , 237--245.

\bibitem[{Jureckova(1984)}]{jurevckova1984regression}
\textsc{Jureckova, J.} (1984).
\newblock Regression quantiles and trimmed least squares estimator under a
  general design.
\newblock \textit{Kybernetika} \textbf{20}, 345--357.

\bibitem[{Kendall \& Smith(1939)}]{kendall1939problem}
\textsc{Kendall, M.~G.} \& \textsc{Smith, B.~B.} (1939).
\newblock The problem of m rankings.
\newblock \textit{Annals of mathematical statistics} .

\bibitem[{Kennedy(1995)}]{kennedy1995randomization}
\textsc{Kennedy, F.~E.} (1995).
\newblock Randomization tests in econometrics.
\newblock \textit{Journal of Business \& Economic Statistics} \textbf{13},
  85--94.

\bibitem[{Kennedy \& Cade(1996)}]{kennedy1996randomization}
\textsc{Kennedy, P.~E.} \& \textsc{Cade, B.~S.} (1996).
\newblock Randomization tests for multiple regression.
\newblock \textit{Communications in Statistics-Simulation and Computation}
  \textbf{25}, 923--936.

\bibitem[{Kildea(1981)}]{kildea1981brown}
\textsc{Kildea, D.} (1981).
\newblock Brown-mood type median estimators for simple regression models.
\newblock \textit{The Annals of Statistics} , 438--442.

\bibitem[{Kline \& Santos(2012)}]{kline2012score}
\textsc{Kline, P.} \& \textsc{Santos, A.} (2012).
\newblock A score based approach to wild bootstrap inference.
\newblock \textit{Journal of Econometric Methods} \textbf{1}, 23--41.

\bibitem[{Koenker(1997)}]{koenker19978}
\textsc{Koenker, R.} (1997).
\newblock 8 rank tests for linear models.
\newblock \textit{Handbook of statistics} \textbf{15}, 175--199.

\bibitem[{Koenker \& Bassett(1978)}]{koenker1978regression}
\textsc{Koenker, R.} \& \textsc{Bassett, G.} (1978).
\newblock Regression quantiles.
\newblock \textit{Econometrica: journal of the Econometric Society} , 33--50.

\bibitem[{Koenker \& Portnoy(1987)}]{koenker1987estimation}
\textsc{Koenker, R.} \& \textsc{Portnoy, S.} (1987).
\newblock {L}-estimation for linear models.
\newblock \textit{Journal of the American statistical Association} \textbf{82},
  851--857.

\bibitem[{Koenker \& Zhao(1994)}]{koenker1994estimatton}
\textsc{Koenker, R.} \& \textsc{Zhao, Q.} (1994).
\newblock {L}-estimatton for linear heteroscedastic models.
\newblock \textit{Journaltitle of Nonparametric Statistics} \textbf{3},
  223--235.

\bibitem[{Koul(1969)}]{koul1969asymptotic}
\textsc{Koul, H.~L.} (1969).
\newblock Asymptotic behavior of wilcoxon type confidence regions in multiple
  linear regression.
\newblock \textit{The Annals of Mathematical Statistics} \textbf{40},
  1950--1979.

\bibitem[{Koul(1970)}]{koul1970class}
\textsc{Koul, H.~L.} (1970).
\newblock A class of adf tests for subhypothesis in the multiple linear
  regression.
\newblock \textit{The Annals of Mathematical Statistics} , 1273--1281.

\bibitem[{Kraft \& Van~Eeden(1972)}]{kraft1972linearized}
\textsc{Kraft, C.~H.} \& \textsc{Van~Eeden, C.} (1972).
\newblock Linearized rank estimates and signed-rank estimates for the general
  linear hypothesis.
\newblock \textit{The Annals of Mathematical Statistics} \textbf{43}, 42--57.

\bibitem[{Kruskal \& Wallis(1952)}]{kruskal1952use}
\textsc{Kruskal, W.~H.} \& \textsc{Wallis, W.~A.} (1952).
\newblock Use of ranks in one-criterion variance analysis.
\newblock \textit{Journal of the American statistical Association} \textbf{47},
  583--621.

\bibitem[{Lahiri(1992)}]{lahiri1992bootstrapping}
\textsc{Lahiri, S.~N.} (1992).
\newblock Bootstrapping {M}-estimators of a multiple linear regression
  parameter.
\newblock \textit{The Annals of Statistics} , 1548--1570.

\bibitem[{Lancaster \& Quade(1985)}]{lancaster1985nonparametric}
\textsc{Lancaster, J.} \& \textsc{Quade, D.} (1985).
\newblock A nonparametric test for linear regression based on combining
  {Kendall}'s tau with the sign test.
\newblock \textit{Journal of the American Statistical Association} \textbf{80},
  393--397.

\bibitem[{LeCam(1953)}]{lecam1953some}
\textsc{LeCam, L.} (1953).
\newblock On some asymptotic properties of maximum likelihood estimates and
  related bayes estimates.
\newblock \textit{Univ. California Pub. Statist.} \textbf{1}, 277--330.

\bibitem[{Lehmann \& Romano(2006)}]{lehmann2006testing}
\textsc{Lehmann, E.~L.} \& \textsc{Romano, J.~P.} (2006).
\newblock \textit{Testing statistical hypotheses}.
\newblock Springer Science \& Business Media.

\bibitem[{Lei et~al.(2018)Lei, Bickel \& El~Karoui}]{lei2018asymptotics}
\textsc{Lei, L.}, \textsc{Bickel, P.~J.} \& \textsc{El~Karoui, N.} (2018).
\newblock Asymptotics for high dimensional regression {M}-estimates: fixed
  design results.
\newblock \textit{Probability Theory and Related Fields} \textbf{172},
  983--1079.

\bibitem[{Liu(1988)}]{liu1988bootstrap}
\textsc{Liu, R.~Y.} (1988).
\newblock Bootstrap procedures under some non-iid models.
\newblock \textit{The Annals of Statistics} \textbf{16}, 1696--1708.

\bibitem[{Liu \& Singh(1992)}]{liu1992efficiency}
\textsc{Liu, R.~Y.} \& \textsc{Singh, K.} (1992).
\newblock Efficiency and robustness in resampling.
\newblock \textit{The Annals of Statistics} \textbf{20}, 370--384.

\bibitem[{Lloyd(1952)}]{lloyd1952least}
\textsc{Lloyd, E.} (1952).
\newblock Least-squares estimation of location and scale parameters using order
  statistics.
\newblock \textit{Biometrika} \textbf{39}, 88--95.

\bibitem[{Mammen(1989)}]{mammen89}
\textsc{Mammen, E.} (1989).
\newblock Asymptotics with increasing dimension for robust regression with
  applications to the bootstrap.
\newblock \textit{The Annals of Statistics} , 382--400.

\bibitem[{Mammen(1993)}]{mammen1993bootstrap}
\textsc{Mammen, E.} (1993).
\newblock Bootstrap and wild bootstrap for high dimensional linear models.
\newblock \textit{The annals of statistics} \textbf{21}, 255--285.

\bibitem[{Manly(1991)}]{manly1991randomization}
\textsc{Manly, B.~F.} (1991).
\newblock \textit{Randomization, bootstrap and Monte Carlo methods in biology}.
\newblock Chapman and Hall/CRC.

\bibitem[{Mann \& Wald(1943)}]{mann1943stochastic}
\textsc{Mann, H.~B.} \& \textsc{Wald, A.} (1943).
\newblock On stochastic limit and order relationships.
\newblock \textit{The Annals of Mathematical Statistics} \textbf{14}, 217--226.

\bibitem[{Markatou \& Ronchetti(1997)}]{markatou19973}
\textsc{Markatou, M.} \& \textsc{Ronchetti, E.} (1997).
\newblock 3 robust inference: The approach based on influence functions.
\newblock \textit{Handbook of statistics} \textbf{15}, 49--75.

\bibitem[{Maxwell(1860)}]{maxwell1860v}
\textsc{Maxwell, J.~C.} (1860).
\newblock {V}. illustrations of the dynamical theory of gases. part {I}. on the
  motions and collisions of perfectly elastic spheres.
\newblock \textit{The London, Edinburgh, and Dublin Philosophical Magazine and
  Journal of Science} \textbf{19}, 19--32.

\bibitem[{McKean \& Hettmansperger(1976)}]{mckean1976tests}
\textsc{McKean, J.~W.} \& \textsc{Hettmansperger, T.~P.} (1976).
\newblock Tests of hypotheses based on ranks in the general linear model.
\newblock \textit{Communications in statistics-theory and methods} \textbf{5},
  693--709.

\bibitem[{McKean \& Hettmansperger(1978)}]{mckean1978robust}
\textsc{McKean, J.~W.} \& \textsc{Hettmansperger, T.~P.} (1978).
\newblock A robust analysis of the general linear model based on one step
  r-estimates.
\newblock \textit{Biometrika} \textbf{65}, 571--579.

\bibitem[{Mehra \& Sen(1969)}]{mehra1969class}
\textsc{Mehra, K.} \& \textsc{Sen, P.} (1969).
\newblock On a class of conditionally distribution-free tests for interactions
  in factorial experiments.
\newblock \textit{The Annals of Mathematical Statistics} \textbf{40}, 658--664.

\bibitem[{Meinshausen(2015)}]{meinshausen2015group}
\textsc{Meinshausen, N.} (2015).
\newblock Group bound: confidence intervals for groups of variables in sparse
  high dimensional regression without assumptions on the design.
\newblock \textit{Journal of the Royal Statistical Society: Series B
  (Statistical Methodology)} \textbf{77}, 923--945.

\bibitem[{Michalewicz(2013)}]{michalewicz2013genetic}
\textsc{Michalewicz, Z.} (2013).
\newblock \textit{Genetic algorithms+ data structures= evolution programs}.
\newblock Springer Science \& Business Media.

\bibitem[{Miller(1974)}]{miller1974unbalanced}
\textsc{Miller, R.~G.} (1974).
\newblock An unbalanced jackknife.
\newblock \textit{The Annals of Statistics} , 880--891.

\bibitem[{Mosteller(1946)}]{mosteller1946some}
\textsc{Mosteller, F.} (1946).
\newblock On some useful" inefficient" statistics.
\newblock \textit{The Annals of Mathematical Statistics} \textbf{17}, 377--408.

\bibitem[{Navidi(1989)}]{navidi1989edgeworth}
\textsc{Navidi, W.} (1989).
\newblock Edgeworth expansions for bootstrapping regression models.
\newblock \textit{The Annals of Statistics} \textbf{17}, 1472--1478.

\bibitem[{Neyman(1959)}]{neyman1959optimal}
\textsc{Neyman, J.} (1959).
\newblock Optimal asymptotic tests of composite hypotheses.
\newblock \textit{Probability and statsitics} , 213--234.

\bibitem[{Neyman(1923)}]{neyman1923application}
\textsc{Neyman, J.~S.} (1923).
\newblock On the application of probability theory to agricultural experiments.
  essay on principles. section 9. (translated and edited by dm dabrowska and tp
  speed, statistical science (1990), 5, 465-480).
\newblock \textit{Annals of Agricultural Sciences} \textbf{10}, 1--51.

\bibitem[{Oja(1987)}]{oja1987permutation}
\textsc{Oja, H.} (1987).
\newblock On permutation tests in multiple regression and analysis of
  covariance problems.
\newblock \textit{Australian Journal of Statistics} \textbf{29}, 91--100.

\bibitem[{Parzen et~al.(1994)Parzen, Wei \& Ying}]{parzen1994resampling}
\textsc{Parzen, M.}, \textsc{Wei, L.} \& \textsc{Ying, Z.} (1994).
\newblock A resampling method based on pivotal estimating functions.
\newblock \textit{Biometrika} \textbf{81}, 341--350.

\bibitem[{Pearson \& Please(1975)}]{pearson1975relation}
\textsc{Pearson, E.} \& \textsc{Please, N.} (1975).
\newblock Relation between the shape of population distribution and the
  robustness of four simple test statistics.
\newblock \textit{Biometrika} \textbf{62}, 223--241.

\bibitem[{Pearson(1929)}]{pearson1929some}
\textsc{Pearson, E.~S.} (1929).
\newblock Some notes on sampling tests with two variables.
\newblock \textit{Biometrika} , 337--360.

\bibitem[{Pearson(1931)}]{pearson1931analysis}
\textsc{Pearson, E.~S.} (1931).
\newblock The analysis of variance in cases of non-normal variation.
\newblock \textit{Biometrika} , 114--133.

\bibitem[{Pearson \& Adyanth{\=a}ya(1929)}]{pearson1929distribution}
\textsc{Pearson, E.~S.} \& \textsc{Adyanth{\=a}ya, N.} (1929).
\newblock The distribution of frequency constants in small samples from
  non-normal symmetrical and skew populations.
\newblock \textit{Biometrika} \textbf{21}, 259--286.

\bibitem[{Peddada \& Patwardhan(1992)}]{peddada1992jackknife}
\textsc{Peddada, S.~D.} \& \textsc{Patwardhan, G.} (1992).
\newblock Jackknife variance estimators in linear models.
\newblock \textit{Biometrika} \textbf{79}, 654--657.

\bibitem[{Pinelis(1994)}]{pinelis1994extremal}
\textsc{Pinelis, I.} (1994).
\newblock Extremal probabilistic problems and hotelling's $t^{2}$ test under a
  symmetry condition.
\newblock \textit{The Annals of Statistics} \textbf{22}, 357--368.

\bibitem[{Pitman(1937{\natexlab{a}})}]{pitman1937significance1}
\textsc{Pitman, E. J.~G.} (1937{\natexlab{a}}).
\newblock Significance tests which may be applied to samples from any
  populations.
\newblock \textit{Supplement to the Journal of the Royal Statistical Society}
  \textbf{4}, 119--130.

\bibitem[{Pitman(1937{\natexlab{b}})}]{pitman1937significance}
\textsc{Pitman, E. J.~G.} (1937{\natexlab{b}}).
\newblock Significance tests which may be applied to samples from any
  populations. {II}. the correlation coefficient test.
\newblock \textit{Supplement to the Journal of the Royal Statistical Society}
  \textbf{4}, 225--232.

\bibitem[{Pitman(1938)}]{pitman1938significance}
\textsc{Pitman, E. J.~G.} (1938).
\newblock Significance tests which may be applied to samples from any
  populations: {III}. the analysis of variance test.
\newblock \textit{Biometrika} \textbf{29}, 322--335.

\bibitem[{Pollard(1991)}]{pollard1991asymptotics}
\textsc{Pollard, D.} (1991).
\newblock Asymptotics for least absolute deviation regression estimators.
\newblock \textit{Econometric Theory} \textbf{7}, 186--199.

\bibitem[{Portnoy(1985)}]{portnoy85}
\textsc{Portnoy, S.} (1985).
\newblock Asymptotic behavior of {M} estimators of $p$ regression parameters
  when $p^{2} / n$ is large; {II}. {Normal} approximation.
\newblock \textit{The Annals of Statistics} , 1403--1417.

\bibitem[{Portnoy \& Koenker(1989)}]{portnoy1989adaptive}
\textsc{Portnoy, S.} \& \textsc{Koenker, R.} (1989).
\newblock Adaptive $ l $-estimation for linear models.
\newblock \textit{The Annals of Statistics} \textbf{17}, 362--381.

\bibitem[{Puri \& Sen(1973)}]{puri1973note}
\textsc{Puri, M.~L.} \& \textsc{Sen, P.} (1973).
\newblock A note on asymptotically distribution free tests for subhypotheses in
  multiple linear regression.
\newblock \textit{The Annals of Statistics} \textbf{1}, 553--556.

\bibitem[{Puri \& Sen(1966)}]{puri1966class}
\textsc{Puri, M.~L.} \& \textsc{Sen, P.~K.} (1966).
\newblock On a class of multivariate multisample rank-order tests.
\newblock \textit{Sankhy{\=a}: The Indian Journal of Statistics, Series A} ,
  353--376.

\bibitem[{Quade(1979)}]{quade1979regression}
\textsc{Quade, D.} (1979).
\newblock Regression analysis based on the signs of the residuals.
\newblock \textit{Journal of the American Statistical Association} \textbf{74},
  411--417.

\bibitem[{Quenouille(1949)}]{quenouille1949problems}
\textsc{Quenouille, M.~H.} (1949).
\newblock Problems in plane sampling.
\newblock \textit{The Annals of Mathematical Statistics} \textbf{20}, 355--375.

\bibitem[{Quenouille(1956)}]{quenouille1956notes}
\textsc{Quenouille, M.~H.} (1956).
\newblock Notes on bias in estimation.
\newblock \textit{Biometrika} \textbf{43}, 353--360.

\bibitem[{Qumsiyeh(1994)}]{qumsiyeh1994bootstrapping}
\textsc{Qumsiyeh, M.~B.} (1994).
\newblock Bootstrapping and empirical {Edgeworth} expansions in multiple linear
  regression models.
\newblock \textit{Communications in Statistics-Theory and Methods} \textbf{23},
  3227--3239.

\bibitem[{Rao \& Zhao(1992)}]{rao1992approximation}
\textsc{Rao, C.~R.} \& \textsc{Zhao, L.} (1992).
\newblock Approximation to the distribution of {M}-estimates in linear models
  by randomly weighted bootstrap.
\newblock \textit{Sankhy{\=a}: The Indian Journal of Statistics, Series A} ,
  323--331.

\bibitem[{Relles(1968)}]{relles68}
\textsc{Relles, D.~A.} (1968).
\newblock Robust regression by modified least-squares.
\newblock Tech. rep., DTIC Document.

\bibitem[{Romano(1989)}]{romano1989bootstrap}
\textsc{Romano, J.~P.} (1989).
\newblock Bootstrap and randomization tests of some nonparametric hypotheses.
\newblock \textit{The Annals of Statistics} , 141--159.

\bibitem[{Romano(1990)}]{romano1990behavior}
\textsc{Romano, J.~P.} (1990).
\newblock On the behavior of randomization tests without a group invariance
  assumption.
\newblock \textit{Journal of the American Statistical Association} \textbf{85},
  686--692.

\bibitem[{Rousseeuw(1984)}]{rousseeuw1984least}
\textsc{Rousseeuw, P.~J.} (1984).
\newblock Least median of squares regression.
\newblock \textit{Journal of the American statistical association} \textbf{79},
  871--880.

\bibitem[{Rousseeuw \& Hubert(1999)}]{rousseeuw1999regression}
\textsc{Rousseeuw, P.~J.} \& \textsc{Hubert, M.} (1999).
\newblock Regression depth.
\newblock \textit{Journal of the American Statistical Association} \textbf{94},
  388--402.

\bibitem[{Rubin(1974)}]{rubin1974estimating}
\textsc{Rubin, D.~B.} (1974).
\newblock Estimating causal effects of treatments in randomized and
  nonrandomized studies.
\newblock \textit{Journal of educational Psychology} \textbf{66}, 688.

\bibitem[{Rubin(1981)}]{rubin1981bayesian}
\textsc{Rubin, D.~B.} (1981).
\newblock The bayesian bootstrap.
\newblock \textit{The annals of statistics} , 130--134.

\bibitem[{Ruppert \& Carroll(1980)}]{ruppert1980trimmed}
\textsc{Ruppert, D.} \& \textsc{Carroll, R.~J.} (1980).
\newblock Trimmed least squares estimation in the linear model.
\newblock \textit{Journal of the American Statistical Association} \textbf{75},
  828--838.

\bibitem[{S{\"a}rndal et~al.(1978)S{\"a}rndal, Thomsen, Hoem, Lindley,
  Barndorff-Nielsen \& Dalenius}]{sarndal1978design}
\textsc{S{\"a}rndal, C.-E.}, \textsc{Thomsen, I.}, \textsc{Hoem, J.~M.},
  \textsc{Lindley, D.}, \textsc{Barndorff-Nielsen, O.} \& \textsc{Dalenius, T.}
  (1978).
\newblock Design-based and model-based inference in survey sampling [with
  discussion and reply].
\newblock \textit{Scandinavian Journal of Statistics} , 27--52.

\bibitem[{Schrader \& Hettmansperger(1980)}]{schrader1980robust}
\textsc{Schrader, R.~M.} \& \textsc{Hettmansperger, T.~P.} (1980).
\newblock Robust analysis of variance based upon a likelihood ratio criterion.
\newblock \textit{Biometrika} \textbf{67}, 93--101.

\bibitem[{Sen(1968{\natexlab{a}})}]{sen1968estimates}
\textsc{Sen, P.~K.} (1968{\natexlab{a}}).
\newblock Estimates of the regression coefficient based on {Kendall}'s tau.
\newblock \textit{Journal of the American statistical association} \textbf{63},
  1379--1389.

\bibitem[{Sen(1968{\natexlab{b}})}]{sen1968class}
\textsc{Sen, P.~K.} (1968{\natexlab{b}}).
\newblock On a class of aligned rank order tests in two-way layouts.
\newblock \textit{The Annals of Mathematical Statistics} \textbf{39},
  1115--1124.

\bibitem[{Sen(1969)}]{sen1969class}
\textsc{Sen, P.~K.} (1969).
\newblock On a class of rank order tests for the parallelism of several
  regression lines.
\newblock \textit{The Annals of Mathematical Statistics} , 1668--1683.

\bibitem[{Sen(1982)}]{sen1982m}
\textsc{Sen, P.~K.} (1982).
\newblock On {M} test in linear models.
\newblock \textit{Biometrika} , 245--248.

\bibitem[{Shao(1988)}]{shao1988resampling}
\textsc{Shao, J.} (1988).
\newblock On resampling methods for variance and bias estimation in linear
  models.
\newblock \textit{The Annals of Statistics} , 986--1008.

\bibitem[{Shao(1989)}]{shao1989jackknifing}
\textsc{Shao, J.} (1989).
\newblock Jackknifing weighted least squares estimators.
\newblock \textit{Journal of the Royal Statistical Society: Series B
  (Methodological)} \textbf{51}, 139--156.

\bibitem[{Shao \& Wu(1987)}]{shao1987heteroscedasticity}
\textsc{Shao, J.} \& \textsc{Wu, C.} (1987).
\newblock Heteroscedasticity-robustness of jackknife variance estimators in
  linear models.
\newblock \textit{The Annals of Statistics} , 1563--1579.

\bibitem[{Shorack(1982)}]{shorack1982bootstrapping}
\textsc{Shorack, G.~R.} (1982).
\newblock Bootstrapping robust regression.
\newblock \textit{Communications in Statistics-Theory and Methods} \textbf{11},
  961--972.

\bibitem[{Siegel(1982)}]{siegel1982robust}
\textsc{Siegel, A.~F.} (1982).
\newblock Robust regression using repeated medians.
\newblock \textit{Biometrika} \textbf{69}, 242--244.

\bibitem[{Sievers(1978)}]{sievers1978weighted}
\textsc{Sievers, G.~L.} (1978).
\newblock Weighted rank statistics for simple linear regression.
\newblock \textit{Journal of the American Statistical Association} \textbf{73},
  628--631.

\bibitem[{Sievers(1983)}]{sievers1983weighted}
\textsc{Sievers, G.~L.} (1983).
\newblock A weighted dispersion function for estimation in linear models.
\newblock \textit{Communications in Statistics-Theory and Methods} \textbf{12},
  1161--1179.

\bibitem[{Silvapulle(1992)}]{silvapulle1992robust}
\textsc{Silvapulle, M.~J.} (1992).
\newblock Robust tests of inequality constraints and one-sided hypotheses in
  the linear model.
\newblock \textit{Biometrika} \textbf{79}, 621--630.

\bibitem[{Singer \& Sen(1985)}]{singer1985m}
\textsc{Singer, J.~M.} \& \textsc{Sen, P.~K.} (1985).
\newblock {M}-methods in multivariate linear models.
\newblock \textit{Journal of multivariate Analysis} \textbf{17}, 168--184.

\bibitem[{Snedecor(1934)}]{snedecor1934calculation}
\textsc{Snedecor, G.~W.} (1934).
\newblock \textit{Calculation and interpretation of analysis of varianceand
  covariance}.
\newblock Collegiate Press, Inc,; Ames Iowa.

\bibitem[{Srivastava(1972)}]{srivastava1972asymptotically}
\textsc{Srivastava, M.} (1972).
\newblock Asymptotically most powerful rank tests for regression parameters in
  {MANOVA}.
\newblock \textit{Annals of the Institute of Statistical Mathematics}
  \textbf{24}, 285--297.

\bibitem[{Student(1908{\natexlab{a}})}]{student1908probable2}
\textsc{Student} (1908{\natexlab{a}}).
\newblock Probable error of a correlation coefficient.
\newblock \textit{Biometrika} , 302--310.

\bibitem[{Student(1908{\natexlab{b}})}]{student1908probable}
\textsc{Student} (1908{\natexlab{b}}).
\newblock The probable error of a mean.
\newblock \textit{Biometrika} , 1--25.

\bibitem[{Ter~Braak(1992)}]{ter1992permutation}
\textsc{Ter~Braak, C.~J.} (1992).
\newblock Permutation versus bootstrap significance tests in multiple
  regression and {ANOVA}.
\newblock In \textit{Bootstrapping and related techniques}. Springer, pp.
  79--85.

\bibitem[{Theil(1950{\natexlab{a}})}]{theil1950rank}
\textsc{Theil, H.} (1950{\natexlab{a}}).
\newblock A rank-invariant method of linear and polynomial regression analysis,
  {I}.
\newblock In \textit{Nederl. Akad. Wetensch. Proc}, vol.~53.

\bibitem[{Theil(1950{\natexlab{b}})}]{theil1950rank2}
\textsc{Theil, H.} (1950{\natexlab{b}}).
\newblock A rank-invariant method of linear and polynomial regression analysis,
  {II}.
\newblock In \textit{Nederl. Akad. Wetensch. Proc}, vol.~53.

\bibitem[{Theil(1950{\natexlab{c}})}]{theil1950rank3}
\textsc{Theil, H.} (1950{\natexlab{c}}).
\newblock A rank-invariant method of linear and polynomial regression analysis,
  {III}.
\newblock In \textit{Nederl. Akad. Wetensch. Proc}, vol.~53.

\bibitem[{Toulis(2019)}]{toulis2019life}
\textsc{Toulis, P.} (2019).
\newblock Life after bootstrap: Residual randomization inference in regression
  models.
\newblock \textit{arXiv preprint arXiv:1908.04218} .

\bibitem[{Tukey(1958)}]{tukey1958bias}
\textsc{Tukey, J.} (1958).
\newblock Bias and confidence in not quite large samples.
\newblock \textit{Ann. Math. Statist.} \textbf{29}, 614.

\bibitem[{Tukey(1960)}]{tukey1960survey}
\textsc{Tukey, J.~W.} (1960).
\newblock A survey of sampling from contaminated distributions.
\newblock \textit{Contributions to probability and statistics} , 448--485.

\bibitem[{Tukey(1962)}]{tukey1962future}
\textsc{Tukey, J.~W.} (1962).
\newblock The future of data analysis.
\newblock \textit{The annals of mathematical statistics} \textbf{33}, 1--67.

\bibitem[{Van~Aelst et~al.(2002)Van~Aelst, Rousseeuw, Hubert \&
  Struyf}]{van2002deepest}
\textsc{Van~Aelst, S.}, \textsc{Rousseeuw, P.~J.}, \textsc{Hubert, M.} \&
  \textsc{Struyf, A.} (2002).
\newblock The deepest regression method.
\newblock \textit{Journal of Multivariate Analysis} \textbf{81}, 138--166.

\bibitem[{van Eeden(1972)}]{van1972analogue}
\textsc{van Eeden, C.} (1972).
\newblock An analogue, for signed rank statistics, of jureckova's asymptotic
  linearity theorem for rank statistics.
\newblock \textit{The Annals of Mathematical Statistics} \textbf{43}, 791--802.

\bibitem[{Wald(1949)}]{wald1949note}
\textsc{Wald, A.} (1949).
\newblock Note on the consistency of the maximum likelihood estimate.
\newblock \textit{The Annals of Mathematical Statistics} \textbf{20}, 595--601.

\bibitem[{Wallace(1958)}]{wallace1958asymptotic}
\textsc{Wallace, D.~L.} (1958).
\newblock Asymptotic approximations to distributions.
\newblock \textit{The Annals of Mathematical Statistics} \textbf{29}, 635--654.

\bibitem[{Wang \& Akritas(2004)}]{wang2004rank}
\textsc{Wang, H.} \& \textsc{Akritas, M.~G.} (2004).
\newblock Rank tests for {ANOVA} with large number of factor levels.
\newblock \textit{Journal of Nonparametric Statistics} \textbf{16}, 563--589.

\bibitem[{Welch(1937)}]{welch1937z}
\textsc{Welch, B.~L.} (1937).
\newblock On the z-test in randomized blocks and latin squares.
\newblock \textit{Biometrika} \textbf{29}, 21--52.

\bibitem[{Welch(1990)}]{welch1990construction}
\textsc{Welch, W.~J.} (1990).
\newblock Construction of permutation tests.
\newblock \textit{Journal of the American Statistical Association} \textbf{85},
  693--698.

\bibitem[{Welsh(1987)}]{welsh1987one}
\textsc{Welsh, A.} (1987).
\newblock One-step {L}-estimators for the linear model.
\newblock \textit{The Annals of Statistics} \textbf{15}, 626--641.

\bibitem[{Welsh(1989)}]{welsh1989m}
\textsc{Welsh, A.} (1989).
\newblock On {M}-processes and {M}-estimation.
\newblock \textit{The Annals of Statistics} \textbf{17}, 337--361.

\bibitem[{Welsh(1991)}]{welsh1991asymptotically}
\textsc{Welsh, A.} (1991).
\newblock Asymptotically efficient adaptive {L}-estimators in linear models.
\newblock \textit{Statistica Sinica} , 203--228.

\bibitem[{Wilks(1938)}]{wilks1938large}
\textsc{Wilks, S.~S.} (1938).
\newblock The large-sample distribution of the likelihood ratio for testing
  composite hypotheses.
\newblock \textit{The Annals of Mathematical Statistics} \textbf{9}, 60--62.

\bibitem[{Wu(1990)}]{wu1990asymptotic}
\textsc{Wu, C.~F.} (1990).
\newblock On the asymptotic properties of the jackknife histogram.
\newblock \textit{The Annals of Statistics} , 1438--1452.

\bibitem[{Wu(1986)}]{wu1986jackknife}
\textsc{Wu, C.-F.~J.} (1986).
\newblock Jackknife, bootstrap and other resampling methods in regression
  analysis.
\newblock \textit{the Annals of Statistics} \textbf{14}, 1261--1295.

\bibitem[{Yohai(1972)}]{yohai72}
\textsc{Yohai, V.~J.} (1972).
\newblock \textit{Robust {M} estimates for the general linear model}.
\newblock Universidad Nacional de la Plata. Departamento de Matematica.

\bibitem[{Yohai \& Maronna(1979)}]{yohai1979asymptotic}
\textsc{Yohai, V.~J.} \& \textsc{Maronna, R.~A.} (1979).
\newblock Asymptotic behavior of {M}-estimators for the linear model.
\newblock \textit{The Annals of Statistics} , 258--268.

\bibitem[{Zellner(1976)}]{zellner1976bayesian}
\textsc{Zellner, A.} (1976).
\newblock Bayesian and non-bayesian analysis of the regression model with
  multivariate student-t error terms.
\newblock \textit{Journal of the American Statistical Association} \textbf{71},
  400--405.

\bibitem[{Zhong \& Chen(2011)}]{zhong2011tests}
\textsc{Zhong, P.-S.} \& \textsc{Chen, S.~X.} (2011).
\newblock Tests for high-dimensional regression coefficients with factorial
  designs.
\newblock \textit{Journal of the American Statistical Association}
  \textbf{106}, 260--274.

\end{thebibliography}
\bibliographystyle{biometrika}

\end{document}